\documentclass[english]{article}
\usepackage{lmodern}

\usepackage[T1]{fontenc}
\usepackage[latin9]{inputenc}
\usepackage{geometry}
\usepackage{color}
\usepackage{babel}
\usepackage{verbatim}
\usepackage{mathrsfs}
\usepackage{caption}
\usepackage{subcaption}
\usepackage{rotating}
\usepackage{float}
\usepackage{units}
\usepackage{multirow}
\usepackage{amsmath}
\usepackage{amsthm}
\usepackage{amssymb}
\usepackage{graphicx}
\usepackage[shortlabels]{enumitem}
\usepackage{amsfonts}
\usepackage{epsfig}
\usepackage{bm}
\usepackage{natbib}
\usepackage[unicode=true,pdfusetitle,
 bookmarks=true,bookmarksnumbered=false,bookmarksopen=false,
 breaklinks=false,pdfborder={0 0 1},backref=page,colorlinks=true]
 {hyperref}
\usepackage{breakurl}

\makeatletter

\floatstyle{ruled}
\newfloat{algorithm}{tbp}{loa}
\providecommand{\algorithmname}{Algorithm}
\floatname{algorithm}{\protect\algorithmname}

 \theoremstyle{definition}

  \theoremstyle{plain}
  \newtheorem{lemma}{\protect\lemmaname}
\theoremstyle{plain}
\newtheorem{theorem}{\protect\theoremname}
\theoremstyle{plain}
\newtheorem{assumption}{\protect\assumptionname}
\theoremstyle{plain}
\theoremstyle{plain}

\theoremstyle{plain}
\theoremstyle{plain}

\newtheorem{corollary}{\protect\corollaryname}
\theoremstyle{definition}
\newtheorem{remark}{\protect\remarkname}

\usepackage{dsfont}
\hypersetup{linkcolor=RoyalBlue,citecolor=RoyalBlue,urlcolor=RoyalBlue}
\usepackage[dvipsnames,svgnames,x11names,hyperref]{xcolor}
\makeatother

\providecommand{\defnname}{Definition}
\providecommand{\examplename}{Example}
\providecommand{\lemmaname}{Lemma}
\providecommand{\theoremname}{Theorem}
\providecommand{\assumptionname}{Assumption}
\providecommand{\propertyname}{Property}
\providecommand{\propositionname}{Proposition}
\providecommand{\corollaryname}{Corollary}
\providecommand{\remarkname}{Remark}

\begin{document}

\title{On Unbiased Score Estimation for Partially Observed Diffusions}

\author{Jeremy Heng\thanks{ESSEC Business School; heng@essec.edu},  
Jeremie Houssineau\thanks{University of Warwick; jeremie.houssineau@warwick.ac.uk}, 
and Ajay Jasra\thanks{Computer, Electrical and Mathematical Sciences and Engineering Division, King Abdullah University of Science and Technology; ajay.jasra@kaust.edu.sa}}
\date{}

\maketitle
\begin{abstract}
We consider the problem of statistical inference for a class of partially-observed diffusion processes, 
with discretely-observed data and finite-dimensional parameters.   
We construct unbiased estimators of the score function, i.e.\ the gradient of the log-likelihood function with respect to 
parameters, with no time-discretization bias. 
These estimators can be straightforwardly employed within stochastic gradient methods to perform maximum likelihood 
estimation or Bayesian inference. 
As our proposed methodology only requires access to a time-discretization scheme such as the Euler--Maruyama method, 
it is applicable to a wide class of diffusion processes and observation models.  
Our approach is based on a representation of the score as a smoothing expectation using Girsanov theorem, and
a novel adaptation of the randomization schemes developed in \citet{mcleish2011,rhee2015unbiased,jacob2020smoothing}.
This allows one to remove the time-discretization bias and burn-in bias when computing smoothing expectations using the conditional particle filter of \citet{andrieu2010particle}. 
Central to our approach is the development of new couplings of multiple conditional particle filters. 
We prove under assumptions that our estimators are unbiased and have finite variance. 
The methodology is illustrated on several challenging applications from population ecology and neuroscience.  
\end{abstract}
\textbf{\small{}Keywords}{\small{}: 
Partially-observed diffusions, unbiased estimation, particle filters, coupling, stochastic gradient methods.}{\small \par}

\section{Introduction\label{sec:Introduction}}
We consider a diffusion process $(X_t)_{t\geq0}$ in $\mathbb{R}^d$, defined as the solution of 
the stochastic differential equation (SDE) 
\begin{align}\label{eq:intro_diff}
	dX_t = a_{\theta}(X_t)dt + \sigma(X_t)dW_t,\quad X_0=x_{\star}\in\mathbb{R}^d,
\end{align}
where $(W_t)_{t\geq 0}$ is a standard Brownian motion in $\mathbb{R}^d$. 
We will assume that the drift function $a:\Theta\times\mathbb{R}^d \rightarrow\mathbb{R}^d$ 
depends on a vector of unknown parameters $\theta\in\Theta\subseteq\mathbb{R}^{d_{\theta}}$, 
but the symmetric diffusion coefficient $\sigma:\mathbb{R}^d \rightarrow\mathbb{R}^{d\times d}$ does not. 
The general case where the diffusion coefficient is parameter dependent requires a different 
treatment; we will discuss how to adapt our proposed methodology in Section~\ref{sec:parameter_diffusion}. 
The drift and diffusion coefficients are assumed to be regular enough for \eqref{eq:intro_diff} 
to admit a (weakly) unique solution for all $t>0$; more precise regularity conditions needed in our 
analysis will be stated. 
We model observations $Y_{t_1},\dots,Y_{t_P}\in\mathbb{R}^{d_y}$ at a collection of time instances 
$0 \leq t_1< \cdots < t_P = T$ as conditionally independent given the latent diffusion process 
$X=(X_t)_{0\leq t\leq T}$, with observation density $g:\Theta\times\mathbb{R}^d\times\mathbb{R}^{d_y}\rightarrow\mathbb{R}^+$. 
For notational ease, we assume that observations are given at unit times, i.e.\ $t_p = p$ for $p\in\{1,\ldots,P\}$ and $P=T$, 
which covers the case of regularly observed data by a time rescaling. 
Irregularly observed data can also be accommodated with minor modifications to our presentation 
and considered for an application in Section~\ref{sec:logistic_diffusion}. 
We note that this \emph{partially observed} setting is distinct from the \emph{discretely observed} case 
where the diffusion process is observed without error \citep{sorensen2004parametric}. 
Extension of our methodology to continuous-time observation models is straightforward and illustrated 
on an application in Section~\ref{sec:neural_network}.

Given a realization $y_{1:T}=(y_{t})_{t=1}^T$ of the observation process, the marginal likelihood of our state 
space model is
\begin{align}\label{eqn:marginal_likelihood}
	p_{\theta}(y_{1:T}) = \mathbb{E}_{\theta}\left[\prod_{t=1}^T g_{\theta}(y_{t}|X_{t})\right],
\end{align}
where $\mathbb{E}_{\theta}$ denotes expectation w.r.t.\ the probability measure $\mathbb{P}_{\theta}$, 
induced by the solution of \eqref{eq:intro_diff} on $\mathcal{C}_d([0,T])$, the space of continuous mappings from $[0,T]$ to $\mathbb{R}^d$. 
To perform parameter inference using gradient-based algorithms, we shall consider approximations of 
the score function $S(\theta)=\nabla_{\theta}\log p_{\theta}(y_{1:T})$, where $\nabla_{\theta}$ denotes the gradient 
w.r.t.\ the parameter $\theta$. 
If estimators $\widehat{S}(\theta)$ of $S(\theta)$ can be constructed, 
one can approximate the maximum likelihood estimator (MLE) $\hat{\theta}\in \arg\max_{\theta\in\Theta} p_{\theta}(y_{1:T})$ 
with stochastic gradient ascent (SGA) 
\begin{align}\label{eqn:SGA}
	\theta_m = \theta_{m-1} + \varepsilon_m\widehat{S}(\theta_{m-1}), \quad m=1,2,\ldots,
\end{align}
where $(\varepsilon_m)_{m=1}^{\infty}$ is a sequence of learning rates. Similarly, in the Bayesian framework, one can 
sample from the posterior distribution $p(d\theta|y_{1:T})\propto p(\theta)p_{\theta}(y_{1:T})d\theta$ using 
the stochastic gradient Langevin dynamics (SGLD) \citep{welling2011bayesian}
\begin{align}\label{eqn:SGLD}
	\theta_m = \theta_{m-1} + \frac{1}{2}\varepsilon_m\left(\nabla_{\theta}\log p(\theta_{m-1}) + 
	\widehat{S}(\theta_{m-1})\right) + \varepsilon_m^{1/2}\eta_m, \quad m=1,2,\ldots,
\end{align}
where $(\eta_m)_{m=1}^{\infty}$ is a sequence of independent standard Gaussian random vectors in $\mathbb{R}^{d_{\theta}}$. 
In both stochastic gradient methods, constraints on parameters can be dealt with using a suitable transformation. 
Under mild regularity conditions, we will see that Girsanov theorem can be applied to represent the score as a smoothing expectation 
\begin{align}\label{eqn:score_function}
	S(\theta) = 
	\mathbb{E}_{\theta}\left[ G_{\theta}(X) ~|~ y_{1:T} \right],
\end{align}
for some functional $G_{\theta}:\mathcal{C}_d([0,T])\rightarrow\mathbb{R}^{d_{\theta}}$. 
For most problems of practical interest, statistical inference is a challenging task due to the intractability 
of $p_{\theta}(y_{1:T})$ and $S(\theta)$. 
%\eqref{eqn:marginal_likelihood} and \eqref{eqn:score_function}. 
The latter is typically the case for two reasons. 
Firstly, as most diffusion processes do not have analytically tractable transition densities, 
one has to resort to time-discretization. 
Secondly, even if transition densities are available or a time-discretization scheme is employed, 
the expectation over the latent process is usually intractable, and one has to rely on Monte Carlo approximations. 
The following describes some existing approaches to these two issues, and how we will resolve them in this article. 

To perform inference and state estimation without any time-discretization bias, 
\citet{beskos2006exact,beskos2009monte,fearnhead2008particle} adopted %exploited 
the methodology developed in \citet{beskos2005exact,beskos2006retrospective} that allows exact 
simulation of diffusion sample paths; see also \citet{pollock2016exact,blanchet2017exact,fearnhead2017continious} 
for recent advances. Although these simulation techniques are elegant, they require a user to know 
various properties of the diffusion process, which may not always be the case in practice. 
An alternative approach to deal with the time-discretization bias, which we shall adopt, 
is to employ the debiasing schemes proposed in 
\citet{mcleish2011,rhee2015unbiased}. 
These methods allow one to unbiasedly estimate an expectation of a functional w.r.t.\ an infinite-dimensional law 
by randomizing over the level of the time-discretization in a type of multilevel Monte Carlo (MLMC) approach. 
The use of unbiased estimators of the score function $S(\theta)$ within stochastic gradient methods 
is particularly appealing. Under suitable assumptions and an appropriate choice of learning rates $(\varepsilon_m)_{m=1}^{\infty}$, 
unbiased scores ensures convergence of the SGA iterates \eqref{eqn:SGA} to a local maxima of the likelihood $p_{\theta}(y_{1:T})$ 
(see e.g.\ \citet{kushner2003stochastic}), and convergence of the SGLD and its ergodic averages 
to the posterior distribution and posterior expectations, respectively \citep{teh2016consistency}. 
In contrast, the use of biased gradient estimates 
leads to asymptotic biases that have to be studied \citep{tadic2017asymptotic}. 
Moreover, these unbiased estimators can also attain better convergence rates 
than Monte Carlo approaches based on the finest discretization level, 
which can be made the same as MLMC under simple modifications \citep{vihola2018unbiased}. 

When employing such debiasing methods, we shall assume access to unbiased estimates 
of differences of the score function associated to successive levels of discretization, i.e.\ $S_l(\theta) - S_{l-1}(\theta)$ 
where $S_l(\theta)$ denotes an approximation of $S(\theta)$ at discretization level $l\in\mathbb{N}$.
As $S_l(\theta)$ can be  written as an expectation w.r.t.\ a finite-dimensional smoothing distribution $\pi_{\theta}^l$,  
biased but consistent approximations can be obtained using various particle smoothing algorithms 
\citep{briers2010smoothing,fearnhead2010sequential,douc2011sequential} in the limit of the Monte Carlo sample size. 
We will focus on the conditional particle filter (CPF) of \citet{andrieu2010particle} which constructs 
a Markov chain Monte Carlo (MCMC) algorithm targeting $\pi_{\theta}^l$ and produces biased 
but consistent approximations of $S_l(\theta)$ in the limit of the number of MCMC iterations. 
By employing recent advances in \citet{jacob2020smoothing} 
that builds on earlier work by \citet{glynn2014exact}, one can remove the MCMC burn-in bias and 
estimate $S_l(\theta)$ unbiasedly by simulating a pair of CPF chains 
that are coupled in such a way that they meet exactly after some random number of iterations. 
For the debiasing schemes of \citet{mcleish2011,rhee2015unbiased} to return estimators of $S(\theta)$ 
with finite variance, the variance of increments $S_l(\theta) - S_{l-1}(\theta)$ has to decrease sufficiently fast to zero with 
the discretization level $l$. 
Therefore, the scores in each increment have to be estimated in a dependent manner, 
using appropriate couplings of pairs of CPF chains between successive discretization levels. 
The construction of such couplings is our main methodological contribution. 
This involves couplings of time-discretizations of the diffusion process 
and couplings of resampling steps that are necessary to prevent weight degeneracy within the CPFs. 
While the former can be achieved using common Brownian increments which is standard in MLMC 
\citep{giles2008multilevel}, the latter requires novel schemes to induce adequate dependencies 
between multiple CPF chains as simple strategies based on common random variables will be inadequate.
%We will present several coupled resampling schemes of interest and consider one of these in our analysis. 
Our main theoretical result is to establish, under assumptions, that the proposed methodology provides 
unbiased estimators of the score function $S(\theta)$ with finite variance. 

This article is structured as follows. 
We introduce a representation of the score function as a smoothing expectation 
in Section \ref{sec:score_cts} and its discrete-time approximation in Section \ref{sec:discretized_score}. 
Our proposed methodology to unbiasedly estimate the score function is presented 
in Section \ref{sec:score_estimation}. We establish various properties of our score estimators 
under assumptions in Section \ref{sec:theory} and illustrate various aspects of our methodology 
on three applications in Section \ref{sec:examples}. 
The proof of our results are detailed in the appendix. 
An R package to reproduce all numerical results can be found at 
\url{https://github.com/jeremyhengjm/UnbiasedScore}.

\section{Score functions}

\subsection{Continuous-time representation}\label{sec:score_cts}
We begin by seeking a representation of the score function of the form in \eqref{eqn:score_function}. 
For notational simplicity, we define $\Sigma(x)=\sigma(x)\sigma(x)^*$ and 
$b_\theta(x) = \Sigma(x)^{-1}\sigma(x)^*a_{\theta}(x)$ for $x\in\mathbb{R}^d$, where $\sigma(x)^*$ 
denotes the transpose of $\sigma(x)$. 
As the expectation in \eqref{eqn:marginal_likelihood} is w.r.t.\ $\mathbb{P}_{\theta}$, which 
depends on the parameter $\theta$, we consider a change of measure to the law $\mathbb{Q}$ induced by 
$dX_t = \sigma(X_t)dW_t$ with $X_0=x_{\star}$ on the space $\mathcal{C}_d([0,T])$. 
Since $\mathbb{P}_{\theta}$ and $\mathbb{Q}$ are equivalent, by Girsanov theorem 
\begin{align}\label{eqn:likelihood_girsanov}
	p_{\theta}(y_{1:T}) = \mathbb{E}_{\mathbb{Q}}\left[\frac{d\mathbb{P}_{\theta}}{d\mathbb{Q}}(X)\prod_{t=1}^T g_{\theta}(y_{t}|X_{t})\right]
\end{align}
and the corresponding Radon--Nikodym derivative is 
\begin{align}
	\frac{d\mathbb{P}_{\theta}}{d\mathbb{Q}}(X) = 
	\exp\left\lbrace -\frac{1}{2}\int_{0}^T \|b_{\theta}(X_t)\|_2^2dt + 
	\int_{0}^{T}b_{\theta}(X_t)^*\Sigma(X_t)^{-1}\sigma(X_t)^* dX_t\right\rbrace,
\end{align}
where $\|x\|_p$ denotes the $\mathbb{L}_p$-norm of a vector $x\in\mathbb{R}^d$. 
We will assume throughout the article that $\theta\mapsto a_{\theta}(x)$ and $\theta\mapsto g_{\theta}(y|x)$ are 
differentiable w.r.t.\ $\theta$ for each $(x,y)\in\mathbb{R}^d\times\mathbb{R}^{d_y}$.
Under mild regularity conditions, we may differentiate \eqref{eqn:likelihood_girsanov} and represent  
the score function as a smoothing expectation \eqref{eqn:score_function} with the functional 
\begin{align}\label{eqn:smoothing_functional}
	G_{\theta}(X) = \nabla_{\theta} \log \frac{d\mathbb{P}_{\theta}}{d\mathbb{Q}}(X) + 
	\sum_{t=1}^T\nabla_{\theta} \log g_{\theta}(y_{t}|X_{t}). 
\end{align}
The first gradient term can be written as 
\begin{align}
	\nabla_{\theta} \log \frac{d\mathbb{P}_{\theta}}{d\mathbb{Q}}(X) 
	=-\int_{0}^{T}\left\{ \nabla_{\theta}a_{\theta}(X_{t})\right\} ^{*}\Sigma(X_{t})^{-1}a_{\theta}(X_{t})dt 
	+ \int_{0}^{T}\left\{ \nabla_{\theta}a_{\theta}(X_{t})\right\} ^{*}\Sigma(X_{t})^{-1}dX_{t},
\end{align}
where $\nabla_{\theta}a_{\theta}(x)\in\mathbb{R}^{d\times d_{\theta}}$ denotes the 
Jacobian matrix w.r.t.\ $\theta$.  
%and $\Lambda(x)=\sigma(x)\Sigma^{-2}(x)\sigma(x)^{*}$ for $x\in\mathbb{R}^d$. 

We can consider some generalization of the above representation. 
Firstly, although we have assumed a deterministic initial condition of $X_0=x_{\star}$ to simplify our presentation, 
random initializations can also be accommodated. 
Assuming that $X_0\sim\mu_{\theta}$ is initialized from a distribution $\mu_{\theta}$ 
on $\mathbb{R}^d$ that admits a differentiable density w.r.t.\ the $d$-dimensional Lebesgue measure, 
a term of $\nabla_{\theta}\log\mu_{\theta}(X_0)$ should be added to the expression in \eqref{eqn:smoothing_functional}. 
This follows by conditioning on the value of $X_0$ and applying the above arguments. 
Secondly, our methodology can also easily handle continuous-time observation models.  
In this case, the term $\sum_{t=1}^T\nabla_{\theta} \log g_{\theta}(y_{t}|X_{t})$ in \eqref{eqn:smoothing_functional} 
would be replaced by the gradient of the log-conditional likelihood $\nabla_{\theta}\log p_{\theta}(y_{1:T}|X)$; see Section~\ref{sec:neural_network} for an application where the observational model is given by  an inhomogenous Poisson point process. 

When the diffusion coefficient also depends on unknown parameters, it may be possible to find an  
invertible transformation $\Psi:\mathbb{R}^d\rightarrow\mathbb{R}^d$, such that the transformed process 
$X_t = \Psi(Z_t)$ satisfies an SDE with a diffusion coefficient that is not parameter dependent. 
We will illustrate two examples of this principle in Section~\ref{sec:examples} using the Lamperti transformation 
and a simple rescaling of each component of the diffusion process. 
In more general cases where such transformations are not available, one has to seek a different representation 
of the score function as the above change of measure is no longer applicable. 
In Section~\ref{sec:parameter_diffusion}, we consider such a representation and discuss how our proposed methodology 
can be adapted. 

\subsection{Discrete-time approximation}\label{sec:discretized_score}
As alluded to in the introduction, we will rely on time-discretizations of the diffusion process \eqref{eq:intro_diff}. 
We will employ a hierarchy of discretizations of the time interval $[0,T]$, indexed by a level parameter $l\in\mathbb{N}_0$ that determines the temporal resolution. Higher levels with finer time resolutions will require increased algorithmic cost. 
For each level $l$, let $0= s_{0} < \cdots < s_{K_{l}} = T$ denote a dyadic uniform discretization of $[0,T]$, 
defined as $s_k=k \Delta_l$ for $k\in\{0,1,\ldots,K_l\}$, where $\Delta_l=2^{-l}$ is the step-size and $K_l=2^l T$ is the 
number of time steps.
Note that, by construction, $(s_k)_{k=0}^{K_l}$ contains the unit observation times $(t_p)_{p=1}^P$. 
%let $0= s_{0} < \cdots < s_{K_{l}} = T$ denote a discretization of $[0,T]$ that contains the 
%observation times, i.e. $(t_p)_{p=1}^P\subseteq (s_k)_{k=0}^{K_l}$ for $K_l\geq P$, and 
%$\Delta_{l,k} = s_k - s_{k-1}$ denote the corresponding step-sizes for $k\in\{1,\ldots,K_{l}\}$. 
We consider the Euler--Maruyama scheme which defines a time-discretized process $X_{0:T}=(X_{s_k})_{k=0}^{K_l}$ according 
to the following recursion for $k\in\{1,\ldots,K_{l}\}$:
\begin{align}\label{eq:disc_state}
	X_{s_k} = X_{s_{k-1}} + a_{\theta}(X_{s_{k-1}}) \Delta_{l} + \sigma(X_{s_{k-1}})(W_{s_k} - W_{s_{k-1}}), 
	\quad X_0 = x_{\star}.
\end{align}
To simplify our notation, we omit the dependence of the time grid and 
time-discretized process on the level parameter (until this distinction is necessary in Section~\ref{sec:unbiased_discretized_increment}), 
and write the transition in \eqref{eq:disc_state} as $X_{s_k}=F_{\theta}^l(X_{s_{k-1}},V_{s_k})$, where 
$V_{s_k}=W_{s_k} - W_{s_{k-1}}$ denotes the Brownian increment. 
Higher-order schemes such as the Milstein method could be employed but would be 
difficult to implement for problems with dimension $d\geq 2$. Future work could 
consider the antithetic truncated Milstein method of \cite{giles2014antithetic} for such settings. 

In the following, we write $\mathcal{N}_d(\mu,\Sigma)$ to denote a $d$-dimensional 
Gaussian distribution with mean $\mu\in\mathbb{R}^d$ and covariance $\Sigma\in\mathbb{R}^{d\times d}$, and 
its density as $x\mapsto\mathcal{N}_d(x;\mu,\Sigma)$. 
The notation $\mathcal{N}_d(0_d,I_d)$ refers to the standard Gaussian distribution, i.e.\ 
with zero mean vector $0_d\in\mathbb{R}^d$ and identity covariance matrix $I_d\in\mathbb{R}^{d\times d}$. 
Under time-discretization, the marginal likelihood of the resulting state space model is
\begin{align}
	p_{\theta}^l(y_{1:T}) = \int_{\mathsf{X}^l} 
	\prod_{t=1}^T g_{\theta}(y_{t}|x_{t}) \delta_{x_{\star}}(dx_0)\prod_{k=1}^{K_l}f_{\theta}^l(dx_{s_k}|x_{s_{k-1}}),
\end{align}
where $\mathsf{X}^l=(\mathbb{R}^d)^{K_l+1}$ is the path space, 
$\delta_{x_{\star}}(dx_0)$ refers to the Dirac measure at the deterministic initial condition $x_{\star}$, 
and $f_{\theta}^l(dx_{s_k}|x_{s_{k-1}})=\mathcal{N}_d(x_{s_k}; x_{s_{k-1}} + a_{\theta}(x_{s_{k-1}}) \Delta_{l},\Delta_{l}\Sigma(x_{s_{k-1}}))dx_{s_k}$ denotes the Gaussian transition kernel corresponding to \eqref{eq:disc_state}. 
Using Fisher's identity \citep[p. 353]{cappe2006inference}, the score function 
$S_l(\theta) = \nabla_{\theta}\log p_{\theta}^l(y_{1:T})$ at discretization level $l$ 
can be written as 
\begin{align}\label{eqn:discretized_score}
	 S_l(\theta) = \int_{\mathsf{X}^l} G_{\theta}^l(x_{0:T}) \, \pi_{\theta}^l(dx_{0:T}),
\end{align}
where $G_{\theta}^l : \mathsf{X}^l \rightarrow \mathbb{R}^{d_{\theta}}$, defined as 
\begin{align}\label{eqn:discrete_test_function}
	G_{\theta}^l(X_{0:T}) 
	= &-\sum_{k=1}^{K_l} \left\{ \nabla_{\theta}a_{\theta}(X_{s_{k-1}})\right\}^{*}\Sigma(X_{s_{k-1}})^{-1}
	a_{\theta}(X_{s_{k-1}})\Delta_{l} \notag\\
	& + \sum_{k=1}^{K_l} \left\{ \nabla_{\theta}a_{\theta}(X_{s_{k-1}})\right\}^{*}\Sigma(X_{s_{k-1}})^{-1}
	(X_{s_{k}} - X_{s_{k-1}}) + 
	\sum_{t=1}^T\nabla_{\theta}\log g_{\theta}(y_{t}|X_{t})
\end{align}
can be seen as an approximation of $G_{\theta}(X)$ in \eqref{eqn:smoothing_functional}, and 
\begin{align}\label{eqn:smoothing_distribution}
	\pi_{\theta}^l(dx_{0:T}) = p_{\theta}^l(y_{1:T})^{-1} 
	\prod_{t=1}^T g_{\theta}(y_{t}|x_{t}) \delta_{x_{\star}}(dx_0)\prod_{k=1}^{K_l}f_{\theta}^l(dx_{s_k}|x_{s_{k-1}})
\end{align}
is the smoothing distribution on $\mathsf{X}^l$. 
%Note that we omit the dependence of $G_{\theta}^l$ and $\pi_{\theta}^l$ on 
%$X_{s_0}=x_{\star}$ since the initialization is deterministic. 
In Section~\ref{sec:theory}, under appropriate regularity conditions, we will study the rate at which the discrete-time 
approximation $S_l(\theta)$ converges to the desired score $S(\theta)$ as the level $l$ tends to infinity. 
The following section concerns numerical approximations of these score functions.  

\section{Score estimation}\label{sec:score_estimation}
\subsection{Unbiased estimation strategy}
The goal of this section is to construct unbiased estimators of the score function $S(\theta)$ that admits the representation 
in \eqref{eqn:score_function} and \eqref{eqn:smoothing_functional}. For each $\theta\in\Theta$, convergence of the 
time-discretized score function \eqref{eqn:discretized_score} as $l\rightarrow\infty$ allows us to write 
\begin{align}\label{eqn:score_as_limit}
	S(\theta) = \lim_{L\rightarrow\infty}\sum_{l=0}^L I_l(\theta),
\end{align}
where $I_l(\theta) = S_l(\theta) - S_{l-1}(\theta)$ denotes the score increment at level $l\in\mathbb{N}_0$ (with $S_{-1}(\theta)=0$). 
We will design algorithms that allow us to construct unbiased estimators 
$\widehat{I}_l(\theta)$ of $I_l(\theta)$ independently for any $l\in\mathbb{N}_0$, i.e.\ 
\begin{align}\label{eq:ub1}
	\mathbb{E}\left[\widehat{I}_l(\theta)\right] = I_l(\theta),\quad l\in\mathbb{N}_0,
\end{align}
where $\mathbb{E}$ denotes expectation w.r.t.\ the law of our algorithms.  
The key insight of the debiasing schemes proposed by \citet{mcleish2011,rhee2015unbiased} is to 
perform a random truncation of the highest level $L$ in \eqref{eqn:score_as_limit}, so that unbiasedness 
is retained when unbiased estimators of the increments are employed. 

Let $(P_l)_{l=0}^{\infty}$ be a given probability mass function (PMF) with support on $\mathbb{N}_0$ and 
define its cumulative tail probabilities as $\mathcal{P}_l = \sum_{k=l}^{\infty} P_k$ for $l\in\mathbb{N}_0$. 
Suppose that 
\begin{align}\label{eq:ub2}
	\sum_{l=0}^\infty \mathcal{P}_l^{-1}\left\lbrace \mathrm{Var}\left[\widehat{I}_l(\theta)^j\right] + 
	\left(S_{l}(\theta)^j-S(\theta)^j\right)^2\right\rbrace < \infty 
\end{align}
for all $j\in\{1,\ldots,d_{\theta}\}$, where $\mathrm{Var}$ denotes variance under $\mathbb{E}$ and 
$x^j$ refers to the $j^{th}$-component of the vector $x\in\mathbb{R}^{d_{\theta}}$. 
If we sample $L$ from $(P_l)_{l=0}^{\infty}$ independently of $(\widehat{I}_l(\theta))_{l=0}^{\infty}$, 
it follows from \citet{rhee2015unbiased} that the \emph{independent sum} estimator
\begin{align}\label{eq:ub3}
	\widehat{S}(\theta) = \sum_{l=0}^L \frac{\widehat{I}_l(\theta)}{\mathcal{P}_l}
\end{align} 
is an unbiased estimator of $S(\theta)$ with finite variance, i.e.\ $\mathbb{E}[\widehat{S}(\theta)]=S(\theta)$ 
and the entries of the covariance matrix $\mathrm{Var}[\widehat{S}(\theta)]$ are finite. 
Inspection of \eqref{eq:ub2} reveals that we have to understand how fast $S_l(\theta)$ converges to $S(\theta)$ 
as $l\rightarrow\infty$, which is established in Theorem \ref{prop:conv_grad_log_like}. 
Moreover, for \eqref{eq:ub1} and \eqref{eq:ub2} to hold, we have to 
compute an unbiased estimator of the score $S_0(\theta)$ at level $l=0$ with finite variance, and 
construct unbiased estimators of the score increments $I_l(\theta)$, whose variance vanishes 
sufficiently fast as $l\rightarrow\infty$ relative to the tails of $(P_l)_{l=0}^{\infty}$.
Developing a methodology that meets these two requirements will be the focus of 
Sections~\ref{sec:unbiased_discretized_score} and \ref{sec:unbiased_discretized_increment} respectively. 
Verifying that our proposed methodology satisfies these requirements constitutes our main theoretical result in Theorem \ref{theo:ub}. 
Alternatives to the estimator in \eqref{eq:ub3} such as the \emph{single term} estimator 
of \citet{rhee2015unbiased} could also be considered here.

\subsection{Unbiased estimation of time-discretized scores}\label{sec:unbiased_discretized_score}
This section considers unbiased estimation of the score $S_l(\theta)$ at discretization level $l\in\mathbb{N}_0$; 
the case $l=0$ will be employed to construct the first summand $\widehat{I}_0(\theta)$ in \eqref{eq:ub3}. 
Our basic algorithmic building block is the CPF of \citet{andrieu2010particle}, which defines a Markov kernel 
on the space of trajectories $X_{0:T}\in\mathsf{X}^l$ that admits 
the smoothing distribution $\pi_{\theta}^l$ as its invariant distribution. 
A detailed description for our application is given in Algorithm \ref{alg:CPF}, 
which has a complexity of $\mathcal{O}(NK_l)$. 
The CPF involves simulating $N\geq 2$ trajectories under the time-discretized model dynamics (Steps 2a \& 2b), 
weighting samples according to the observation density (Step 2c), and resampling from the weighted particle approximation (Step 2d). 
We will consider multinomial resampling, in which case, $\mathcal{R}(w_t^{1:N})$ refers to the categorical distribution 
on $\{1,\ldots,N\}$ with probabilities $w_t^{1:N}=(w_t^n)_{n=1}^N$. 
The main difference to a standard bootstrap particle filter (BPF) \citep{gordon1993novel} is that the input trajectory 
$X_{0:T}^{\star}$ is conditioned to survive all resampling steps. 
The algorithm outputs a trajectory $X_{0:T}^{\circ}$ by sampling the particle indexes 
$(B_{s_k})_{k=0}^{K_l}$ after the terminal step (Steps 3a \& 3b). 
We will write $X_{0:T}^{\circ}\sim M_{\theta}^l(\cdot |X_{0:T}^{\star})$ 
to denote a trajectory generated by the CPF kernel at parameter $\theta\in\Theta$ and 
discretization level $l\in\mathbb{N}_0$. 
%Let $\nu_{\theta}^l$ denote a distribution on $\mathsf{X}^l$ which will be used for initialization. 
%Our analysis in Section~\ref{sec:theory} will consider the law of the time-discretized dynamics, i.e. 

We will initialize the CPF Markov chain using the law of the time-discretized dynamics 
\begin{align}\label{eqn:law_discrete_dynamics}
	\nu_{\theta}^l(dx_{0:T})=\delta_{x_{\star}}(dx_0)\prod_{k=1}^{K_l}f_{\theta}^l(dx_{s_k}|x_{s_{k-1}}). 
\end{align}
%Our simulations in Section~\ref{sec:examples} 
One could also consider the law of a trajectory sampled from a BPF, 
which provides a good approximation of $\pi_{\theta}^l$ for sufficiently large $N$ \citep[Theorem 1]{andrieu2010particle}. 
Under mild assumptions, the Markov chain $(X_{0:T}(i))_{i=0}^{\infty}$ generated by 
\begin{align}\label{eqn:cpf_chain}
	X_{0:T}(0)\sim\nu_{\theta}^l,\quad X_{0:T}(i)\sim M_{\theta}^l(\cdot |X_{0:T}(i-1)),
\end{align}
for $i\geq1$, is uniformly ergodic \citep{chopin2015particle, lindsten2015uniform,andrieu2018uniform}. 
Hence one can adopt a standard MCMC approach to approximate the score $S_l(\theta)$ by the average 
\begin{align}
	 A_{l}^{b:I}(\theta) = \frac{1}{I-b+1}\sum_{i=b}^IG_{\theta}^l(X_{0:T}(i)),
\end{align}
for some fixed burn-in $0\leq b\leq I$, which is consistent as the number of iterations $I\rightarrow\infty$. 
However, as the Markov chain is not started at stationarity, the MCMC average $A_{l}^{b:I}(\theta)$ is a biased estimator 
for any finite $I\in\mathbb{N}$. Although this burn-in bias can be reduced by increasing $b$, 
tuning it to control the bias is a difficult task in practice. 

By building on the work of \citet{glynn2014exact}, \citet{jacob2020smoothing} showed how to obtain 
unbiased estimators of smoothing expectations by simulating a pair of coupled CPF chains 
$(X_{0:T}(i),\bar{X}_{0:T}(i))_{i=0}^{\infty}$ on the product space $\mathsf{Z}^l=\mathsf{X}^l\times\mathsf{X}^l$ 
with the same marginal law. 
This is achieved using a coupling of two CPFs as described in Algorithm \ref{alg:2-CCPF}, which we will 
refer to as the 2-CCPF. 
Writing $(X_{0:T}^{\circ},\bar{X}_{0:T}^{\circ})\sim \bar{M}_{\theta}^l(\cdot|X_{0:T}^{\star},\bar{X}_{0:T}^{\star})$ 
to denote a pair of trajectories generated by the 2-CCPF kernel given $(X_{0:T}^{\star},\bar{X}_{0:T}^{\star})\in\mathsf{Z}^l$ as input, 
marginally we have $X_{0:T}^{\circ}\sim M_{\theta}^l(\cdot |X_{0:T}^{\star})$ and 
$\bar{X}_{0:T}^{\circ}\sim M_{\theta}^l(\cdot |\bar{X}_{0:T}^{\star})$. 
The two main ingredients of this coupling are the use of common Brownian increments 
(Steps 2a \& 2b), and a coupling of the resampling 
distributions $\mathcal{R}(w_{t}^{1:N})$ and $\mathcal{R}(\bar{w}_{t}^{1:N})$ 
denoted as $\bar{\mathcal{R}}(w_{t}^{1:N},\bar{w}_{t}^{1:N})$ (Steps 2e \& 3a). 
We refer readers to the references in \citet{jacob2020smoothing} for various coupled resampling schemes. 
Our focus is the maximal coupling \citep{chopin2015particle,jasra2017multilevel} that maximizes the 
probability of having identical ancestors at each step of the CPFs. 
This is detailed in Algorithm \ref{alg:2maximal}, where we have suppressed the time dependence for notational simplicity. 
We could also consider an improvement of Step 5 
that samples from the residual distributions with common uniform random variables. 
%Although we will only analyze Algorithm \ref{alg:2maximal}, our simulations will also consider an improvement of Step 5 
%that samples from the residual distributions with common uniform random variables. 
As the cost of implementing Algorithm \ref{alg:2maximal} is $\mathcal{O}(N)$, the overall cost of Algorithm 
\ref{alg:2-CCPF} is still $\mathcal{O}(NK_l)$. 

We initialize by sampling $(X_{0:T}(0),\bar{X}_{0:T}(0))$ from a coupling $\bar{\nu}_{\theta}^l$ with 
$\nu_{\theta}^l$ as its marginals. 
The choice of $\bar{\nu}_{\theta}^l$ could be explored but we will consider the independent coupling for simplicity. 
The pair of CPF chains is then generated as 
\begin{align}\label{eqn:coupled2_chains}
	X_{0:T}(1)\sim M_{\theta}^l(\cdot |X_{0:T}(0)),\quad 
	(X_{0:T}(i+1),\bar{X}_{0:T}(i))\sim  \bar{M}_{\theta}^l(\cdot | X_{0:T}(i),\bar{X}_{0:T}(i-1)), 
\end{align}
for $i\geq 1$. Marginally, $(X_{0:T}(i))_{i=0}^{\infty}$ and $(\bar{X}_{0:T}(i))_{i=0}^{\infty}$ have the same 
law as the Markov chain generated by \eqref{eqn:cpf_chain}. 
It can be shown that each application of 2-CCPF allows the chains to meet with some positive probability 
\citep[Theorem 3.1]{jacob2020smoothing}, that depends on the number of trajectories $N$ and observations $T$ 
\citep[Theorem 8 \& 9]{lee2020coupled}. Moreover, the construction of 2-CCPF ensures that the chains are faithful, 
i.e.\ $X_{0:T}(i)=\bar{X}_{0:T}(i-1)$ for all $i\geq \tau_{\theta}^l$, where $\tau_{\theta}^l=\inf\{i\geq 1: X_{0:T}(i)=\bar{X}_{0:T}(i-1)\}$ 
denotes the meeting time. Using the time-averaged estimator of \citet{jacob2020smoothing}, 
an unbiased estimator of the score $S_l(\theta)$ is given by
\begin{align}\label{eqn:unbiased_discretized_score}
	 \widehat{S}_l(\theta) = A_{l}^{b:I}(\theta) + \sum_{i=b+1}^{\tau_{\theta}^l-1}\min\left(1,\frac{i-b}{I-b+1}\right)\left(G_{\theta}^l(X_{0:T}(i)) 
	 - G_{\theta}^l(\bar{X}_{0:T}(i-1))\right).
\end{align}
The second term corrects for the bias of the MCMC average $A_{l}^{b:I}(\theta)$ and is equal to zero if $b+1 > \tau_{\theta}^l-1$. 
Under assumptions that will be stated in Section~\ref{sec:theory}, $\widehat{S}_l(\theta)$ 
has finite variance and finite expected cost, for any choice of $N\geq2$ and $0\leq b\leq I$. 
Assuming that 2-CCPF costs twice as much as CPF, the cost of computing $\widehat{S}_l(\theta)$ 
is $\max\{2\tau_{\theta}^l-1,I+\tau_{\theta}^l-1\}$ applications of the CPF kernel $M_{\theta}^l$. 
%In Section \ref{sec:examples}, we will experiment with various choices of $(N, b, I)$ 
%to maximize efficiency. 

\begin{algorithm}
\protect\caption{Conditional particle filter (CPF) at parameter $\theta\in\Theta$ and discretization level $l\in\mathbb{N}_0$  \label{alg:CPF}}

\textbf{Input}: a trajectory $X_{0:T}^{\star}=(X_{s_k}^{\star})_{k=0}^{K_l}\in\mathsf{X}^{l}$. 

For time step $k=0$
\begin{description}[itemsep=0pt,parsep=0pt,topsep=0pt,labelindent=0.5cm]
	\item (1a) Set $X_{0}^{n}=x_{\star}$ for $n\in\{1,\ldots,N\}$.
	
	\item (1b) Set $w_0^{n}=N^{-1}$ and $A_{0}^{n}=n$ for $n\in\{1,\ldots,N\}$.  
\end{description}

\textcompwordmark{}

For time step $k\in\{1,\ldots,K_l\}$ 
\begin{description}[itemsep=0pt,parsep=0pt,topsep=0pt,labelindent=0.5cm]
	\item (2a) Sample Brownian increment $V_{s_k}^n \sim \mathcal{N}_d(0_d,\Delta_{l}I_d)$ independently for $n\in\{1,\ldots,N-1\}$.

	\item (2b) Set $X_{s_k}^n =F_{\theta}^l(X_{s_{k-1}}^{A_{s_{k-1}}^{n}}, V_{s_k}^n)$ 
		for $n\in\{1,\ldots,N-1\}$, and $X_{s_k}^N = X_{s_k}^{\star}$.

	\item If there is an observation at time $t=s_k\in\{1,\ldots,T\}$ 
	\begin{description}[itemsep=0pt,parsep=0pt,topsep=0pt,labelindent=0.5cm]
            \item (2c) Compute normalized weight $w_t^n \propto g_{\theta}(y_{t}|X_{t}^n)$ for $n\in\{1,\ldots,N\}$.
            
	   \item (2d) If $t < T$, sample ancestor $A_{t}^{n}\sim \mathcal{R}(w_{t}^{1:N})$ independently 
            for $n\in\{1,\ldots,N-1\}$ and set $A_{t}^N=N$.
	\end{description}

	\item Else 
	\begin{description}[itemsep=0pt,parsep=0pt,topsep=0pt,labelindent=0.5cm]
 		\item (2e) Set $A_{s_k}^{n}=n$ for $n\in\{1,\ldots,N\}$.
	\end{description}
\end{description}

\textcompwordmark{}

After the terminal step
\begin{description}[itemsep=0pt,parsep=0pt,topsep=0pt,labelindent=0.5cm]
	\item (3a) Sample particle index $B_{T}\sim \mathcal{R}(w_{T}^{1:N})$. 

	\item (3b) Set particle index $B_{s_k} = A_{s_k}^{B_{s_{k+1}}}$ for $k\in\{0,1,\ldots,K_l-1\}$. 
\end{description}

\textbf{Output}: a trajectory $X_{0:T}^{\circ}=(X_{s_k}^{B_{s_k}})_{k=0}^{K_l}\in\mathsf{X}^{l}$.
\end{algorithm}

\begin{algorithm}[H]
\protect\caption{Two coupled CPF (2-CCPF) at parameter $\theta\in\Theta$ and discretization level $l\in\mathbb{N}_0$\label{alg:2-CCPF}}

\textbf{Input}: a pair of trajectories $(X_{0:T}^{\star},\bar{X}_{0:T}^{\star})=(X_{s_k}^{\star},\bar{X}_{s_k}^{\star})_{k=0}^{K_l}\in\mathsf{Z}^{l}$. 

For time step $k=0$
\begin{description}[itemsep=0pt,parsep=0pt,topsep=0pt,labelindent=0.5cm]
	\item (1a) Set $X_{0}^{n}=x_{\star}$ and $\bar{X}_{0}^{n}=x_{\star}$ for $n\in\{1,\ldots,N\}$.

	\item (1b) Set $w_0^{n}=N^{-1}, \bar{w}_0^{n}=N^{-1}$ and $A_{0}^{n}=n, \bar{A}_{0}^{n}=n$ for $n\in\{1,\ldots,N\}$.
\end{description}
\textcompwordmark{}

For time step $k\in\{1,\ldots,K_l\}$ 
\begin{description}[itemsep=0pt,parsep=0pt,topsep=0pt,labelindent=0.5cm]
	\item (2a) Sample Brownian increment $V_{s_k}^n \sim \mathcal{N}_d(0_d,\Delta_{l}I_d)$ independently for $n\in\{1,\ldots,N-1\}$.

	\item (2b) Set $X_{s_k}^n = F_{\theta}^l(X_{s_{k-1}}^{A_{s_{k-1}}^{n}},V_{s_k}^n)$ and 
	$\bar{X}_{s_k}^n = F_{\theta}^l(\bar{X}_{s_{k-1}}^{\bar{A}_{s_{k-1}}^{n}},V_{s_k}^n)$ 
	for $n\in\{1,\ldots,N-1\}$. 
	
	\item (2c) Set $X_{s_k}^N = X_{s_k}^{\star}$ and $\bar{X}_{s_k}^N = \bar{X}_{s_k}^{\star}$.

	\item If there is an observation at time $t=s_k\in\{1,\ldots,T\}$ 
	\begin{description}[itemsep=0pt,parsep=0pt,topsep=0pt,labelindent=0.5cm]
		\item (2d) Compute normalized weights
		$w_t^n \propto g_{\theta}(y_{t}|X_{t}^n)$ and 
		$\bar{w}_t^n \propto g_{\theta}(y_{t}|\bar{X}_{t}^n)$ for $n\in\{1,\ldots,N\}$.

		\item (2e) If $t < T$, sample ancestors $(A_{t}^{n},\bar{A}_{t}^{n})\sim \bar{\mathcal{R}}(w_{t}^{1:N},\bar{w}_{t}^{1:N})$ independently 
		for $n\in\{1,\ldots,N-1\}$ and set $A_{t}^N=N,\bar{A}_{t}^N=N$.
	\end{description}

	\item Else 
	\begin{description}[itemsep=0pt,parsep=0pt,topsep=0pt,labelindent=0.5cm]
		\item (2f) Set $A_{s_k}^{n}=n$ and $\bar{A}_{s_k}^{n}=n$ for $n\in\{1,\ldots,N\}$.
	\end{description}
\end{description}

\textcompwordmark{}

After the terminal step
\begin{description}[itemsep=0pt,parsep=0pt,topsep=0pt,labelindent=0.5cm]
	\item (3a) Sample particle indexes $(B_{T},\bar{B}_{T})\sim \bar{\mathcal{R}}(w_{T}^{1:N},\bar{w}_{T}^{1:N})$. 

	\item (3b) Set particle indexes $B_{s_k} = A_{s_k}^{B_{s_{k+1}}}$ and $\bar{B}_{s_k} = \bar{A}_{s_k}^{\bar{B}_{s_{k+1}}}$ for $k\in\{0,1,\ldots,K_l-1\}$. 
\end{description}

\textbf{Output}: a pair of trajectories $(X_{0:T}^{\circ},\bar{X}_{0:T}^{\circ})=(X_{s_k}^{B_{s_k}},\bar{X}_{s_k}^{\bar{B}_{s_k}})_{k=0}^{K_l}\in\mathsf{Z}^l$. 
\end{algorithm}

\begin{algorithm}
\protect\caption{Maximal coupling of two resampling distributions $\mathcal{R}(w^{1:N})$ and $\mathcal{R}(\bar{w}^{1:N})$ \label{alg:2maximal}}

\textbf{Input}: normalized weights $w^{1:N}=(w^n)_{n=1}^N$ and $\bar{w}^{1:N}=(\bar{w}^n)_{n=1}^N$.

(1) Compute the overlap $o^n=\min\{w^n,\bar{w}^n\}$ for $n\in\{1,\ldots,N$\}. 

(2) Compute the mass of the overlap $\mu=\sum_{n=1}^No^n$ and normalize $O^n=o^n/\mu$ for $n\in\{1,\ldots,N\}$. 

(3) Compute the residuals $r^n=(w^n-o^n)/(1-\mu)$ and $\bar{r}^n=(\bar{w}^n-o^n)/(1-\mu)$ for $n\in\{1,\ldots,N\}$. 

With probability $\mu$ 

\begin{description}[itemsep=0pt,parsep=0pt,topsep=0pt,labelindent=0.5cm]
	\item (4) Sample $A\sim\mathcal{R}(O^{1:N})$ and set $\bar{A}=A$.
\end{description}

Otherwise
\begin{description}[itemsep=0pt,parsep=0pt,topsep=0pt,labelindent=0.5cm]
	\item (5) Sample $A\sim\mathcal{R}(r^{1:N})$ and $\bar{A}\sim\mathcal{R}(\bar{r}^{1:N})$ independently. 
\end{description}

\textbf{Output}: indexes $(A,\bar{A})$.
\end{algorithm}

\subsection{Unbiased estimation of score increments}\label{sec:unbiased_discretized_increment}
We now consider unbiased estimation of the score increment $I_l(\theta) = S_l(\theta) - S_{l-1}(\theta)$ 
at level $l\in\mathbb{N}$ to construct the term $\widehat{I}_l(\theta)$ 
in the independent sum estimator $\widehat{S}(\theta)$ in \eqref{eq:ub3}. 
A naive approach that employs the unbiased estimation framework in Section~\ref{sec:unbiased_discretized_score} 
to estimate the scores $S_{l-1}(\theta)$ and $S_l(\theta)$ independently, and take the difference to estimate 
the increment $I_l(\theta)$, will satisfy the unbiasedness requirement in \eqref{eq:ub1}. 
However, as the variance of the increment will not decrease with the discretization level $l$ under independent estimation 
of the successive scores, one cannot choose a PMF $(P_l)_{l=0}^{\infty}$ such that 
the condition in \eqref{eq:ub2} holds. 
This prompts an extension of the preceding framework that allows the scores in each increment to be estimated 
in a dependent manner. 

%Our basic algorithmic building block is the CPF of \citet{andrieu2010particle}, which defines a Markov kernel 
%on the space of trajectories $X_{0:T}\in\mathsf{X}^l$ that admits 
%the smoothing distribution $\pi_{\theta}^l$ as its invariant distribution. 

Our proposed methodology involves simulating two pairs of coupled CPF chains, 
$(X_{0:T}^{l-1}(i),\bar{X}_{0:T}^{l-1}(i))_{i=0}^{\infty}$ on $\mathsf{Z}^{l-1}$ 
for discretization level $l-1$, 
and $(X_{0:T}^{l}(i),\bar{X}_{0:T}^{l}(i))_{i=0}^{\infty}$ on $\mathsf{Z}^{l}$ 
for discretization level $l$. 
This relies on the coupling of four CPFs detailed in Algorithm \ref{alg:4-CCPF}, 
which will be referred to as the 4-CCPF. 
In the algorithmic description, $(s_k)_{k=0}^{K_l}$ denotes a time discretization of $[0,T]$ at level $l$, 
and a trajectory at the coarser level $l-1$ is written as $X_{0:T}^{l-1}=(X_{s_{2k}}^{l-1})_{k=0}^{K_{l-1}}$.  
Given trajectories $(X_{0:T}^{l-1,\star},\bar{X}_{0:T}^{l-1,\star},X_{0:T}^{l,\star},\bar{X}_{0:T}^{l,\star})\in 
\mathsf{Z}^{l-1}\times\mathsf{Z}^{l}$ as input, 
we will write 
\begin{align}\label{eqn:4CCPF}
	(X_{0:T}^{l-1,\circ},\bar{X}_{0:T}^{l-1,\circ},X_{0:T}^{l,\circ},\bar{X}_{0:T}^{l,\circ})\sim 
	\bar{M}_{\theta}^{l-1,l}(\cdot|X_{0:T}^{l-1,\star},\bar{X}_{0:T}^{l-1,\star},X_{0:T}^{l,\star},\bar{X}_{0:T}^{l,\star})
\end{align}
to denote the trajectories generated by 4-CCPF. 
The 4-CCPF kernel $\bar{M}_{\theta}^{l-1,l}$ is a four-marginal coupling of the CPF kernels $M_{\theta}^{l-1}$ 
and $M_{\theta}^{l}$, in the sense that marginally, we have $X_{0:T}^{l-1,\circ}\sim M_{\theta}^{l-1}(\cdot |X_{0:T}^{l-1,\star})$ and 
$\bar{X}_{0:T}^{l-1,\circ}\sim M_{\theta}^{l-1}(\cdot |\bar{X}_{0:T}^{l-1,\star})$ at level $l-1$, and 
$X_{0:T}^{l,\circ}\sim M_{\theta}^l(\cdot |X_{0:T}^{l,\star})$ and 
$\bar{X}_{0:T}^{l,\circ}\sim M_{\theta}^l(\cdot |\bar{X}_{0:T}^{l,\star})$ at level $l$. 
The two main ingredients of 4-CCPF are the use of common Brownian increments within each level 
(Steps 2b \& 2e) and across levels (Steps 2a \& 2d), 
and an appropriate four-marginal coupling of the resampling 
distributions 
$\mathcal{R}(w_{t}^{l-1,1:N})$, $\mathcal{R}(\bar{w}_{t}^{l-1,1:N})$, 
$\mathcal{R}(w_{t}^{l,1:N})$, $\mathcal{R}(\bar{w}_{t}^{l,1:N})$  
denoted by 
$\bar{\mathcal{R}}(w_{t}^{l-1,1:N},\bar{w}_{t}^{l-1,1:N},w_{t}^{l,1:N},\bar{w}_{t}^{l,1:N})$ (Steps 2h \& 3a).
While the use of common Brownian increments is a standard choice in MLMC 
\citep{giles2008multilevel}, constructing coupled resampling schemes that induce sufficient dependencies 
between the four CPF chains, for the variance of the estimated increment to decrease with the discretization level, 
requires new algorithmic design. 

Algorithm \ref{alg:maximal-maximal} presents a coupled resampling scheme that will be the focus of our analysis in Section~\ref{sec:theory}. 
Here we suppress the time dependence for notational simplicity. 
Given normalized weights $w^{l-1,1:N}=(w^{l-1,n})_{n=1}^N$, $\bar{w}^{l-1,1:N}=(\bar{w}^{l-1,n})_{n=1}^N$ 
at discretization level $l-1$ and $w^{l,1:N}=(w^{l,n})_{n=1}^N$, $\bar{w}^{l,1:N}=(\bar{w}^{l,n})_{n=1}^N$ at discretization level $l$, 
the algorithm samples ancestor indexes $(A^{l-1},\bar{A}^{l-1},A^{l},\bar{A}^{l})$ from the maximal coupling of the maximal couplings 
$\bar{\mathcal{R}}(w^{l-1,1:N},w^{l,1:N})$ and $\bar{\mathcal{R}}(\bar{w}^{l-1,1:N},\bar{w}^{l,1:N})$. 
That is, amongst all possible four-marginal couplings, this scheme maximizes the probabilities of having 
identical ancestors across levels, i.e.\ $A^{l-1}=A^l$ and $\bar{A}^{l-1}=\bar{A}^l$, 
and identical pair of ancestors within the levels, i.e.\ $(A^{l-1},A^l)=(\bar{A}^{l-1},\bar{A}^l)$. 
The cost of Algorithm \ref{alg:maximal-maximal} is random as it employs rejection samplers \citep{thorisson2000} 
in Steps 2b, 3b and 4, wherein $\mathcal{U}_{[0,1]}$ denotes the uniform distribution on $[0,1]$. 
As the expected cost is $\mathcal{O}(N)$, the overall cost of Algorithm~\ref{alg:4-CCPF} is $\mathcal{O}(NK_l)$ on average. 
Note that a naive approach to sample from the desired coupled resampling scheme in place of 
Algorithm \ref{alg:maximal-maximal} would involve a deterministic but prohibitive cost of $\mathcal{O}(N^2)$. 
In Step 4, we denote the respective PMF of 
$\bar{\mathcal{R}}(w^{l-1,1:N},w^{l,1:N})$ and $\bar{\mathcal{R}}(\bar{w}^{l-1,1:N},\bar{w}^{l,1:N})$ as
\begin{align}\label{eqn:PMF_maximal}
	&R^{l-1,l}(A,B)=\mathbb{I}_{\mathcal{D}}(A,B)o^{A} + 
		\frac{(w^{l-1,A}-o^A)(w^{l,B}-o^B)}{1-\mu}, \\
	&\bar{R}^{l-1,l}(A,B)=\mathbb{I}_{\mathcal{D}}(A,B)\bar{o}^{A} + 
		\frac{(\bar{w}^{l-1,A}-\bar{o}^A)(\bar{w}^{l,B}-\bar{o}^B)}{1-\bar{\mu}}, \notag
\end{align}
for $(A,B)\in\{1,\ldots,N\}^2$, where $\mathbb{I}_{\mathcal{D}}(A,B)$ is the indicator function on the diagonal set 
$\mathcal{D}=\{(A,B)\in\{1,\ldots,N\}^2 : A=B\}$, 
$o^n=\min\{w^{l-1,n},w^{l,n}\}$ and $\bar{o}^n=\min\{\bar{w}^{l-1,n},\bar{w}^{l,n}\}$ for $n\in\{1,\ldots,N\}$ 
are the overlapping measures, and 
$\mu=\sum_{n=1}^No^n$ and $\bar{\mu}=\sum_{n=1}^N\bar{o}^n$ 
are their corresponding mass. 
From the expressions in \eqref{eqn:PMF_maximal}, one can check that the three cases considered in Algorithm 
\ref{alg:maximal-maximal} are necessary to ensure faithfulness of the pair of chains on each discretization level. 
More precisely, if the input trajectories satisfy $X_{0:T}^{l-1,\star}=\bar{X}_{0:T}^{l-1,\star}$ 
and/or $X_{0:T}^{l,\star}=\bar{X}_{0:T}^{l,\star}$, then under the 4-CCPF \eqref{eqn:4CCPF}, 
the output trajectories satisfy 
$X_{0:T}^{l-1,\circ}=\bar{X}_{0:T}^{l-1,\circ}$ and/or $X_{0:T}^{l,\circ}=\bar{X}_{0:T}^{l,\circ}$ 
almost surely. 

We note that the two-marginal couplings induced by the 4-CCPF kernel $\bar{M}_{\theta}^{l-1,l}$ 
on each discretization level are not the same as the 2-CCPF kernels $\bar{M}_{\theta}^{l-1}$ and $\bar{M}_{\theta}^l$. 
Although it is not a requirement of our methodology, 
this property would hold if we consider a modification of Algorithm \ref{alg:maximal-maximal} that samples from 
the maximal coupling of the maximal couplings 
$\bar{\mathcal{R}}(w^{l-1,1:N},\bar{w}^{l-1,1:N})$ and $\bar{\mathcal{R}}(w^{l,1:N},\bar{w}^{l,1:N})$. 
However, as with simple coupling strategies based on common random variables, 
such a coupled resampling scheme does not induce adequate dependencies between the CPF chains across discretization levels. 
This will be illustrated experimentally in Section~\ref{sec:ornstein_uhlenbeck}. 
To understand the rationale behind Algorithm \ref{alg:maximal-maximal}, we observe that the 
two-marginal couplings induced by the 4-CCPF kernel across discretization levels 
are given by the multilevel CPF (ML-CPF) described in Algorithm \ref{alg:ML-CPF}. 
That is, writing $M_{\theta}^{l-1,l}$ as the ML-CPF kernel,
which is a coupling of the CPF kernels $M_{\theta}^{l-1}$ and $M_{\theta}^{l}$, we have 
$(X_{0:T}^{l-1,\circ},X_{0:T}^{l,\circ})\sim M_{\theta}^{l-1,l}(\cdot|X_{0:T}^{l-1,\star},X_{0:T}^{l,\star})$ and 
$(\bar{X}_{0:T}^{l-1,\circ},\bar{X}_{0:T}^{l,\circ})\sim M_{\theta}^{l-1,l}(\cdot|\bar{X}_{0:T}^{l-1,\star},\bar{X}_{0:T}^{l,\star})$ 
under the 4-CCPF kernel in \eqref{eqn:4CCPF}. 
The ML-CPF is similar to the multilevel particle filter of \citet{jasra2017multilevel},  
who proposed multilevel estimators of filtering expectations that are non-asymptotically biased 
but consistent in the limit of our computational budget. 
Even though our objective is markedly different, as we seek non-asymptotically unbiased and 
(almost surely) finite cost estimators of the score function which is a smoothing expectation, 
the connection to multilevel estimation alludes to better convergence rates than Monte Carlo 
approaches based on the finest discretization level. 

We now describe simulation of the two pairs of CPF chains 
$(X_{0:T}^{l-1}(i),\bar{X}_{0:T}^{l-1}(i))_{i=0}^{\infty}$ and 
$(X_{0:T}^{l}(i),\bar{X}_{0:T}^{l}(i))_{i=0}^{\infty}$ using ML-CPF and 4-CCPF. 
We initialize $(X_{0:T}^{l-1}(0),\bar{X}_{0:T}^{l-1}(0),X_{0:T}^{l}(0),\bar{X}_{0:T}^l(0))$ from 
a four-marginal coupling $\bar{\nu}_{\theta}^{l-1,l}$ that satisfies 
$X_{0:T}^{l-1}(0),\bar{X}_{0:T}^{l-1}(0)\sim\nu_{\theta}^{l-1}$ and $X_{0:T}^{l}(0),\bar{X}_{0:T}^l(0)\sim\nu_{\theta}^{l}$. 
For simplicity, we assume $\bar{\nu}_{\theta}^{l-1,l}$ is such that 
each of the pairs across levels $(X_{0:T}^{l-1}(0),X_{0:T}^{l}(0))$ and $(\bar{X}_{0:T}^{l-1}(0),\bar{X}_{0:T}^{l}(0))$ 
independently follow a coupling of $\nu_{\theta}^{l-1}$ and $\nu_{\theta}^{l}$, denoted as $\nu_{\theta}^{l-1,l}$. 
The choice of $\nu_{\theta}^{l-1,l}$ will be taken as the joint law of the time-discretized dynamics 
under common Brownian increments in our analysis. 
%or the law of a pair of trajectories sampled from 
%a multilevel particle filter, i.e. a modification of the ML-CPF in Algorithm \ref{alg:ML-CPF} 
%that does not condition on the survival of input trajectories. 
We then sample $(X_{0:T}^{l-1}(1),X_{0:T}^l(1))\sim M_{\theta}^{l-1,l}(\cdot | X_{0:T}^{l-1}(0),X_{0:T}^l(0))$ with ML-CPF, 
and subsequently for $i\geq 1$, iteratively sample 
\begin{align}\label{eqn:coupled4_chains}
	(X_{0:T}^{l-1}(i+1),\bar{X}_{0:T}^{l-1}(i),X_{0:T}^{l}(i+1),\bar{X}_{0:T}^{l}(i))\sim 
	\bar{M}_{\theta}^{l-1,l}(\cdot|X_{0:T}^{l-1}(i),\bar{X}_{0:T}^{l-1}(i-1),X_{0:T}^{l}(i),\bar{X}_{0:T}^{l}(i-1))
\end{align}
from 4-CCPF. Marginally, the single CPF chains have the same law as a Markov chain generated by \eqref{eqn:cpf_chain} at discretization level $l-1$ or $l$. 
Since each application of 4-CCPF allows the pair of chains on each level to meet with some positive probability (see Lemma \ref{lem:diag_prob_kernel} in  Appendix~\ref{app:CCPF}), 
and by construction remain faithful thereafter, we define 
the meeting time at level $l$ as $\tau_{\theta}^l=\inf\{i\geq 1: X_{0:T}^l(i)=\bar{X}_{0:T}^l(i-1)\}$  
and the stopping time at level $l$ as $\bar{\tau}_{\theta}^l=\max\{\tau_{\theta}^{l-1},\tau_{\theta}^l\}$. 
Note that the 4-CCPF collapses to the ML-CPF after the stopping time, i.e.\ 
for $i>\bar{\tau}_{\theta}^l$, the transition in \eqref{eqn:coupled4_chains} is equivalent to sampling 
$(X_{0:T}^{l-1}(i+1),X_{0:T}^{l}(i+1))\sim M_{\theta}^{l-1,l}(\cdot|X_{0:T}^{l-1}(i),X_{0:T}^{l}(i))$ and 
setting $\bar{X}_{0:T}^{l-1}(i)=X_{0:T}^{l-1}(i+1), \bar{X}_{0:T}^{l}(i)=X_{0:T}^{l}(i+1)$. 
For any choice of burn-in $b\in\mathbb{N}_0$ and number of iterations $I\geq b$, we can compute unbiased estimators 
$\widehat{S}_{l-1}(\theta)$ and $\widehat{S}_l(\theta)$
of the scores $S_{l-1}(\theta)$ and $S_l(\theta)$ 
using the time-averaged estimator in \eqref{eqn:unbiased_discretized_score} based on the pair of CPF chains 
on level $l-1$ and $l$, respectively.
We can then obtain an unbiased estimator of the score increment $I_l(\theta)$ using the difference 
$\widehat{I}_l(\theta)=\widehat{S}_l(\theta) - \widehat{S}_{l-1}(\theta)$, 
which has finite variance and finite expected cost under the assumptions in Section~\ref{sec:theory}. 
The cost of computing $\widehat{I}_l(\theta)$ is\footnote{This assumes that 
$\mathrm{Cost}(M_{\theta}^{l-1,l})=\mathrm{Cost}(M_{\theta}^{l-1}) + \mathrm{Cost}(M_{\theta}^{l})$, 
$\mathrm{Cost}(\bar{M}_{\theta}^{l-1,l})=2\mathrm{Cost}(M_{\theta}^{l-1}) + 2\mathrm{Cost}(M_{\theta}^{l})$ 
and $\mathrm{Cost}(M_{\theta}^{l})=2\mathrm{Cost}(M_{\theta}^{l-1})$.} 
$\max\{2\tau_{\theta}^{l-1}-1,I+\tau_{\theta}^{l-1}-1\}/2 + \max\{2\tau_{\theta}^l-1,I+\tau_{\theta}^l-1\}$ applications of the CPF kernel $M_{\theta}^l$. 

% alternative 4 way max (might just remove this)
% use common random variables in residuals (already in previous subsection)

\begin{algorithm}
\protect\caption{Four coupled CPF (4-CCPF) at parameter $\theta\in\Theta$ and discretization levels $l-1$ and $l\in\mathbb{N}$.\label{alg:4-CCPF}}

\textbf{Input}: a pair of trajectories $(X_{0:T}^{l-1,\star},\bar{X}_{0:T}^{l-1,\star})=
(X_{s_{2k}}^{l-1,\star},\bar{X}_{s_{2k}}^{l-1,\star})_{k=0}^{K_{l-1}}\in\mathsf{Z}^{l-1}$ 
and a pair of trajectories $(X_{0:T}^{l,\star},\bar{X}_{0:T}^{l,\star})=(X_{s_k}^{l,\star},\bar{X}_{s_k}^{l,\star})_{k=0}^{K_{l}}\in\mathsf{Z}^{l}$.

For time step $k=0$
\begin{description}[itemsep=0pt,parsep=0pt,topsep=0pt,labelindent=0.5cm]
	\item (1a) Set $X_{0}^{l-1,n}=x_{\star},\bar{X}_{0}^{l-1,n}=x_{\star}$ and 
	$X_{0}^{l,n}=x_{\star},\bar{X}_{0}^{l,n}=x_{\star}$ for $n\in\{1,\ldots,N\}$.

	\item (1b) Set $w_0^{l-1,n}=N^{-1},\bar{w}_0^{l-1,n}=N^{-1}$ and 
	$w_0^{l,n}=N^{-1},\bar{w}_0^{l,n}=N^{-1}$ for $n\in\{1,\ldots,N\}$.
	\item (1c) Set $A_{0}^{l-1,n}=n, \bar{A}_{0}^{l-1,n}=n$ and 
	$A_{0}^{l,n}=n,\bar{A}_{0}^{l,n}=n$ for $n\in\{1,\ldots,N\}$.
\end{description}
\textcompwordmark{}

For time step $k\in\{1,\ldots,K_l\}$ 
\begin{description}[itemsep=0pt,parsep=0pt,topsep=0pt,labelindent=0.5cm]
	\item (2a) Sample Brownian increment $V_{s_k}^{l,n} \sim \mathcal{N}_d(0_d,\Delta_{l}I_d)$ at level $l$ independently for $n\in\{1,\ldots,N-1\}$.

	\item (2b) Set $X_{s_k}^{l,n} = F_{\theta}^l(X_{s_{k-1}}^{l,A_{s_{k-1}}^{l,n}},V_{s_k}^{l,n})$ and 
	$\bar{X}_{s_k}^{l,n} = F_{\theta}^l(\bar{X}_{s_{k-1}}^{l,\bar{A}_{s_{k-1}}^{l,n}}, V_{s_k}^{l,n})$ 
	at level $l$ for $n\in\{1,\ldots,N-1\}$.

	\item (2c) Set $X_{s_k}^{l,N} = X_{s_k}^{l,\star}$ and $\bar{X}_{s_k}^{l,N} = \bar{X}_{s_k}^{l,\star}$.
	
	\item If $k\in\{2, 4, \ldots, K_l\}$
	\begin{description}[itemsep=0pt,parsep=0pt,topsep=0pt,labelindent=0.5cm]
		\item (2d) Set Brownian increment $V_{s_k}^{l-1,n} =  V_{s_{k-1}}^{l,n} + V_{s_k}^{l,n}$ at level $l-1$ for $n\in\{1,\ldots,N-1\}$.
		
		\item (2e) Set $X_{s_k}^{l-1,n} = F_{\theta}^{l-1}(X_{s_{k-1}}^{l-1,A_{s_{k-1}}^{l-1,n}}, V_{s_k}^{l-1,n})$ and 
		$\bar{X}_{s_k}^{l-1,n} = F_{\theta}^{l-1}(\bar{X}_{s_{k-1}}^{l-1,\bar{A}_{s_{k-1}}^{l-1,n}}, V_{s_k}^{l-1,n})$ 
		at level $l-1$ for $n\in\{1,\ldots,N-1\}$. 

		\item (2f) Set $X_{s_k}^{l-1,N} = X_{s_k}^{l-1,\star}$ and $\bar{X}_{s_k}^{l-1,N} = \bar{X}_{s_k}^{l-1,\star}$.
	\end{description}
	
	\item If there is an observation at time $t=s_k\in\{1,\ldots,T\}$ 
	\begin{description}[itemsep=0pt,parsep=0pt,topsep=0pt,labelindent=0.5cm]
		\item (2g) Compute normalized weights 
		$w_t^{l-1,n} \propto g_{\theta}(y_{t}|X_{t}^{l-1,n}), 
		\bar{w}_t^{l-1,n} \propto g_{\theta}(y_{t}|\bar{X}_{t}^{l-1,n})$ and 
		$w_t^{l,n} \propto g_{\theta}(y_{t}|X_{t}^{l,n}), 
		\bar{w}_t^{l,n} \propto g_{\theta}(y_{t}|\bar{X}_{t}^{l,n})$ for $n\in\{1,\ldots,N\}$.

		\item (2h) If $t < T$, sample ancestors $(A_{t}^{l-1,n},\bar{A}_{t}^{l-1,n},A_{t}^{l,n},\bar{A}_{t}^{l,n})
		\sim\bar{\mathcal{R}}(w_{t}^{l-1,1:N},\bar{w}_{t}^{l-1,1:N},w_{t}^{l,1:N},\bar{w}_{t}^{l,1:N})$ independently 
		for $n\in\{1,\ldots,N-1\}$ and set $A_{t}^{l-1,N}=N,\bar{A}_{t}^{l-1,N}=N,A_{t}^{l,N}=N,\bar{A}_{t}^{l,N}=N$.
	\end{description}

	\item Else 
	\begin{description}[itemsep=0pt,parsep=0pt,topsep=0pt,labelindent=0.5cm]
		\item (2i) If $k\in\{2, 4, \ldots, K_l\}$, set $A_{s_k}^{l-1,n}=n$ and $\bar{A}_{s_k}^{l-1,n}=n$ at level $l-1$ for $n\in\{1,\ldots,N\}$.
		
		\item (2j) Set $A_{s_k}^{l,n}=n$ and $\bar{A}_{s_k}^{l,n}=n$ at level $l$ for $n\in\{1,\ldots,N\}$.
	\end{description}
\end{description}

\textcompwordmark{}

After the terminal step
\begin{description}[itemsep=0pt,parsep=0pt,topsep=0pt,labelindent=0.5cm]
	\item (3a) Sample particle indexes $(B_{T}^{l-1},\bar{B}_{T}^{l-1},B_{T}^{l},\bar{B}_{T}^{l})\sim 
	\bar{\mathcal{R}}(w_{T}^{l-1,1:N},\bar{w}_{T}^{l-1,1:N}, 
	w_{T}^{l,1:N},\bar{w}_{T}^{l,1:N})$. 
	
	\item (3b) Set particle indexes $B_{s_{2k}}^{l-1} = A_{s_{2k}}^{l-1,B_{s_{2(k+1)}}^{l-1}}$ and 
	$\bar{B}_{s_{2k}}^{l-1} = \bar{A}_{s_{2k}}^{l-1,\bar{B}_{s_{2(k+1)}}^{l-1}}$ 
	at level $l-1$ for $k\in\{0,1,\ldots,K_{l-1}-1\}$. 

	\item (3c) Set particle indexes $B_{s_k}^l = A_{s_k}^{l,B_{s_{k+1}}^l}$ and $\bar{B}_{s_k}^l = \bar{A}_{s_k}^{l,\bar{B}_{s_{k+1}}^l}$ at level $l$ for $k\in\{0,1,\ldots,K_l-1\}$.

\end{description}

\textbf{Output}: a pair of trajectories $(X_{0:T}^{l-1,\circ},\bar{X}_{0:T}^{l-1,\circ})=
(X_{s_{2k}}^{l-1,B_{s_{2k}}^{l-1}},\bar{X}_{s_{2k}}^{l-1,\bar{B}_{s_{2k}}^{l-1}})_{k=0}^{K_{l-1}}\in\mathsf{Z}^{l-1}$  
and a pair of trajectories $(X_{0:T}^{l,\circ},\bar{X}_{0:T}^{l,\circ})=(X_{s_k}^{l,B_{s_k}^l},\bar{X}_{s_k}^{l,\bar{B}_{s_k}^l})_{k=0}^{K_l}\in\mathsf{Z}^{l}$.
\end{algorithm}

\begin{algorithm}
\protect\caption{Maximal coupling of the maximal couplings $\bar{\mathcal{R}}(w^{l-1,1:N},w^{l,1:N})$ 
and $\bar{\mathcal{R}}(\bar{w}^{l-1,1:N},\bar{w}^{l,1:N})$\label{alg:maximal-maximal}}

\textbf{Input}: normalized weights $w^{l-1,1:N}=(w^{l-1,n})_{n=1}^N$, $\bar{w}^{l-1,1:N}=(\bar{w}^{l-1,n})_{n=1}^N$ at level $l-1$ and  
$w^{l,1:N}=(w^{l,n})_{n=1}^N$, $\bar{w}^{l,1:N}=(\bar{w}^{l,n})_{n=1}^N$ at level $l$.

(1) Sample $(A^{l-1},A^l)\sim \bar{\mathcal{R}}(w^{l-1,1:N},w^{l,1:N})$ using Algorithm \ref{alg:2maximal}.
 
If normalized weights at level $l-1$ are identical and normalized weights at level $l$ are non-identical
\begin{description}[itemsep=0pt,parsep=0pt,topsep=0pt,labelindent=0.5cm]
	\item (2a) Set $\bar{A}^{l-1}=A^{l-1}$.
	
	\item (2b) With probability $\bar{w}^{l,\bar{A}^{l-1}}/\bar{w}^{l-1,\bar{A}^{l-1}}$, set $\bar{A}^l=\bar{A}^{l-1}$; 
	otherwise sample $A\sim\mathcal{R}(\bar{w}^{l,1:N})$ and $U\sim\mathcal{U}_{[0,1]}$ 
	until $U>\bar{w}^{l-1,A}/\bar{w}^{l,A}$, and set $\bar{A}^l=A$. 
\end{description}

If normalized weights at level $l-1$ are non-identical and normalized weights at level $l$ are identical
\begin{description}[itemsep=0pt,parsep=0pt,topsep=0pt,labelindent=0.5cm]
	\item (3a) Set $\bar{A}^{l}=A^{l}$.
	
	\item (3b) With probability $\bar{w}^{l-1,\bar{A}^{l}}/\bar{w}^{l,\bar{A}^{l}}$, set $\bar{A}^{l-1}=\bar{A}^{l}$; 
	otherwise sample $A\sim\mathcal{R}(\bar{w}^{l-1,1:N})$ and $U\sim\mathcal{U}_{[0,1]}$ 
	until $U>\bar{w}^{l,A}/\bar{w}^{l-1,A}$, and set $\bar{A}^{l-1}=A$. 
\end{description}

Otherwise
\begin{description}[itemsep=0pt,parsep=0pt,topsep=0pt,labelindent=0.5cm]
	\item (4) With probability $\bar{R}^{l-1,l}(A^{l-1},A^l)/R^{l-1,l}(A^{l-1},A^l)$, set $(\bar{A}^{l-1},\bar{A}^l)=(A^{l-1},A^l)$; 
	otherwise sample $(A,B)\sim \bar{\mathcal{R}}(\bar{w}^{l-1,1:N},\bar{w}^{l,1:N})$ and $U\sim\mathcal{U}_{[0,1]}$ 
	until $U>R^{l-1,l}(A,B)/\bar{R}^{l-1,l}(A,B)$, and set $(\bar{A}^{l-1},\bar{A}^l)=(A,B)$.
	
\end{description}

\textbf{Output}: indexes $(A^{l-1},\bar{A}^{l-1},A^{l},\bar{A}^{l})$.
\end{algorithm}

\begin{algorithm}
\protect\caption{Multilevel CPF (ML-CPF) at parameter $\theta\in\Theta$ and discretization levels $l-1$ and $l\in\mathbb{N}$.\label{alg:ML-CPF}}

\textbf{Input}: a trajectory $X_{0:T}^{l-1,\star}=
(X_{s_{2k}}^{l-1,\star})_{k=0}^{K_{l-1}}\in\mathsf{X}^{l-1}$ 
and a trajectory $X_{0:T}^{l,\star}=(X_{s_k}^{l,\star})_{k=0}^{K_{l}}\in\mathsf{X}^{l}$.

For time step $k=0$
\begin{description}[itemsep=0pt,parsep=0pt,topsep=0pt,labelindent=0.5cm]
	\item (1a) Set $X_{0}^{l-1,n}=x_{\star}$ and 
	$X_{0}^{l,n}=x_{\star}$ for $n\in\{1,\ldots,N\}$.

	\item (1b) Set $w_0^{l-1,n}=N^{-1}, w_0^{l,n}=N^{-1}$ and 
	$A_{0}^{l-1,n}=n, A_{0}^{l,n}=n$ for $n\in\{1,\ldots,N\}$. 
\end{description}
\textcompwordmark{}

For time step $k\in\{1,\ldots,K_l\}$ 
\begin{description}[itemsep=0pt,parsep=0pt,topsep=0pt,labelindent=0.5cm]
	\item (2a) Sample Brownian increment $V_{s_k}^{l,n} \sim \mathcal{N}_d(0_d,\Delta_{l}I_d)$ at level $l$ independently for $n\in\{1,\ldots,N-1\}$.

	\item (2b) Set $X_{s_k}^{l,n} = F_{\theta}^l(X_{s_{k-1}}^{l,A_{s_{k-1}}^{l,n}},V_{s_k}^{l,n})$ 
	at level $l$ for $n\in\{1,\ldots,N-1\}$, and $X_{s_k}^{l,N} = X_{s_k}^{l,\star}$.  
	
	\item If $k\in\{2, 4, \ldots, K_l\}$
	\begin{description}[itemsep=0pt,parsep=0pt,topsep=0pt,labelindent=0.5cm]
		\item (2c) Set Brownian increment $V_{s_k}^{l-1,n} =  V_{s_{k-1}}^{l,n} + V_{s_k}^{l,n}$ at level $l-1$ for $n\in\{1,\ldots,N-1\}$.
		
		\item (2d) Set $X_{s_k}^{l-1,n} = F_{\theta}^{l-1}(X_{s_{k-1}}^{l-1,A_{s_{k-1}}^{l-1,n}}, V_{s_k}^{l-1,n})$ 
		at level $l-1$ for $n\in\{1,\ldots,N-1\}$, and $X_{s_k}^{l-1,N} = X_{s_k}^{l-1,\star}$. 

	\end{description}

	\item If there is an observation at time $t=s_k\in\{1,\ldots,T\}$ 
	\begin{description}[itemsep=0pt,parsep=0pt,topsep=0pt,labelindent=0.5cm]
		\item (2e) Compute normalized weights 
		$w_t^{l-1,n} \propto g_{\theta}(y_{t}|X_{t}^{l-1,n})$ and 
		$w_t^{l,n} \propto g_{\theta}(y_{t}|X_{t}^{l,n})$ for $n\in\{1,\ldots,N\}$.

		\item (2f) If $t < T$, sample ancestors $(A_{t}^{l-1,n},A_{t}^{l,n})
		\sim\bar{\mathcal{R}}(w_{t}^{l-1,1:N},w_{t}^{l,1:N})$ independently 
		for $n\in\{1,\ldots,N-1\}$ and set $A_{t}^{l-1,N}=N,A_{t}^{l,N}=N$.
	\end{description}

	\item Else 
	\begin{description}[itemsep=0pt,parsep=0pt,topsep=0pt,labelindent=0.5cm]
		\item (2i) If $k\in\{2, 4, \ldots, K_l\}$, set $A_{s_k}^{l-1,n}=n$ at level $l-1$ for $n\in\{1,\ldots,N\}$.
		
		\item (2j) Set $A_{s_k}^{l,n}=n$ at level $l$ for $n\in\{1,\ldots,N\}$.
	\end{description}
\end{description}

\textcompwordmark{}

After the terminal step
\begin{description}[itemsep=0pt,parsep=0pt,topsep=0pt,labelindent=0.5cm]
	\item (3a) Sample particle indexes $(B_{T}^{l-1},B_{T}^{l})\sim \bar{\mathcal{R}}(w_{T}^{l-1,1:N}, w_{T}^{l,1:N})$. 
	
	\item (3b) Set particle indexes $B_{s_{2k}}^{l-1} = A_{s_{2k}}^{l-1,B_{s_{2(k+1)}}^{l-1}}$ 
	at level $l-1$ for $k\in\{0,1,\ldots,K_{l-1}-1\}$. 

	\item (3c) Set particle indexes $B_{s_k}^l = A_{s_k}^{l,B_{s_{k+1}}^l}$ at level $l$ for $k\in\{0,1,\ldots,K_l-1\}$.

\end{description}

\textbf{Output}: a trajectory $X_{0:T}^{l-1,\circ}=(X_{s_{2k}}^{l-1,B_{s_{2k}}^{l-1}})_{k=0}^{K_{l-1}}\in\mathsf{X}^{l-1}$  
and a trajectory $X_{0:T}^{l,\circ}=(X_{s_k}^{l,B_{s_k}^l})_{k=0}^{K_l}\in\mathsf{X}^{l}$.
\end{algorithm}

\subsection{Summary of proposed methodology and choice of tuning parameters}\label{sec:summary_proposed_methodology}

We consolidate the algorithms presented in this section by 
summarizing our proposed methodology to unbiasedly estimate the score function $S(\theta)$ below. 
\textcompwordmark{}
\begin{flushleft}
\textbf{Input}: PMF $(P_l)_{l=0}^{\infty}$, number of particles $N$, burn-in $b$ and number of iterations $I$. 

\begin{enumerate}[(1),itemsep=0pt,parsep=0pt,topsep=0pt]

	\item Sample highest discretization level $L$ from $(P_l)_{l=0}^{\infty}$. 

	\item Simulate coupled CPF chains $(X_{0:T}(i),\bar{X}_{0:T}(i))$ using CPF (Algorithm \ref{alg:CPF}) and 
	2-CCPF (Algorithm \ref{alg:2-CCPF}) with $N$ particles at discretization level $0$ as described in \eqref{eqn:coupled2_chains} 
	for iteration $i=0,1,\ldots,\max\{I, \tau_{\theta}^0\}$, and 
	compute unbiased estimator $\widehat{I}_0(\theta)$ of ${I}_0(\theta)$ using time-averaged estimator 
	\eqref{eqn:unbiased_discretized_score} with burn-in $b$ and $I$ iterations.
	
	\item Independently for $l = 1,\ldots,L$, simulate two pairs of coupled CPF chains $(X_{0:T}^{l-1}(i),\bar{X}_{0:T}^{l-1}(i))$ 
	and $(X_{0:T}^{l}(i),\bar{X}_{0:T}^{l}(i))$ using ML-CPF (Algorithm \ref{alg:ML-CPF}) and 
	4-CCPF (Algorithm \ref{alg:4-CCPF}) with $N$ particles at discretization levels $l-1$ and $l$ as described in \eqref{eqn:coupled4_chains} 
	for iteration $i=0,1,\ldots,\max\{I, \bar{\tau}_{\theta}^l\}$, and 
	compute unbiased estimators $\widehat{S}_{l-1}(\theta)$ and $\widehat{S}_l(\theta)$ of the scores 
	$S_{l-1}(\theta)$ and $S_l(\theta)$ using time-averaged estimator 
	\eqref{eqn:unbiased_discretized_score} with burn-in $b$ and $I$ iterations. 
	Compute unbiased estimator of the score increment $I_l(\theta)$ with 
	$\widehat{I}_l(\theta)=\widehat{S}_l(\theta) - \widehat{S}_{l-1}(\theta)$. 
	
	\item Compute unbiased estimator of the score function $S(\theta)$ using 
	independent sum estimator $\widehat{S}(\theta) = \sum_{l=0}^L \widehat{I}_l(\theta)/\mathcal{P}_l$, 
	where $\mathcal{P}_l = \sum_{k=l}^{\infty} P_k$.
	
\end{enumerate}

\textbf{Output}: unbiased estimator $\widehat{S}(\theta)$ of the score function $S(\theta)$. 
\end{flushleft}

We will establish unbiasedness and finite variance properties of $\widehat{S}(\theta)$ in Section~\ref{sec:theory}. 
The cost of the above procedure is $c(\theta)=\sum_{l=0}^Lc_l(\theta)$, where 
$c_l(\theta)$ is the cost of computing $\widehat{I}_l(\theta)$. 
From Sections~\ref{sec:unbiased_discretized_score} and \ref{sec:unbiased_discretized_increment}, we have
$c_l(\theta)=a_\theta^l\times\mathrm{Cost}(M_{\theta}^l)$, with 
$a_\theta^l=\max(2\tau_{\theta}^0-1,I+\tau_{\theta}^0-1)$ for level $l=0$, and 
$a_\theta^l=\max(2\tau_{\theta}^{l-1}-1,I+\tau_{\theta}^{l-1}-1)/2 + \max(2\tau_{\theta}^l-1,I+\tau_{\theta}^l-1)$ 
for level $l\in\mathbb{N}$. 
We take the view here that the cost per application of the CPF kernel $M_{\theta}^l$ is 
$\mathrm{Cost}(M_{\theta}^l)=N K_l = N 2^l T$. Hence the expected cost of computing $\widehat{S}(\theta)$ is 
\begin{align}\label{eqn:expected_cost}
	\mathbb{E}\left[c(\theta)\right] = \sum_{l=0}^{\infty}\mathbb{E}\left[c_l(\theta)\right] \mathcal{P}_l 
	= NT\sum_{l=0}^{\infty} \mathbb{E}\left[a_{\theta}^l\right] 2^l \mathcal{P}_l. 
\end{align}
We will see that one cannot find a PMF $(P_l)_{l=0}^{\infty}$ so that the variance and expected cost of $\widehat{S}(\theta)$ 
are both finite. 
We defer further discussions and the selection of the distribution of $L$ to Section~\ref{sec:theory}, and  
assume for now that we have a given PMF $(P_l)_{l=0}^{\infty}$ that ensures finite variance 
but infinite expected cost. In this regime, we will refrain from discussions of asymptotic efficiency 
in the sense of \citet{glynn1992asymptotic}. 

We now discuss the choice of tuning parameters and various algorithmic considerations. 
In the above description, the choice of $(N, b, I)$ could be level-dependent but optimizing 
these tuning parameters is outside the scope of this work. 
Following the discussion in \citet[Theorem 1]{andrieu2010particle} and the empirical findings in \citet{jacob2020smoothing},  
we will scale the number of particles $N$ linearly with the number of observations $T$. 
Although the variance of $\widehat{I}_l(\theta)$ decreases as we increase the burn-in 
parameter $b\in\mathbb{N}_0$, setting $b$ too large would be inefficient. 
\citet{jacob2020smoothing,jacob2020unbiased} proposed choosing $b$ according to the distribution of the meeting time. 
In our context, as the stopping time $\bar{\tau}_{\theta}^l$ typically decreases as the level $l$ increases, 
a conservative strategy is to select $b$ based on the stopping time of a low discretization level, 
which can be simulated by running ML-CPF and 4-CCPF as in Step 3. 
We will illustrate this numerically in Section~\ref{sec:examples} and experiment with 
various choices of $b$. 
%Even though the estimators $(\widehat{I}_l(\theta))_{l=0}^{\infty}$ are unbiased for any choice of $b$, 
%the variance of $\widehat{S}(\theta)$ might be infinite if $b$ is too small. 
%This is due to the nested nature of our unbiased estimation strategy, i.e. 
After selecting $b$, one can choose the number of iterations $I\geq b$ to further reduce the variance of 
$\widehat{I}_l(\theta)$, and hence that of $\widehat{S}(\theta)$, at a cost \eqref{eqn:expected_cost} 
that grows linearly with $I$. 
On the other hand, when employing score estimators within a stochastic gradient method, 
taking large values of $I$ to obtain low variance gradient estimators would be inefficient. 
Choosing the tuning parameters $(N, b, I)$ to maximize the efficiency of the resulting 
stochastic gradient method is a highly non-trivial problem, and could be the topic of future work. 

Given a choice of $(N, b, I)$, one can produce $R\in\mathbb{N}$ independent replicates  
$\widehat{S}(\theta)_r = \sum_{l=0}^{L_r} \widehat{I}_l(\theta)_r/\mathcal{P}_l, r\in\{1,\ldots,R\}$,
of $\widehat{S}(\theta)$ in parallel, 
and compute the average $\bar{S}(\theta)=R^{-1}\sum_{r=1}^R\widehat{S}(\theta)_r$ to approximate $S(\theta)$. 
To see the connection to MLMC, we follow \citet[Example 3]{vihola2018unbiased} by noting that $\bar{S}(\theta)$ 
has the same distribution as the random variable
\begin{align}\label{eqn:MLMC_connection}
	\sum_{l=0}^{\infty}\frac{1}{\mathbb{E}\left[R_l\right]}\sum_{r=1}^{R_l}\widehat{I}_l(\theta)_r,
\end{align}
where $R_l=\sum_{r=1}^R\mathbb{I}(L_r\geq l)$ has expectation $\mathbb{E}[R_l] = R \mathcal{P}_l$. 
\citet{vihola2018unbiased} proposed new unbiased estimators with lower variance than $\bar{S}(\theta)$ 
by sampling the random variables $(L_r)_{r=1}^R$ that define $(R_l)_{l=0}^{\infty}$ in \eqref{eqn:MLMC_connection} using stratification. 
As the number of replicates $R\rightarrow\infty$, these improved estimators also attain the same limiting variance 
as the idealized MLMC estimator 
\begin{align}
	\widetilde{S}(\theta) = \sum_{l=0}^{\infty}\frac{1}{\tilde{R}_l}\sum_{r=1}^{\tilde{R}_l}\widehat{I}_l(\theta)_r,
\end{align}
where $\tilde{R}_l = \lfloor R\mathcal{P}_l\rfloor$ is allocated using the PMF $(P_l)_{l=0}^{\infty}$.
From \citet[Thoerem 5]{vihola2018unbiased}, this asymptotic variance is given by 
$\lim_{R\rightarrow\infty}R\mathrm{Var}[\widetilde{S}(\theta)^j]=
\sum_{l=0}^{\infty} \mathrm{Var}[\widehat{I}_l(\theta)^j] /\mathcal{P}_l
<\mathrm{Var}[\widehat{S}(\theta)]$, 
for all $j\in\{1,\ldots,d_{\theta}\}$.

% scaling N = O(T) seems adequate based on simulation results in Jacob et al.; empirically, this just needs to be large enough
% (b, I): we need b large enough, heuristics for b and I to maximizing efficiency
% averaging over repeats should not be done independently - use Matti's estimators
% independent repeats (parallel computing considerations) and connection to MLMC (Matti's paper)

%we will experiment with various choices of $(N, b, I)$ 
%to maximize efficiency. 

\section{Analysis}\label{sec:theory}
This section is concerned with the theoretical validity of our approach. We first introduce some notation 
needed to state the assumptions which we will rely on. 
Let $(\mathsf{E},\mathcal{E})$ be an arbitrary measurable space. 
We write $\mathcal{B}_b(\mathsf{E})$ as the collection of real-valued, bounded and measurable functions on $\mathsf{E}$.
For real-valued $\varphi:\mathsf{E}\rightarrow\mathbb{R}$ (and vector-valued $\varphi:\mathsf{E}\rightarrow\mathbb{R}^d$),
let $\mathcal{C}^j(\mathsf{E})$ (and $\mathcal{C}_d^j(\mathsf{E})$) denote the collection of $j\in\mathbb{N}$ times continuously differentiable functions,  
and $\mathcal{C}(\mathsf{E})$ (and $\mathcal{C}_d(\mathsf{E})$) for the collection of continuous functions. 
We write $\varphi\in\textrm{Lip}_{\|\cdot\|_2}(\mathbb{R}^{d})$ if the real-valued function 
$\varphi:\mathbb{R}^d\rightarrow\mathbb{R}$ is Lipschitz w.r.t.\ the $\mathbb{L}_2$-norm $\|\cdot\|_2$, i.e.\ 
if there exists a constant $C<\infty$ such that $|\varphi(x)-\varphi(y)| \leq C\|x-y\|_2$ for all $x,y\in\mathbb{R}^{d}$, 
and $\|\varphi\|_{\textrm{Lip}}$ as the Lipschitz constant. 

We recall the definitions of $\Sigma(x)=\sigma(x)\sigma(x)^*$ and 
$b_\theta(x) = \Sigma(x)^{-1}\sigma(x)^*a_{\theta}(x)$ for $x\in\mathbb{R}^d$ 
as they appear in the following.

\begin{assumption}\label{ass:D1}
The drift function $a:\Theta\times\mathbb{R}^d \rightarrow\mathbb{R}^d$ and diffusion coefficient 
$\sigma:\mathbb{R}^d \rightarrow\mathbb{R}^{d\times d}$ satisfy:
\begin{enumerate}[(i)] 

	\item (Smoothness) For any $\theta\in\Theta$, $a_{\theta}^j\in \mathcal{C}^2(\mathbb{R}^d)$ 
	for all $j\in\{1,\ldots,d\}$ components of $a_{\theta}$, and $\sigma^{j,k} \in \mathcal{C}^2(\mathbb{R}^d)$ 
	for all $(j,k)\in\{1,\ldots, d\}$ components of $\sigma$. 
	Also, for any $x\in\mathbb{R}^d$, $\theta\mapsto a_{\theta}^j(x)\in\mathcal{C}(\Theta)$ for all $j\in\{1,\dots,d\}$. 
	
	\item (Uniform ellipticity) $\Sigma(x)$ is uniformly positive definite for all $x\in \mathbb{R}^d$.
	
	\item (Globally Lipschitz) For any $\theta\in\Theta$, there exists a constant $C<\infty$ such that 
	$|a_{\theta}^j(x)-a_{\theta}^j(x')|+|\sigma^{j,k}(x)-\sigma^{j,k}(x')| \leq C \|x-x'\|_2$ for all 
	$(x,x') \in \mathbb{R}^d\times\mathbb{R}^d$ and $(j,k)\in\{1,\dots,d\}^2$. 
\end{enumerate}
\end{assumption}

\begin{assumption}\label{ass:D2}
The drift function $a:\Theta\times\mathbb{R}^d \rightarrow\mathbb{R}^d$, diffusion coefficient 
$\sigma:\mathbb{R}^d \rightarrow\mathbb{R}^{d\times d}$ and 
observation density $g:\Theta\times\mathbb{R}^d\times\mathbb{R}^{d_y}\rightarrow\mathbb{R}^+$ satisfy:

\begin{enumerate}[(i)] 
	\item the inverse of $x\mapsto\Sigma(x)$ satisfies
	$[\Sigma^{-1}]^{j,k}\in\mathcal{B}_b(\mathbb{R}^d)\cap\textrm{\emph{Lip}}_{\|\cdot\|_2}(\mathbb{R}^d)$ for all $(j,k)\in\{1,\dots,d\}^2$.
	
	\item For any $\theta\in\Theta$, $a_\theta^j\in\mathcal{B}_b(\mathbb{R}^d)$ for all $j\in\{1,\dots,d\}$, and  
	$\sigma^{j,k}\in\mathcal{B}_b(\mathbb{R}^d)$ for all $(j,k)\in\{1,\dots,d\}^2$. 
	
	\item For any $\theta\in\Theta$, there exists $0<\underline{C}<\overline{C}<\infty$ such that 
	$\underline{C}\leq g_{\theta}(y|x)\leq \overline{C}$ for all $(x,y)\in \mathbb{R}^d\times\mathbb{R}^{d_y}$. 
	In addition, for any $(\theta,y)\in\Theta\times\mathbb{R}^{d_y}$, we have $g_{\theta}(y|\cdot)\in\textrm{\emph{Lip}}_{\|\cdot\|_2}(\mathbb{R}^d)$.
	
	\item For any $(\theta,y)\in\Theta\times\mathbb{R}^{d_y}$, $[\nabla_{\theta}\log g_{\theta}(y|\cdot)]^j
	\in\mathcal{B}_b(\mathbb{R}^d)\cap\textrm{\emph{Lip}}_{\|\cdot\|_2}(\mathbb{R}^d)$ for all $j\in\{1,\dots,d_{\theta}\}$.

	\item For any $\theta\in\Theta$, $[\nabla_{\theta} [b_{\theta}]^j]^k, [\nabla_{\theta} (b_{\theta}^j)^2]^k
	\in\mathcal{B}_b(\mathbb{R}^d)\cap\textrm{\emph{Lip}}_{\|\cdot\|_2}(\mathbb{R}^d)$ for all 
	$(j,k)\in\{1,\dots,d\}\times\{1,\dots,d_{\theta}\}$.	
\end{enumerate}	
\end{assumption}

Assumptions \ref{ass:D1} and \ref{ass:D2} should be understood as sufficient conditions 
to verify the validity of our proposed methodology, and are not necessary for its implementation.  
Although some of these assumptions are strong, they have been adopted to simplify the exposition 
of our analysis, as is common in theoretical works on particle filtering.  
Some assumptions can be relaxed at the expense of more involved and lengthy 
technical arguments. 

We first give an intermediate result on the convergence of the time-discretized score functions $S_l(\theta)$ 
defined in \eqref{eqn:discretized_score}.

\begin{theorem}\label{prop:conv_grad_log_like}
Under Assumptions \ref{ass:D1} and \ref{ass:D2}, for any $(T,\theta)\in\mathbb{N}\times\Theta$, there exists a constant $C<\infty$ such that 
for any $l\in\mathbb{N}_0$, $\|S_l(\theta) - S(\theta) \|_1 \leq C\Delta_l^{1/2}$. 
\end{theorem}

We note that the upper-bound can be sharpened to $\mathcal{O}(\Delta_l)$. 
The proof which involves studying the time-discretization of diffusions 
can be found in Appendix \ref{app:diff_proc}. 
The following is our main result that establishes basic properties of our score estimator. 

\begin{theorem}\label{theo:ub}
Under Assumptions \ref{ass:D1} and \ref{ass:D2}, for any number of particles $N\geq 2$, 
burn-in $b\in\mathbb{N}_0$ and number of iterations $I\geq b$, 
there exists a choice of PMF $(P_l)_{l=0}^{\infty}$ such that for any $\theta\in\Theta$, 
the score estimator $\widehat{S}(\theta)$ in \eqref{eq:ub3} is unbiased and has finite variance. 
\end{theorem}

The proof of Theorem \ref{theo:ub} in Appendix \ref{app:CCPF} involves analyzing 
various aspects of the 4-CCPF chains (Appendix \ref{sec:analysis_first} \& \ref{sec:analysis_second}) 
and its initialization (Appendix \ref{sec:analysis_third}), followed by the 
properties of our score estimators that are constructed using these coupled Markov chains (Appendix \ref{sec:analysis_fourth}). 
%It follows from the proof of Theorem \ref{theo:ub} in Appendix \ref{app:CCPF} 
It follows from the proof that the left-hand side of \eqref{eq:ub2} is upper-bounded by 
\begin{align}\label{eqn:upperbound_criterion}
	C(\theta)\sum_{l=0}^{\infty}\mathcal{P}_l^{-1}\Delta_l^{2\phi}, 
\end{align}
where $C(\theta)<\infty$ is a parameter-dependent constant and 
$\phi\in(0,1/2)$ is a constant determined in our analysis that does not depend on $l$. 
Hence any choice of PMF $(P_l)_{l=0}^{\infty}$ such that the sum \eqref{eqn:upperbound_criterion} 
is finite would be valid; e.g.\ $P_l\propto \Delta_l^{2\phi\alpha}$ for any $\alpha\in(0,1)$. 
Due to the technical complexity of the problem and algorithms under consideration, 
the rate in \eqref{eqn:upperbound_criterion} is not sharp. 
We conjecture that the correct rate corresponds to having $\phi=1/4$ and 
a better rate of $\phi=1/2$ can be obtained in the case of constant diffusion coefficient $\sigma$. 

Using Lemma \ref{lem:coup_prob} in Appendix \ref{app:CCPF}, we can upper-bound the expected cost \eqref{eqn:expected_cost} 
by 
\begin{align}\label{eqn:upperbound_cost}
	C(\theta,T,N,b,I)NT\sum_{l=0}^{\infty} 2^l \mathcal{P}_l,
\end{align}
where $C(\theta,T,N,b,I)<\infty$ is another constant that is independent of $l$. 
As we are in a setting where $\phi\leq 1/2$, there is no choice of PMF $(P_l)_{l=0}^{\infty}$ 
that can keep both \eqref{eqn:upperbound_criterion} and \eqref{eqn:upperbound_cost} finite 
\citep[Section~4]{rhee2015unbiased}. 
This is a consequence of employing the Euler--Maruyama discretization scheme \eqref{eq:disc_state} 
and the choice of coupled resampling scheme in Algorithm \ref{alg:maximal-maximal}; 
future work could consider the antithetic truncated Milstein scheme of \citet{giles2014antithetic} 
and improved coupled resampling schemes such as \citet{ballesio2020wasserstein}. 
Our numerical implementations will follow the approach of \citet{rhee2015unbiased} 
by choosing the PMFs $P_l\propto \Delta_l^{1/2}l(\log_2(1+l))^2$ 
and $P_l\propto \Delta_l l (\log_2(1+l))^2$ in the case of non-constant and constant diffusion coefficients, respectively. 
These choices yield score estimators with finite variance but infinite expected cost, 
which is a sensible compromise when employing these gradients within stochastic gradient methods. 
Moreover, under these choices, the unbiased estimators also achieve 
computational complexities that are similar to MLMC estimators \citep{giles2008multilevel}.

\section{Applications}\label{sec:examples}

\subsection{Ornstein--Uhlenbeck process}\label{sec:ornstein_uhlenbeck}
We consider an Ornstein--Uhlenbeck process $X=(X_t)_{0\leq t\leq T}$ in $\mathbb{R}$, defined by the SDE 
\begin{align}\label{eqn:OU_SDE}
	dX_t = \theta_1(\theta_2-X_t)dt + \sigma dW_t,\quad X_0=0.
\end{align}
The parameter $\theta_{1}> 0$ can be interpreted as the speed of the mean reversion to the 
long-run equilibrium value $\theta_{2}\in\mathbb{R}$. 
This corresponds to \eqref{eq:intro_diff} with initial condition $x_{\star}=0$, linear drift function $a_{\theta}(x)=\theta_1(\theta_2-x)$ 
and constant diffusion coefficient $\sigma(x)=\sigma>0$ for $x\in\mathbb{R}$. 
We assume that the process is observed at unit times with Gaussian measurement errors, 
i.e.\ $Y_t|X\sim g_{\theta}(\cdot|X_t)=\mathcal{N}(X_t,\theta_3)$ independently for $t\in\{1,\ldots,T\}$ and some $\theta_3>0$. 
We will generate observations $y_{1:T}$ by simulating from the model with 
parameter $\theta=(\theta_1,\theta_2,\theta_3)=(2,7,1)$. 
This setting is considered as it is possible to compute the score function \eqref{eqn:score_function} 
exactly using Kalman smoothing; see Appendix \ref{app:OU} for details and 
model-specific expressions to evaluate \eqref{eqn:discrete_test_function}. 

Figure~\ref{fig:ou_stoppingtimes} illustrates how the distribution of the 
stopping time $\bar{\tau}_{\theta}^l$ varies with the discretization level $l$ on a simulated dataset with $T=25$ observations. 
We took $l=3$ as the lowest discretization level as lower levels lead to numerically unstable trajectories. 
As alluded earlier, the coupled resampling scheme proposed in Algorithm \ref{alg:maximal-maximal} (referred to as ``maximal'') 
leads to smaller and more stable stopping times for large enough levels. 
As alternatives, we consider a modification (``other maximal'') that ensures 4-CCPF admits 2-CCPF as marginals on each level 
and a scheme that uses common uniform random variables (``common uniforms''). 
While the schemes based on maximal couplings have comparable stopping times, the approach based on common uniform 
variables gives rise to significantly larger stopping times.
Figure \ref{fig:ou_varincrements} reveals that these two alternative coupled resampling schemes 
do not induce sufficient dependencies between the four CPF chains. 
As the variance of the estimated increment does not decrease 
with the discretization level, this precludes their use within our unbiased score estimation framework. 

To show the impact of the choice of $b$ and $I$, we considered three types of estimators corresponding to 
having $b=0$ and $I=b$ (``naive''); $b=90\%\textrm{-quantile}(\bar{\tau}_{\theta}^l)$ at level $l=3$ and $I=b$ (``simple''); 
and $b=90\%\textrm{-quantile}(\bar{\tau}_{\theta}^l)$ and $I= 10b$ (``time-averaged''), where
$90\%\textrm{-quantile}(\bar{\tau}_{\theta}^l)$ denotes the 90\% sample quantile of the stopping time at level $l$.
The benefits of increasing $b$ and $I$ in terms of variance reduction are consistent with findings in \citet{jacob2020smoothing,jacob2020unbiased}.
Under our proposed coupling, all three choices yield estimators of increments whose variance decrease exponentially with the level, 
which agrees with our theoretical results (see Lemma \ref{lem:mc_fifth_lem} in Appendix \ref{app:CCPF}).
Hence we can employ any of these estimators within the score estimation methodology outlined 
in Section \ref{sec:summary_proposed_methodology}. 
Figure \ref{fig:ou_squarederrors} displays the resulting squared error $\| \widehat{S}(\theta) - S(\theta) \|_2^2$ 
and cost of $100$ independent replicates. This plot suggests having $N=128$ particles is sufficient in the case of $T=25$ observations. 
The choice of $b$ and $I$ also allows a tradeoff between error and cost. 
As we increase the number of observations $T$, Figure \ref{fig:ou_stoppingtimes_nobs} shows it is important 
to scale the number of particles $N$ linearly with $T$ to obtain stable and non-exponentially increasing stopping times. 
Lastly, Figures~\ref{fig:ou_averaging1} and \ref{fig:ou_averaging2} concern the averaging of independent replicates of the score estimator 
$(\widehat{S}(\theta)_r)_{r=1}^R$. Figure~\ref{fig:ou_averaging1} shows that the average $\bar{S}(\theta)=R^{-1}\sum_{r=1}^R\widehat{S}(\theta)_r$ 
satisfies the standard Monte Carlo rate as $R\rightarrow\infty$, which is consistent with its properties in Theorem \ref{theo:ub}, 
at a linear cost in $R$ as illustrated in Figure~\ref{fig:ou_averaging2}.

\begin{figure}
\centering
\begin{subfigure}[t]{0.45\textwidth}
\includegraphics[trim=0pt 32pt 0pt 25pt,clip,width=.95\textwidth]{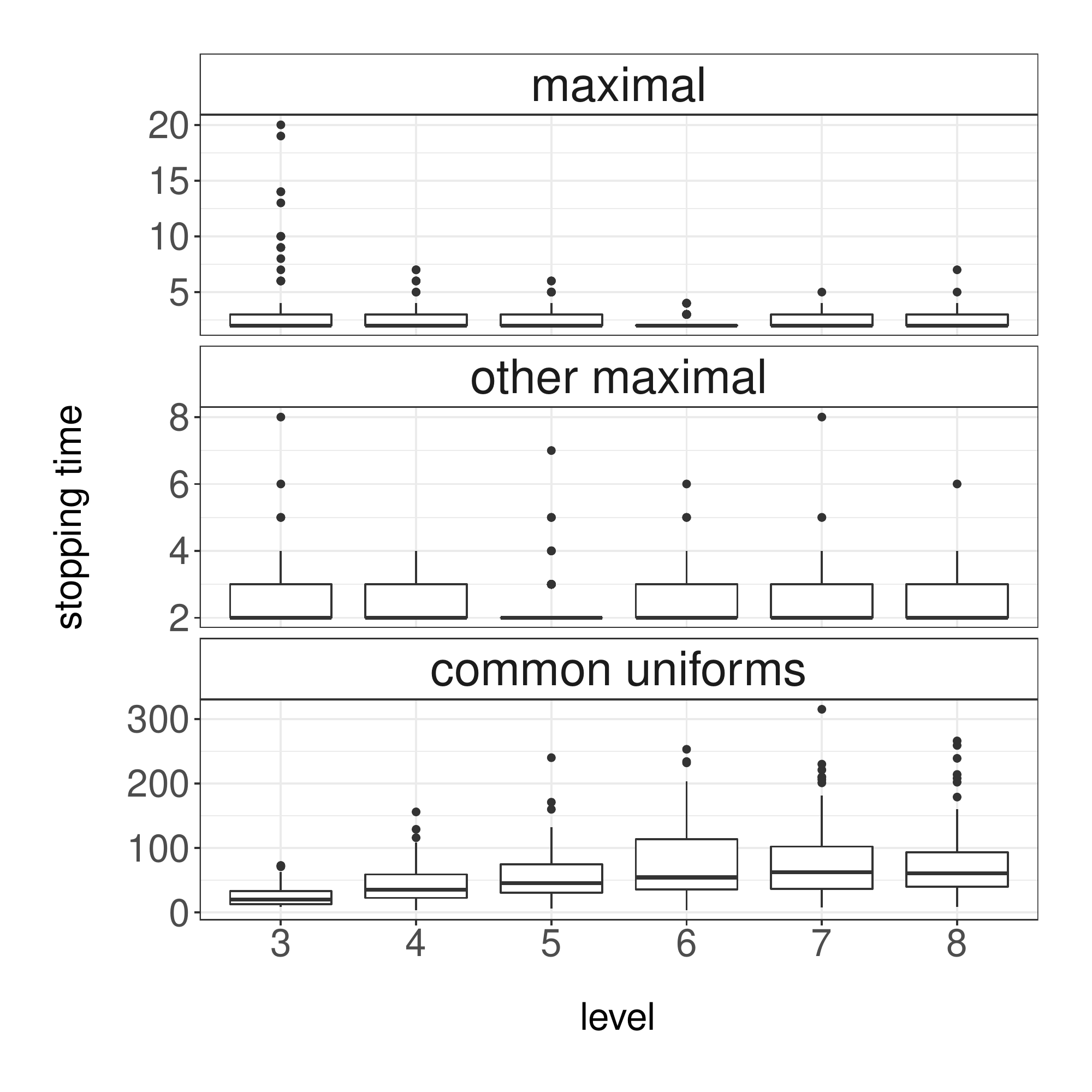}
\caption{Boxplots of stopping time $\bar{\tau}_{\theta}^l$ against discretization level $l$}
\label{fig:ou_stoppingtimes}
\end{subfigure} \quad
\begin{subfigure}[t]{0.45\textwidth}
\includegraphics[trim=0pt 32pt 0pt 25pt,clip,width=.95\textwidth]{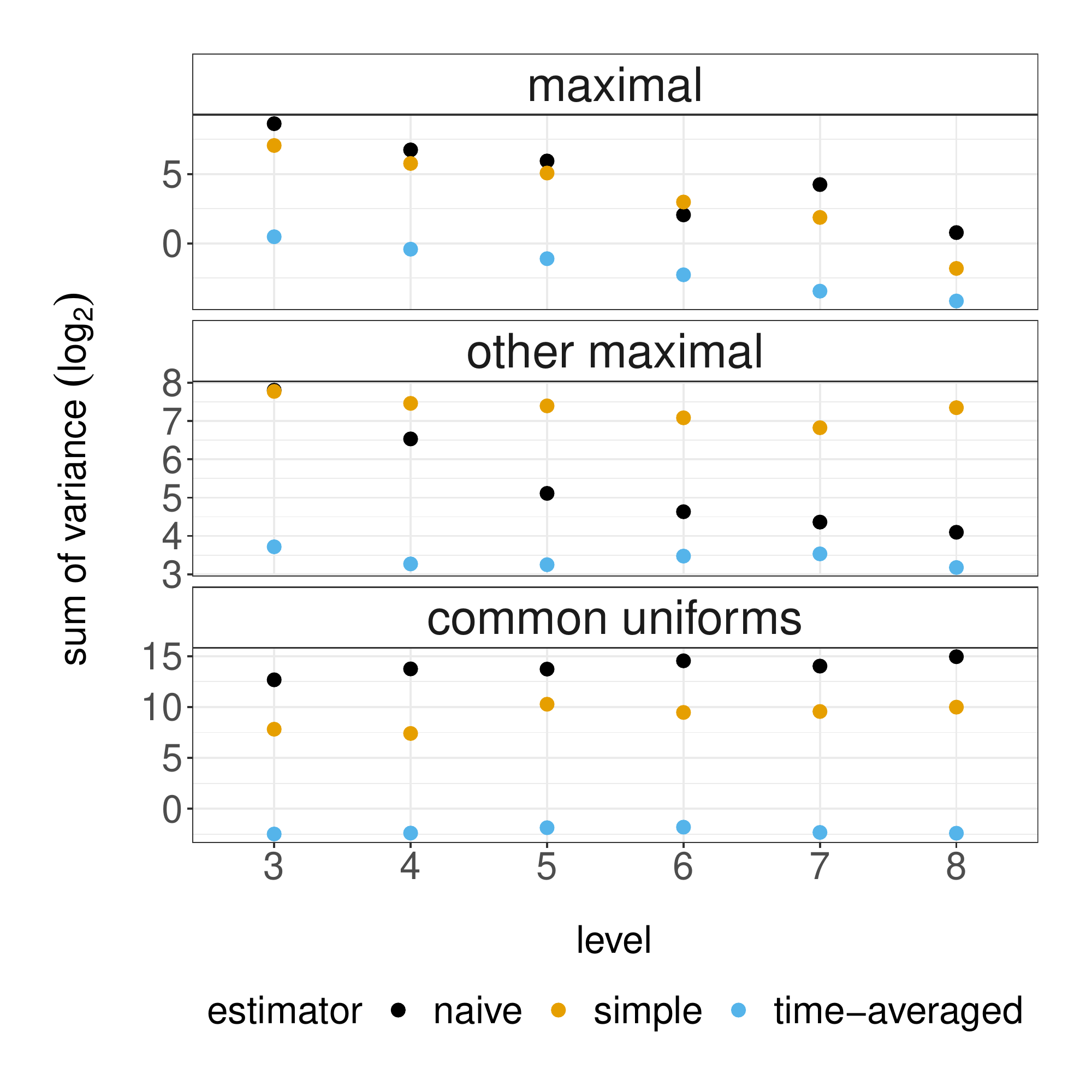}
\caption{Variance of score increment estimator summed over components 
$\sum_{j=1}^{d_{\theta}}\mathrm{Var}[\widehat{I}_l(\theta)^j]$ against discretization level $l$} 
\label{fig:ou_varincrements}
\end{subfigure}
\begin{subfigure}[t]{0.45\textwidth}
\includegraphics[trim=0pt 32pt 0pt 0pt,clip,width=.95\textwidth]{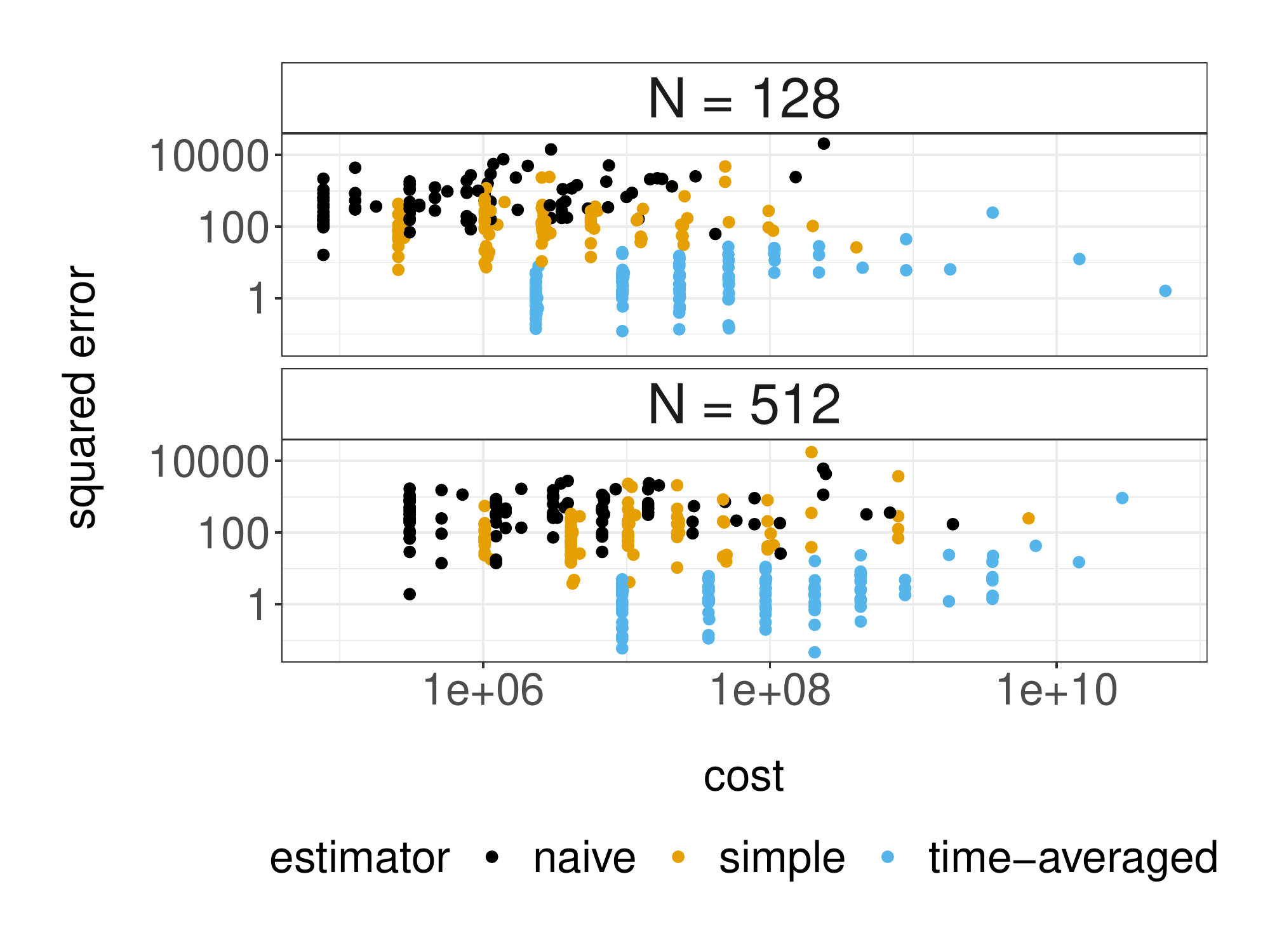}
\caption{Squared error $\| \widehat{S}(\theta) - S(\theta) \|_2^2$ against cost with $N\in\{128,512\}$ particles}
\label{fig:ou_squarederrors}
\end{subfigure} \quad
\begin{subfigure}[t]{0.45\textwidth}
\includegraphics[trim=0pt 32pt 0pt 0pt,clip,width=.95\textwidth]{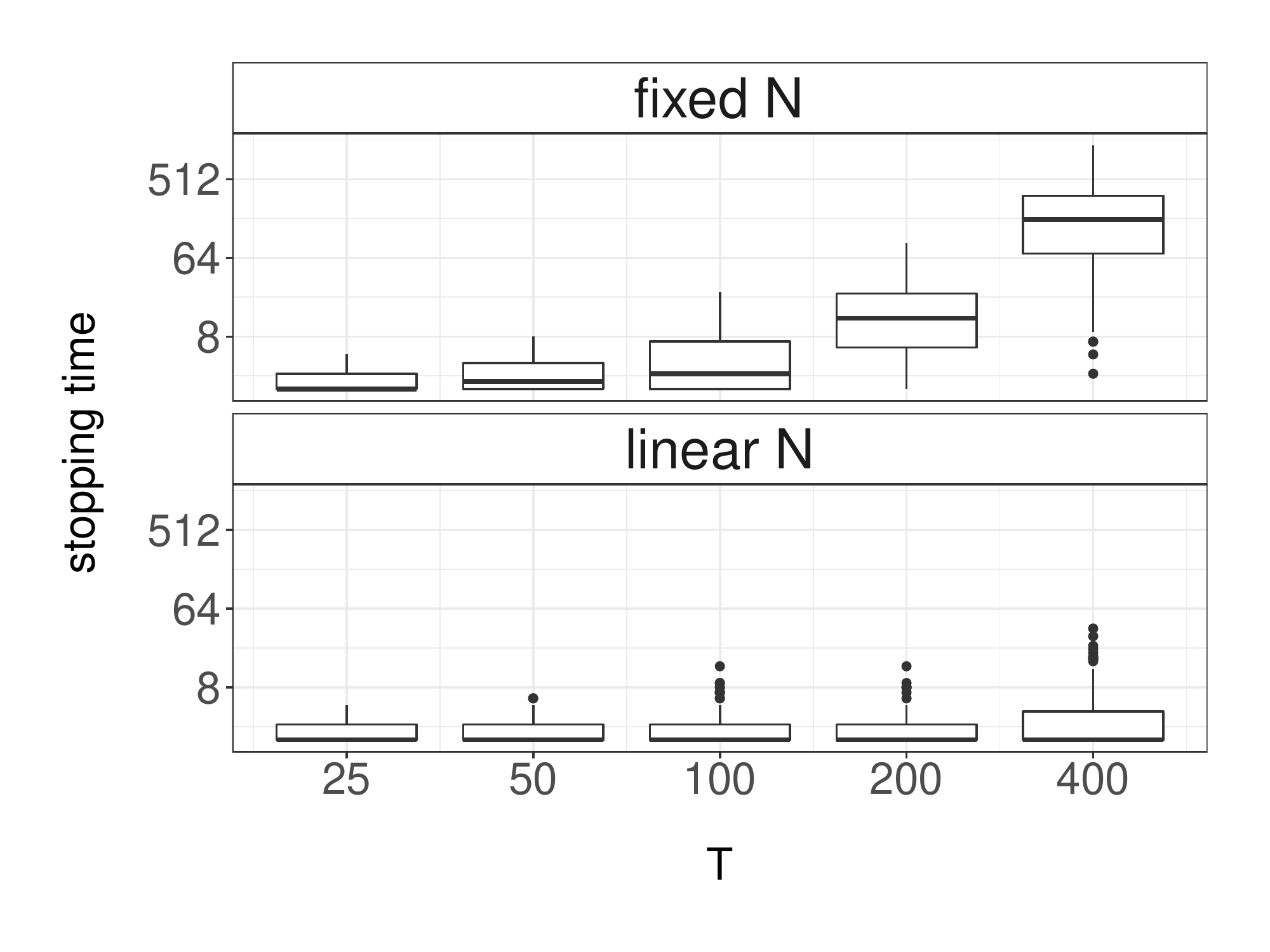}
\caption{Boxplots of stopping time $\bar{\tau}_{\theta}^l$ against number of observations $T$ for discretization level $l=6$ with $N=128$ fixed or linearly increasing number of particles}
\label{fig:ou_stoppingtimes_nobs}
\end{subfigure}
\begin{subfigure}[t]{0.45\textwidth}
\includegraphics[trim=0pt 32pt 0pt 0pt,clip,width=.95\textwidth]{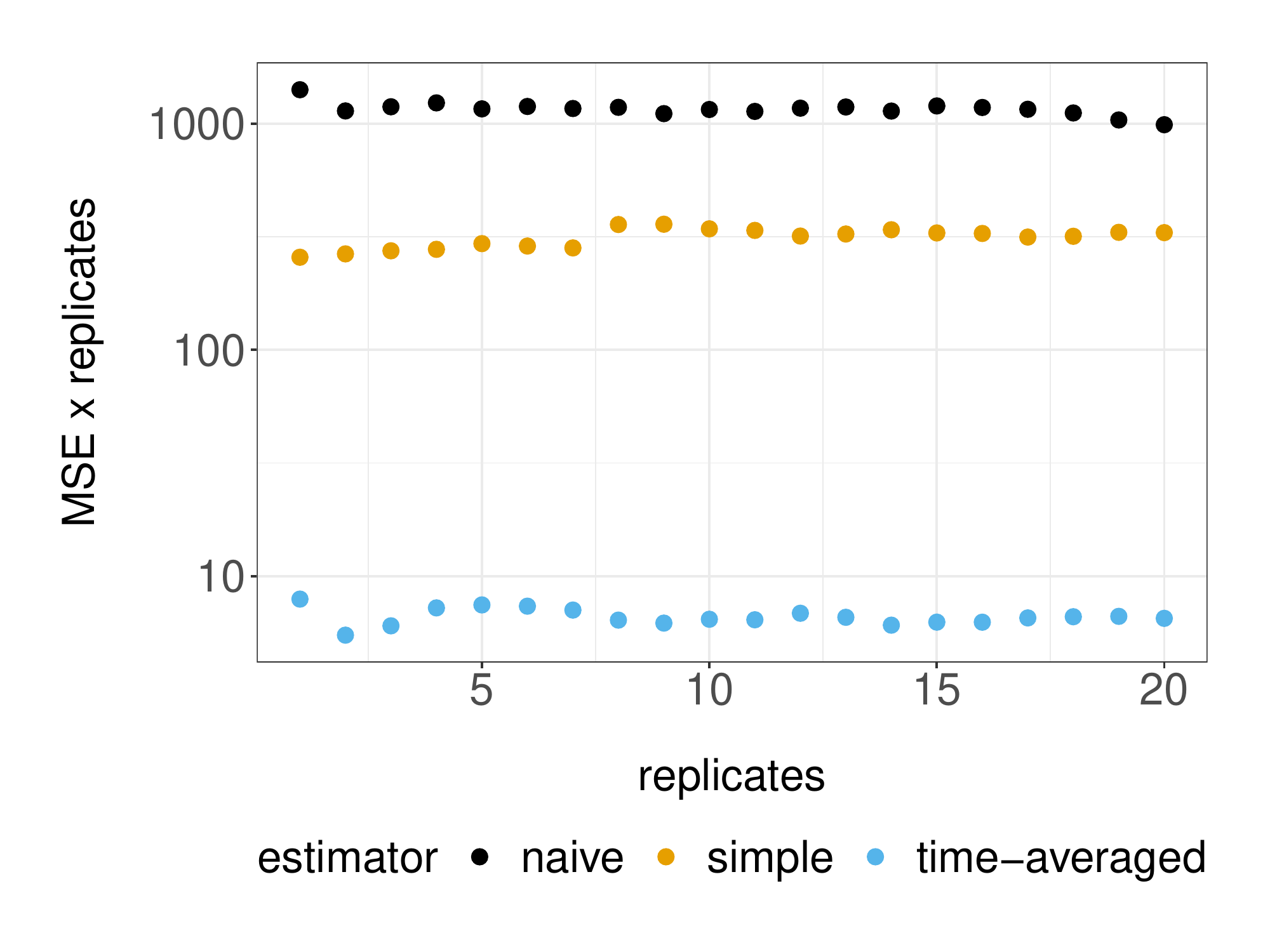}
\caption{Mean squared error $\mathrm{MSE}(\theta) = \mathbb{E}[\| \bar{S}(\theta) - S(\theta) \|_2^2]$ against number of averaged replicates~$R$}
\label{fig:ou_averaging1}
\end{subfigure} \quad
\begin{subfigure}[t]{0.45\textwidth}
\includegraphics[trim=0pt 32pt 0pt 0pt,clip,width=.95\textwidth]{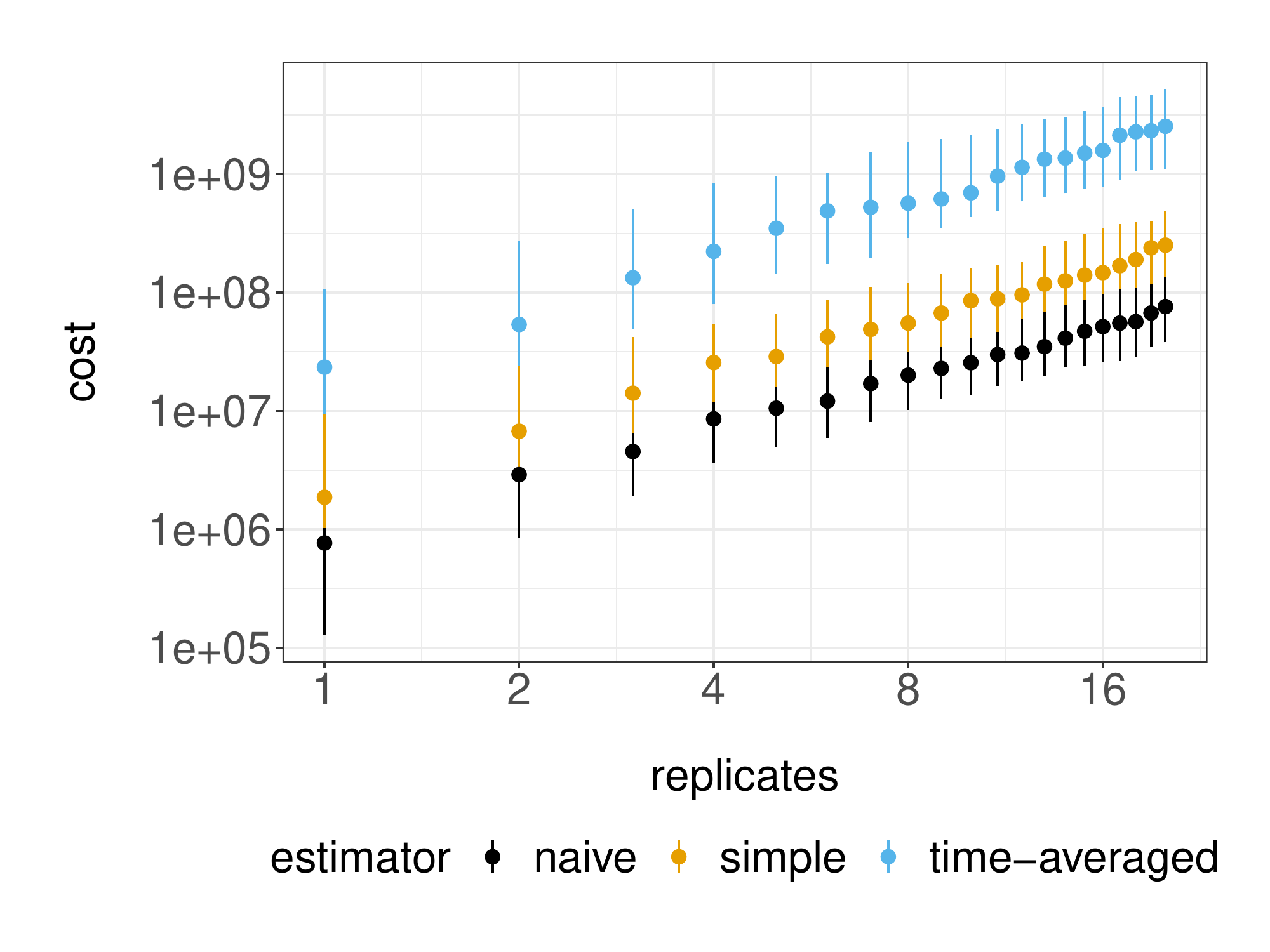}
\caption{Boxplots of cost against number of averaged replicates $R$}
\label{fig:ou_averaging2}
\end{subfigure} 
\caption{Behaviour of different coupling methods and score estimators at the data generating parameter $\theta=(2,7,1)$ of the Ornstein--Uhlenbeck model in Section~\ref{sec:ornstein_uhlenbeck}. $T=25$ observations and $N=128$ particles were employed unless stated otherwise. These plots are based on $100$ independent repetitions.}
\end{figure}

\subsection{Logistic diffusion model for population dynamics of red kangaroos}
\label{sec:logistic_diffusion}

% Fixed parameter 
% - plot stopping times against level
% - to decrease stopping times
% 	- experiment with initialization from latent process and sample from BPF. 
% 	- experiment with maximal with and without step 5
% 	- adaptive resampling	
% - plot increments against levels for naive and simple score estimators
% 	- maybe illustrate an insufficient and sufficient N?
% - try out Matti's estimators again 

% stochastic gradient ascent
% with decreasing learning rate
% - run multiple runs (from fixed parameter)  with naive and simple score estimators 
% - for the "average" run, plot how meeting time and variance of score estimators behave along the run
% - plot distribution of parameters along SGA run to see if we converge 
% with constant learning rate
% - run multiple runs (from fixed parameter)  with naive and simple score estimators 
% - for each run, compute polyak's average 
% - plot distribution of polyak's average along SGA run to see if we converge

Next we consider an application from population ecology to model the dynamics 
of a population of red kangaroos (Macropus rufus) in New South Wales, Australia. 
Figure \ref{fig:logistic_diffusion_observations} displays data $y_{t_1},\ldots,y_{t_P}\in\mathbb{N}_{0}^{2}$ 
from \citet{caughley1987kangaroos}, which are double transect counts on $P=41$ occasions at 
irregular times $(t_p)_{p=1}^P$ between $1973$ to $1984$. 
The latent population size $Z=(Z_t)_{t_1\leq t\leq t_P}$ is assumed to follow a logistic diffusion process 
with environmental variance \citep{dennis1988analysis,knape2012fitting}
\begin{align}\label{eqn:logistic_diffusion}
	dZ_{t}=(\theta_{3}^{2}/2+\theta_{1}-\theta_{2}Z_{t})Z_{t}dt+\theta_{3}Z_{t}dW_{t},\quad 
	Z_{t_1}\sim\mathcal{LN}(5,10^2),
\end{align}
where $\mathcal{LN}$ denotes the log-Normal distribution. 
The parameters $\theta_{1}\in\mathbb{R}$ and $\theta_{2}>0$ can be seen as coefficients describing how the growth rate depends on the population size.  
As the parameter $\theta_{3}>0$ appears in the diffusion coefficient of \eqref{eqn:logistic_diffusion}, 
we apply the Lamperti transformation $X_t = \Psi(Z_t)=\log(Z_t)/\theta_3$. 
By It\^{o}'s lemma, the transformed process $X=(X_t)_{t_1\leq t \leq t_P}$ satisfies the SDE \eqref{eq:intro_diff} 
with random initialization $X_{t_1}\sim\mu_{\theta}=\mathcal{N}(5/\theta_3,10^2/\theta_3^{2})$, 
drift function $a_{\theta}(x)=\theta_1/\theta_3-(\theta_2/\theta_3)\exp(\theta_3x)$ and 
unit diffusion coefficient $\sigma(x)=1$ for $x\in\mathbb{R}$. 
The observations $(Y_{t_p})_{p=1}^P$ are modelled as conditionally independent given $X$ and 
negative Binomial distributed, i.e.\ 
the observation density at time $t\in\{t_1,\ldots,t_P\}$ is 
$g_{\theta}(y_t|x_t)=\mathcal{NB}(y_t^{1};\theta_{4},\exp(\theta_3x_t))\mathcal{NB}(y_t^{2};\theta_{4},\exp(\theta_3x_t))$, 
where $\theta_{4}>0$. 
We will use a parameterization of the negative Binomial distribution that is common in ecology, 
$\mathcal{NB}(y;r,\mu)=\frac{\Gamma(y+r)}{\Gamma(r)y!}(\frac{r}{r+\mu})^{r}(\frac{\mu}{r+\mu})^{y}$ for $y\in\mathbb{N}_{0}$, where $r>0$ is the dispersion parameter and $\mu>0$ is the mean parameter. The $d_{\theta}=4$ unknown parameters to be inferred are 
$\theta=(\theta_1,\theta_2,\theta_3,\theta_4)\in\Theta=\mathbb{R}\times(0,\infty)^3$.  

Application of our methodology requires some minor modifications. 
As the initial distribution $\mu_{\theta}$ depends on $\theta_{3}$, 
the representation of score functions in \eqref{eqn:score_function} and \eqref{eqn:discretized_score} require  
adding $\nabla_{\theta}\log\mu_{\theta}(X_{t_1})$ to \eqref{eqn:smoothing_functional} and \eqref{eqn:discrete_test_function};  
see Appendix~\ref{app:logistic_diffusion} for model-specific expressions. 
To deal with irregular observation times $(t_p)_{p=1}^P$, 
we set the step-size at discretization level zero as the size of the smallest time interval, i.e.\ $\Delta_{0}=\min_{p=2,\ldots,P}t_{p}-t_{p-1}$.  
Higher levels $l\in\mathbb{N}$ will employ $\Delta_{l}=\Delta_{0}2^{-l}$. 
For level $l\in\mathbb{N}_{0}$, the first time interval $[t_1,t_2]$ is discretized using $\Delta_{l}$ sequentially, 
i.e.\ we set $s_{k}=t_1+k\Delta_{l}$ for $k\in\{0,\ldots,m_{l,1}\}$ with 
$m_{l,1}=\left\lfloor (t_{2}-t_{1})/\Delta_{l}\right\rfloor$, 
and $s_{k}=t_{2}$ for $k=m_{l,1}+1$ if $(t_{2}-t_{1})/\Delta_{l}\notin\mathbb{N}$. 
The subsequent time intervals are then discretized in the same manner.

Figure~\ref{fig:logistic_diffusion_results1} illustrates how the median and the 90\% quantile of the 
stopping time $\bar{\tau}_{\theta}^l$ vary with the discretization level $l$, the impact of 
the number of particles $N$ and the benefits of employing adaptive resampling.  
As before, the coupled resampling scheme in Algorithm \ref{alg:maximal-maximal} 
results in stopping times that are smaller for higher discretization levels, 
with less variability over levels as the number of particles increases. 
Moreover, resampling only when the effective sample size is less than $N/2$ allows us to induce 
more dependencies between the multiple CPF chains at lower discretization levels. 
Using $N=256$ particles and adaptive resampling, Figure~\ref{fig:logistic_diffusion_results2} examines 
the rate at which the variance of the estimated increment decreases with the discretization level. 
Here we consider the ``naive'' and ``simple'' estimators described in Section \ref{sec:ornstein_uhlenbeck}, 
with a burn-in of $b=90\%\textrm{-quantile}(\bar{\tau}_{\theta}^l)$ at level $l=3$, 
and omit the more costly ``time-averaged'' estimator. 
From the plot, both type of estimators have similar rate of decay and are valid choices 
in our score estimation methodology.  
Using the ``simple'' estimator, Figure~\ref{fig:logistic_diffusion_results3} verifies that 
the average of $R$ independent replicates of the resulting score estimator $\bar{S}(\theta)$ 
satisfies the standard Monte Carlo rate as $R\rightarrow\infty$. 

Lastly, we perform Bayesian parameter inference by employing our score estimators 
within the SGLD framework \citep{welling2011bayesian}. We rely on logarithmic transformations 
to deal with positivity parameter constraints, and specify the prior distribution  
$(\theta_1, \log \theta_2, \log \theta_3, \log \theta_4)\sim \mathcal{N}_{d_{\theta}}(\mu_0,\Sigma_0)$, 
where $\mu_0=(0, -1, -1, -1)$ and $\Sigma_0=\textrm{diag}(5^2, 2^2, 2^2, 2^2)$. 
As $\log \theta_3$ has a significantly different scale compared to the other parameters, 
we let the learning rate in \eqref{eqn:SGLD} be component-dependent 
by taking $\varepsilon_m=\textrm{diag}((100+m)^{-0.6}(10^{-2}, 10^{-2}, 10^{-4}, 10^{-2}))$ at iteration $m\geq 1$. 
The algorithmic settings used to produce score estimators are the same as in Figure~\ref{fig:logistic_diffusion_results3} 
with $R=1$ realization. 
Figure~\ref{fig:logistic_diffusion_results4} shows the trace plot over $7500$ iterations of the resulting SGLD algorithm 
for each parameter.

\begin{figure}
\centering
\begin{subfigure}[t]{0.45\textwidth}
\includegraphics[trim=0pt 32pt 0pt 25pt,clip,width=.95\textwidth]{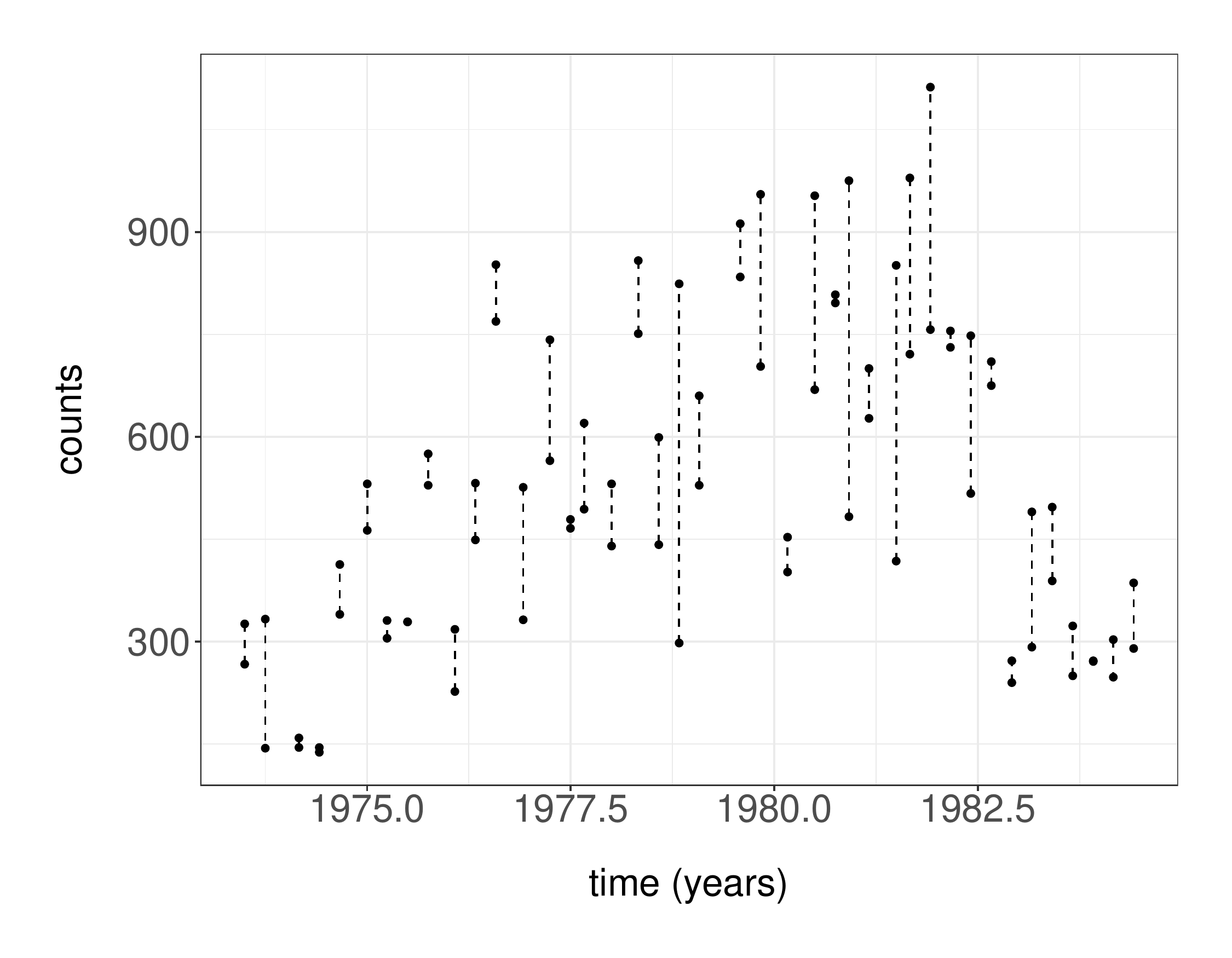}
\caption{Double transect counts}
\label{fig:logistic_diffusion_observations}
\end{subfigure} \quad
\begin{subfigure}[t]{0.45\textwidth}
\includegraphics[trim=0pt 32pt 0pt 25pt,clip,width=.95\textwidth]{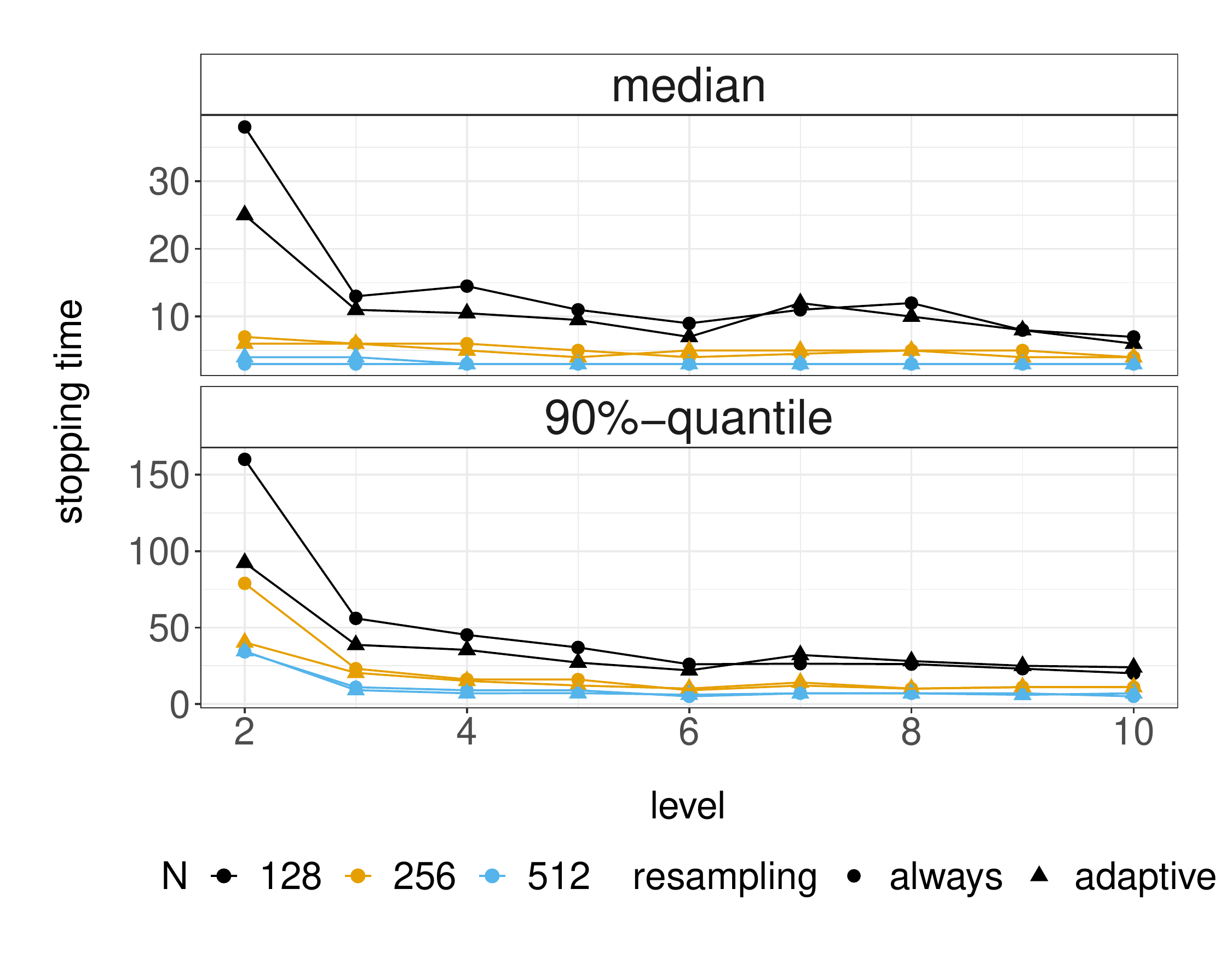}
\caption{Sample median and 90\% quantile of stopping time $\bar{\tau}_{\theta}^l$ against discretization level $l$}
\label{fig:logistic_diffusion_results1}
\end{subfigure} \quad
\begin{subfigure}[t]{0.45\textwidth}
\includegraphics[trim=0pt 32pt 0pt 25pt,clip,width=.95\textwidth]{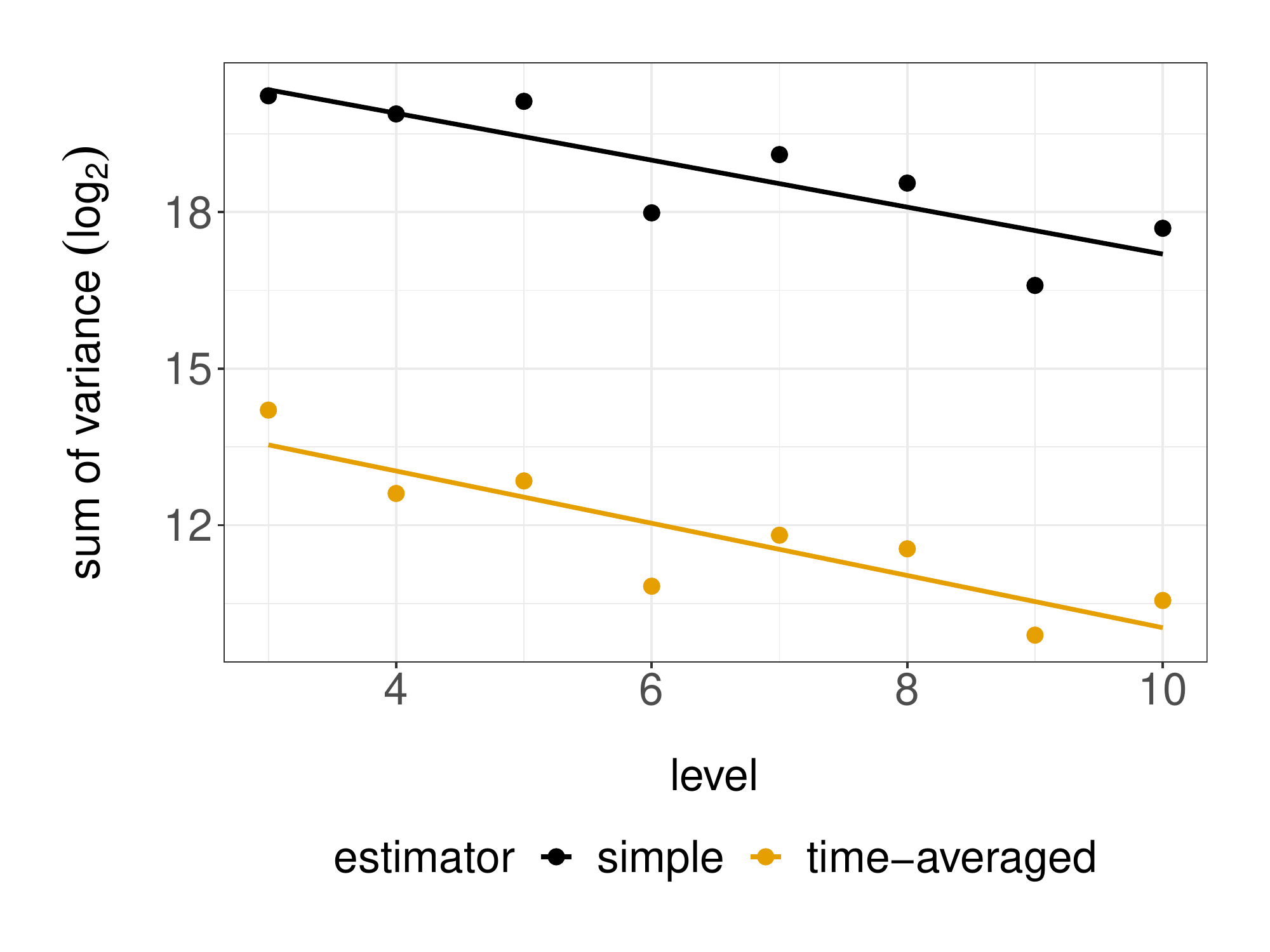}
\caption{Variance of score increment estimator summed over components 
$\sum_{j=1}^{d_{\theta}}\mathrm{Var}[\widehat{I}_l(\theta)^j]$ against discretization level $l$}
\label{fig:logistic_diffusion_results2}
\end{subfigure} \quad
\begin{subfigure}[t]{0.45\textwidth}
\includegraphics[trim=0pt 32pt 0pt 25pt,clip,width=.95\textwidth]{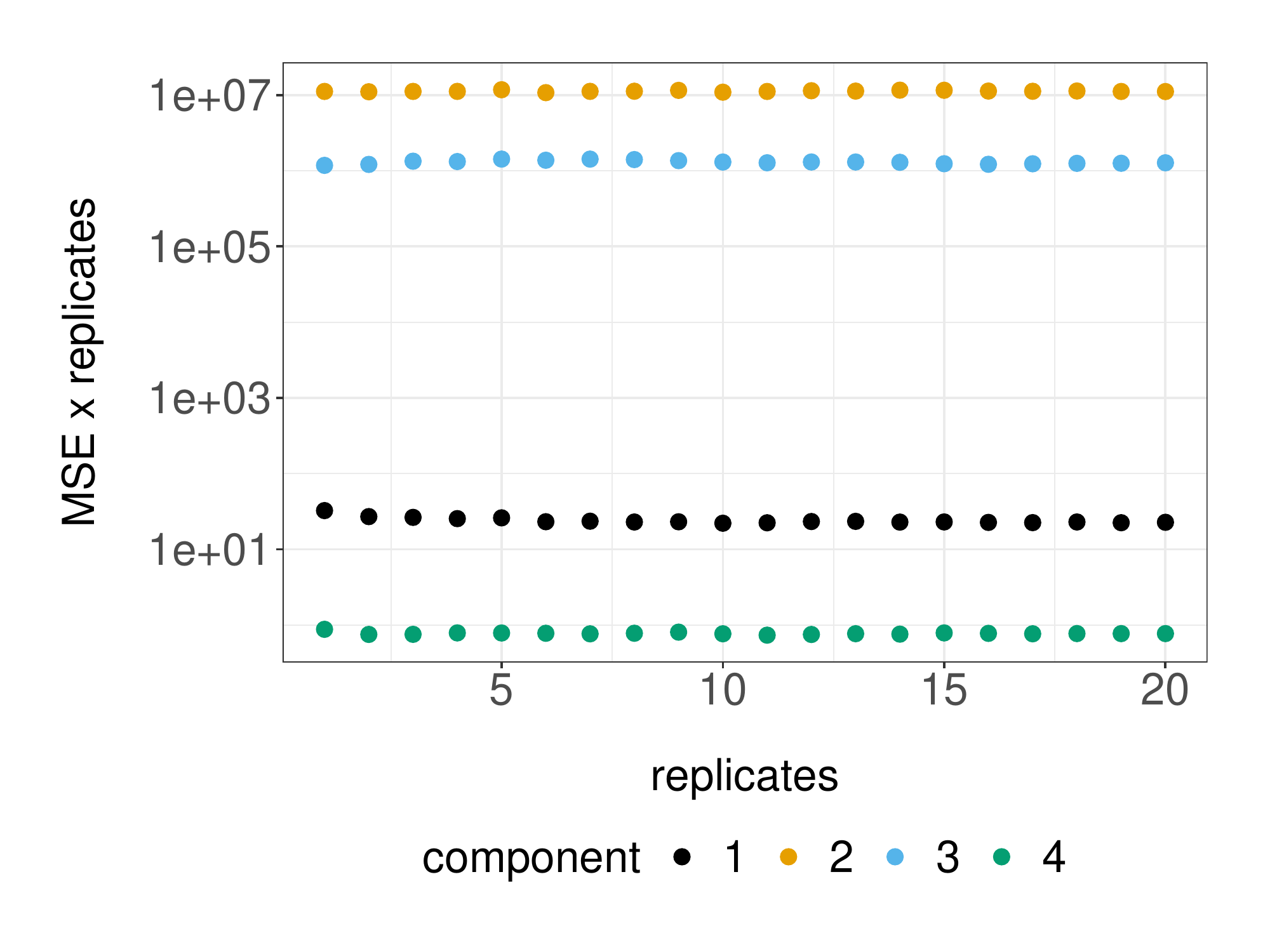}
\caption{Mean squared error $\mathrm{MSE}(\theta)^j = \mathbb{E}[|\bar{S}(\theta)^j - S(\theta)^j |^2]$ against number of averaged replicates~$R$ 
for component $j\in\{1,\ldots,d_{\theta}\}$}
\label{fig:logistic_diffusion_results3}
\vspace{1em}
\end{subfigure}
\begin{subfigure}[t]{\textwidth}
\centering
\includegraphics[trim=5pt 5pt 5pt 0pt,clip,width=.85\textwidth]{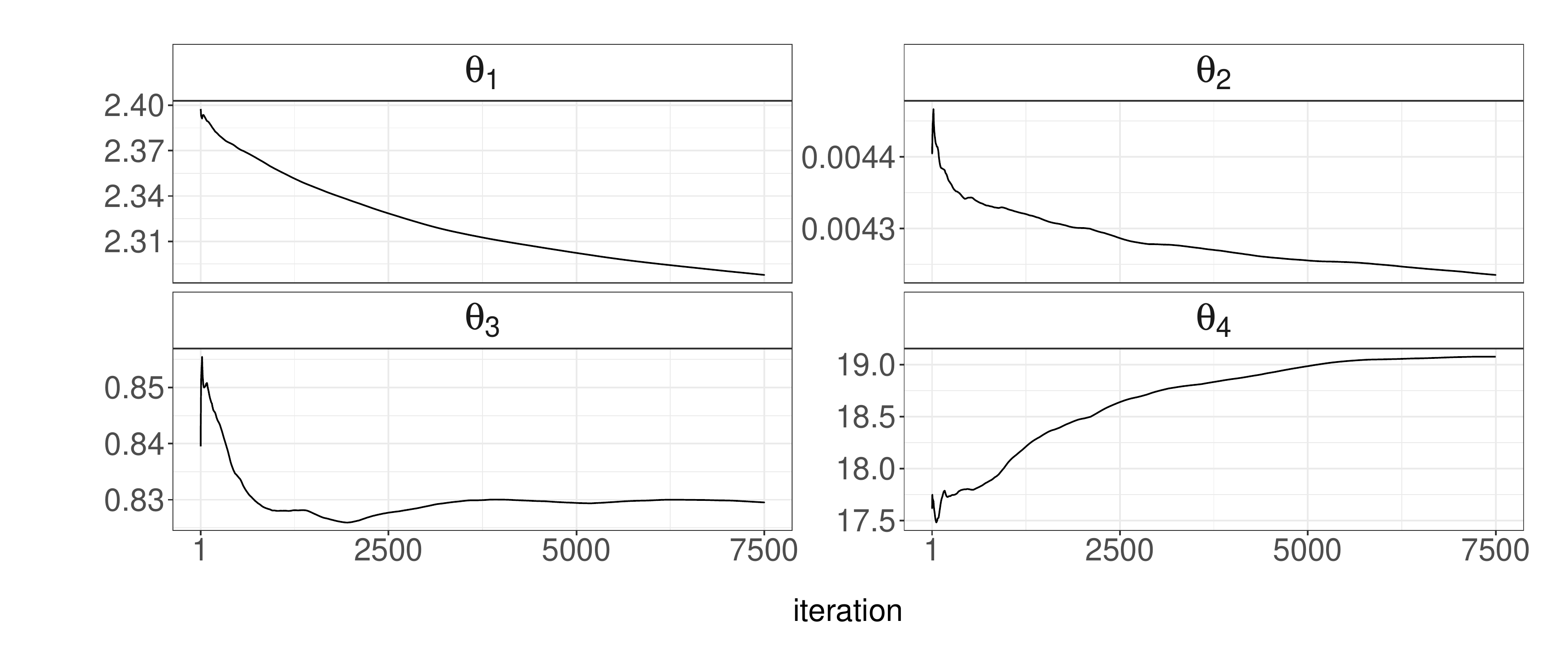}
\caption{SGLD for each component.}
\label{fig:logistic_diffusion_results4}
\end{subfigure}
\caption{Behaviour at parameter $\theta=(2.397, 4.429 \times 10^{-3}, 0.840, 17.631)$ of the logistic diffusion model in Section~\ref{sec:logistic_diffusion}. $N=256$ particles were employed unless stated otherwise. These plots are based on $100$ independent repetitions.}
\end{figure}

\subsection{Neural network model for grid cells in the medial entorhinal cortex}\label{sec:neural_network}

As our final application, we consider a neural network model for single neurons to analyze grid
cells spike data\footnote{\url{https://www.ntnu.edu/kavli/research/grid-cell-data}}
recorded in the medial entorhinal cortex of rats that were running
on a linear track \citep{hafting2008hippocampus}. The neural states $Z_{t}=(Z_{t}^{1},Z_{t}^{2})$
of two grid cells that were simultaneously recorded is assumed to follow 
\begin{align}\label{eqn:nn_model}
dZ_{t}^{1} & =\left(\alpha_{1}\tanh(\beta_{1}Z_{t}^{2}+\gamma_{1})-\delta_{1}Z_{t}^{1}\right)dt+\sigma_{1}dW_{t}^{1},\\
dZ_{t}^{2} & =\left(\alpha_{2}\tanh(\beta_{2}Z_{t}^{1}+\gamma_{2})-\delta_{2}Z_{t}^{2}\right)dt+\sigma_{2}dW_{t}^{2},\notag
\end{align}
for $t\in[0,T]$, where $(\alpha_{1},\alpha_{2})\in\mathbb{R}^{2}$
controls the amplitude, $(\beta_{1},\beta_{2})\in\mathbb{R}^{2}$
describes the connectivity between the cells, $(\gamma_{1},\gamma_{2})\in\mathbb{R}^{2}$
are baseline levels, $(\delta_{1},\delta_{2})\in(0,\infty)^{2}$
determines the strength of the mean reversion towards the origin.
We assume $Z_{0}=(0,0)$ at the beginning of the experiment. 
This diffusion is motivated
by an example in \cite{kappen2016adaptive}, and modified for our purposes.
To infer the unknown diffusivity parameters $(\sigma_{1},\sigma_{2})\in(0,\infty)^2$,
we consider the transformation 
$X_{t}=(X_{t}^{1},X_{t}^{2})=\Psi(Z_t)=(Z_{t}^{1}/\sigma_{1},Z_{t}^{2}/\sigma_{2})$, 
which rescales each component of the diffusion. By It\^{o}'s formula,
the transformed process $X=(X_t)_{0\leq t\leq T}$ satisfies the diffusion model (\ref{eq:intro_diff})
with initialization $x_{\star}=(0,0)$, drift function 
\begin{align}
a_{\theta}(x)=\left(\begin{array}{c}
a_{\theta}^{1}(x)\\
a_{\theta}^{2}(x)
\end{array}\right)=\left(\begin{array}{c}
\alpha_{1}\tanh(\beta_{1}\sigma_{2}x^{2}+\gamma_{1})/\sigma_{1}-\delta_{1}x^{1}\\
\alpha_{2}\tanh(\beta_{2}\sigma_{1}x^{1}+\gamma_{2})/\sigma_{2}-\delta_{2}x^{2}
\end{array}\right),
\end{align}
and diffusion coefficient $\sigma(x)=I_{2}$ for $x=(x^1,x^2)\in\mathbb{R}^2$. 

The experimental data over a duration of $T=20$ seconds contains time stamps 
in $[0,T]$ when a spike at one of the two cells is recorded using
tetrodes. Following \cite{brown2005theory}, we adopt an inhomogenous Poisson
point process to model these times. 
Let $t_p = p T 2^{-6}$ for $p\in\{0,1,\ldots,P\}$ denote a dyadic uniform discretization of $[0,T]$ 
into $P=2^6$ time intervals. 
Given the latent process $X=(X_t)_{0\leq t\leq T}$, the number of spikes $Y_{t_{p}}^{i}$ occurring 
in each time interval $[t_{p-1},t_{p}]$ at cell $i=1,2$ 
is assumed to be conditionally independent of the other time intervals and 
the activity in the other cell, and follow a Poisson distribution
with rate $\int_{t_{p-1}}^{t_{p}}\lambda_{i}(X_{t}^{i})dt$. The intensity
function for grid cell $i=1,2$ is modelled as $\lambda_{i}(X_{t}^{i})=\exp(\kappa_{i}+X_{t}^{i})$,
where $\kappa_{i}\in\mathbb{R}$ represents a baseline level. 
The observed counts $y_{t_p}=(y_{t_{p}}^{1},y_{t_{p}}^{2})$ for interval $p\in\{1,\ldots,P\}$, 
computed from the experimental data, are displayed in Figure \ref{fig:neural_network_observations}. 
The conditional likelihood of the observation model is 
$p_{\theta}(y_{t_{1}},\ldots,y_{t_{P}}|X) = \prod_{p=1}^{P}g_{\theta}(y_{t_p}| (X_t)_{t_{p-1}\leq t\leq t_{p}})$ 
with the intractable observation density 
\begin{align}
        g_{\theta}(y_{t_p}| (X_t)_{t_{p-1}\leq t\leq t_{p}}) = 
        \prod_{i=1}^2\mathcal{P}oi\left(y_{t_p}^i;\int_{t_{p-1}}^{t_p} \lambda_i(X^i_t)dt\right),        
\end{align}
where $\mathcal{P}oi(y;\lambda)=\lambda^y\exp(-\lambda)/y!$ for $y\in\mathbb{N}_0$ 
denotes the PMF of a Poisson distribution with rate $\lambda>0$. 
To approximate the conditional likelihood, at level $l \geq 6$, we discretize the time interval $[0,T]$ 
in a similar manner using $s_k=k \Delta_l$ for $k\in\{0,1,\ldots,K_l\}$, where $\Delta_l=T2^{-l}$ 
is the step-size and $K_l=2^l$ is the number of time steps. 
Under the time-discretized process $X_{0:T}=(X_{s_k})_{k=0}^{K_l}$, 
the resulting approximation of the conditional likelihood is 
$p^l_{\theta}(y_{t_{1}},\ldots,y_{t_{P}}| X_{0:T}) = \prod_{p=1}^P g^l_{\theta} (y_{t_p} | (X_t)_{t_{p-1}\leq t\leq t_{p}})$ 
with the corresponding observation density 
\begin{align}\label{eqn:NN_obs_density}
	g^l_{\theta} (y_{t_p} | (X_t)_{t_{p-1}\leq t\leq t_{p}}) = 
	\prod_{i=1}^2\mathcal{P}oi\left(y_{t_p}^i; \Delta_l\sum_{t: t_{p-1}\leq t\leq t_p}\lambda_i(X^i_t)\right).        
\end{align}
By using these level-dependent observation densities \eqref{eqn:NN_obs_density} in Section \ref{sec:discretized_score}, 
our proposed methodology can then be applied. 
There are $d_{\theta}=12$ parameters $\theta=(\theta_1,\theta_2)$ to be inferred, where 
$\theta_i = (\alpha_{i},\beta_{i},\gamma_{i},\delta_{i},\sigma_{i},\kappa_{i})$ denote the parameters 
associated to cell $i=1,2$.
We refer the reader to Appendix~\ref{app:neural_network} for model-specific expressions 
to evaluate \eqref{eqn:discrete_test_function}. 

We consider an extension of the proposed method based on the conditional ancestor sampling particle filter 
(CASPF) \citep{lindsten2014particle} as the basic algorithmic building block. 
As CASPF has better mixing properties than CPF, we observe smaller stopping times in Figure~\ref{fig:neural_network_results1} 
for lower discretization levels. 
Figure~\ref{fig:neural_network_results2} verifies the validity of using both MCMC algorithms and 
``naive'' and ``simple'' estimators (as described in Section \ref{sec:ornstein_uhlenbeck}) within our methodology. 
For ``simple'' estimators, the burn-in was taken as $b=90\%\textrm{-quantile}(\bar{\tau}_{\theta}^l)$ at level $l=11$. 
The rate at which the variance of the estimated increment decays with the discretization level 
is similar across algorithms and estimators. 

Lastly, we combine our score estimators and the SGA scheme in \eqref{eqn:SGA} to perform 
maximum likelihood estimation. 
The score estimation relies on the CASPF algorithm and the ``simple'' estimator with a burn-in of $b=100$. 
Positivity parameter constraints are dealt with using logarithmic transformations and 
a constant learning rate of $\varepsilon_m = 10^{-3}$ is employed. 
Figure~\ref{fig:neural_network_results3} illustrates how the distribution of the Polyak--Ruppert average 
evolves over the iterations, estimated using independent runs of SGA. 
We note that only $87$ out of $100$ runs were considered, as there were 
$13$ instances of the variance of the score estimator driving the SGA algorithm to 
regions of the parameter space where the stopping times are prohibitively large, 
causing the SGA to stall. 
This pathological behaviour is due to poor mixing properties of the 
underlying CASPF algorithm at very unlikely regions of the parameter space. 
To improve the MCMC algorithm, one has to simulate particle dynamics 
that incorporate information from the observation sequence $y_{t_{1}},\ldots,y_{t_{P}}$, 
which could be investigated in future work. 
The parameter estimates of $\beta_1$ and $\beta_2$ support the use of a joint model 
\eqref{eqn:nn_model} for both grid cells, and indicates that these cells are positively dependent. 

\begin{figure}
\centering
\begin{subfigure}[t]{0.45\textwidth}
\includegraphics[trim=5pt 5pt 5pt 5pt,clip,width=.95\textwidth]{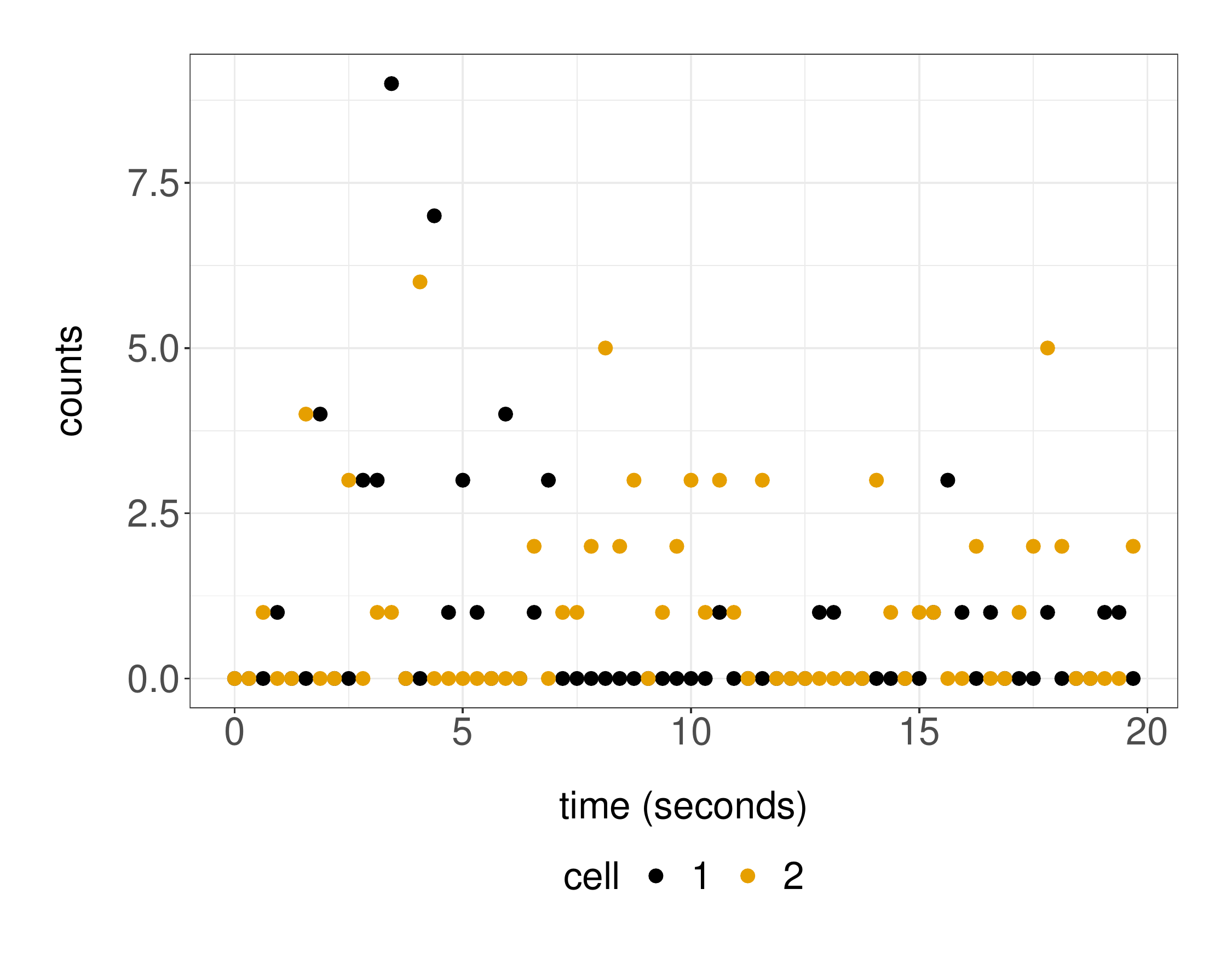}
\caption{Counts $y_{t_p}=(y_{t_{p}}^{1},y_{t_{p}}^{2})$ on time intervals of duration $T2^{-6}=0.3125$ second}
\label{fig:neural_network_observations}
\end{subfigure} \quad
\begin{subfigure}[t]{0.45\textwidth}
\includegraphics[trim=5pt 5pt 5pt 5pt,clip,width=.95\textwidth]{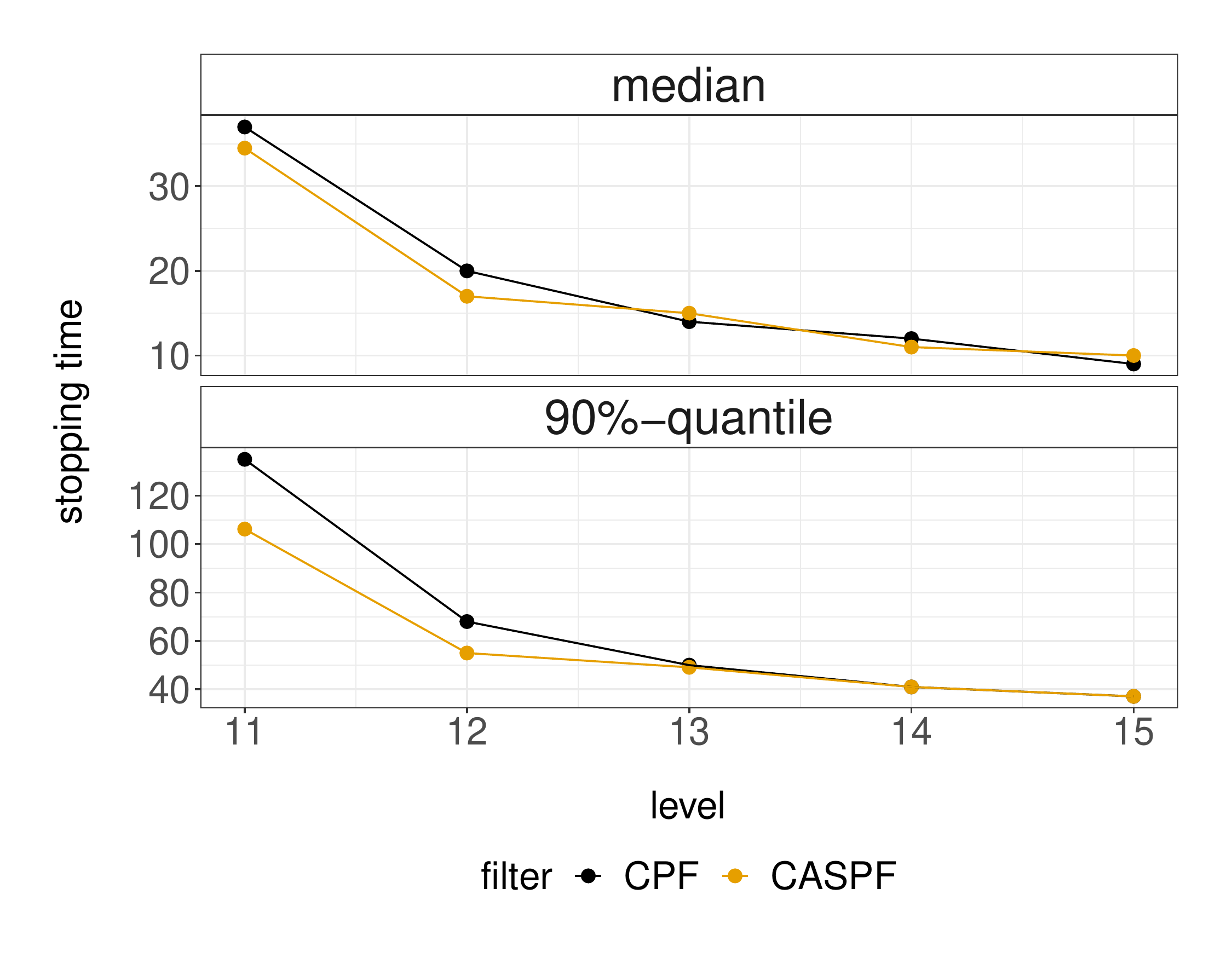}
\caption{Sample median and 90\% quantile of stopping time $\bar{\tau}_{\theta}^l$ against discretization level $l$}
\label{fig:neural_network_results1}
\end{subfigure} 
\begin{subfigure}[t]{0.45\textwidth}
\vspace{1em}
\includegraphics[trim=5pt 5pt 5pt 0pt,clip,width=.95\textwidth]{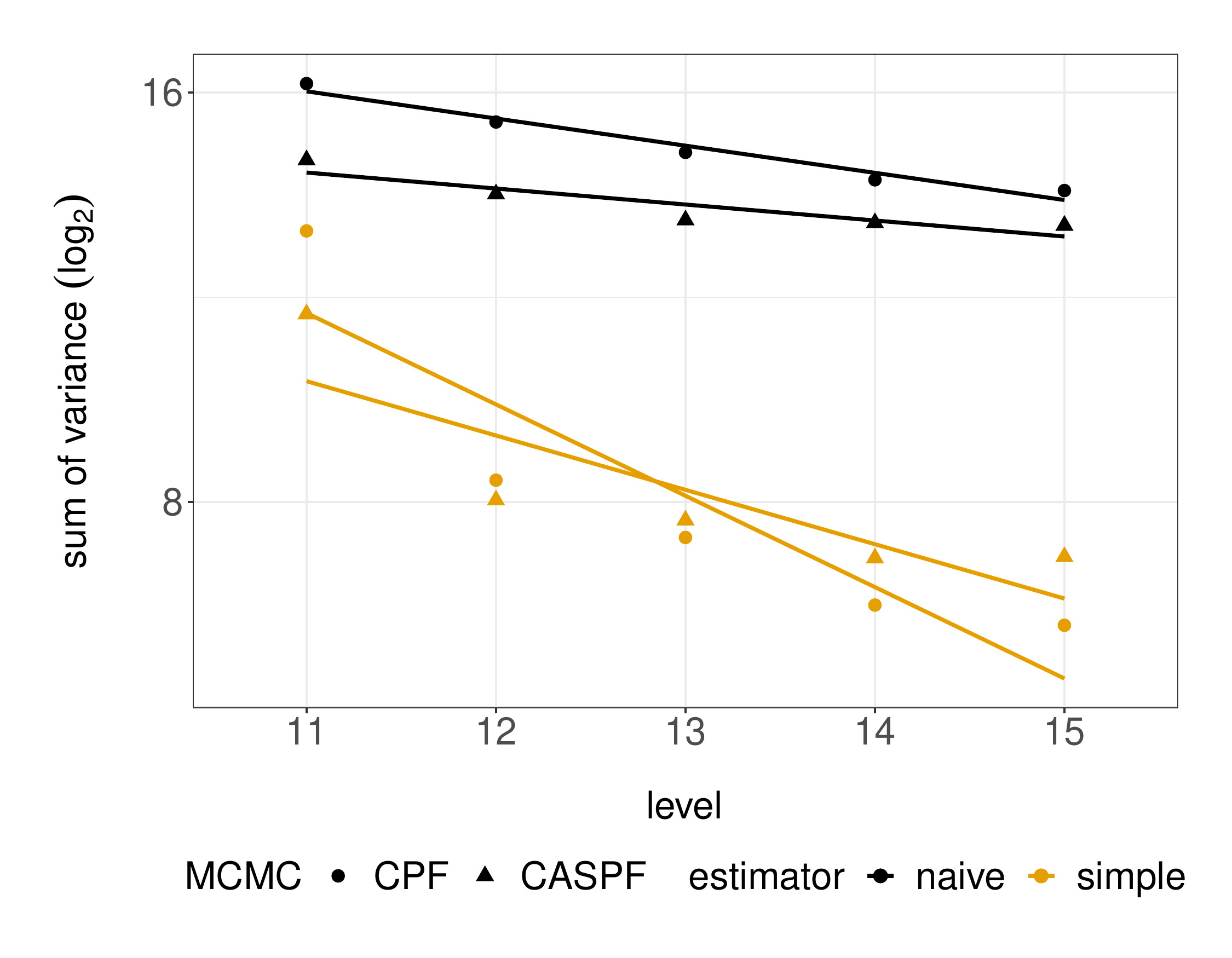}
\caption{Variance of score increment estimator summed over components 
$\sum_{j=1}^{d_{\theta}}\mathrm{Var}[\widehat{I}_l(\theta)^j]$ against discretization level $l$}
\label{fig:neural_network_results2}
\end{subfigure} \quad
\begin{subfigure}[t]{0.45\textwidth}
\vspace{1em}
\includegraphics[trim=5pt 5pt 5pt 0pt,clip,width=.95\textwidth]{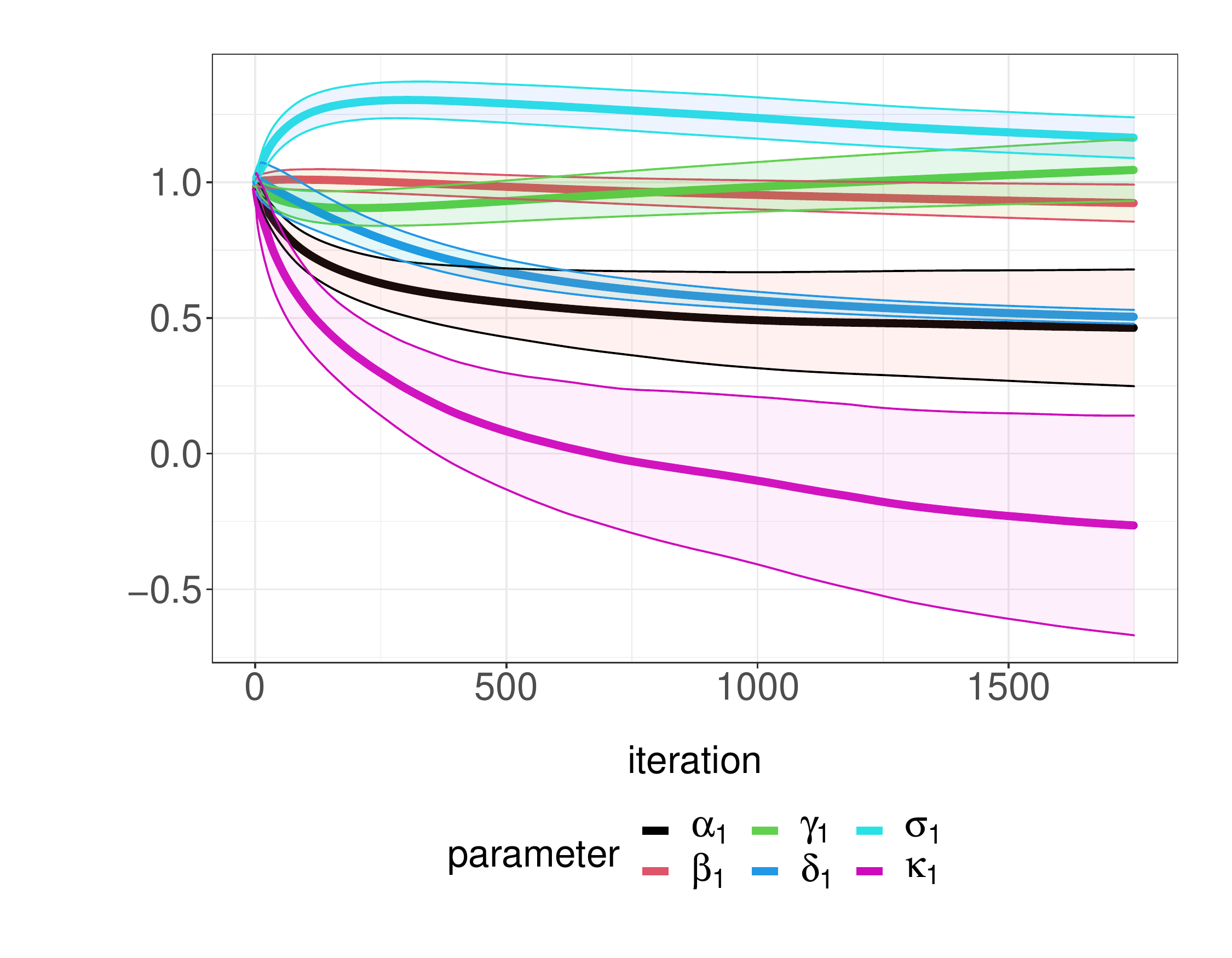}
\caption{Mean ($\pm$ one standard deviation) of Polyak--Ruppert average for 
parameter estimates of $\theta_1$ over $87$ runs of SGA}
\label{fig:neural_network_results3}
\end{subfigure}
\caption{Behaviour at parameter $\theta=(1,\dots,1)$ of the neural network model in Section~\ref{sec:neural_network}. 
The algorithmic settings involve $N=256$ particles and adaptive resampling. 
These plots are based on $1000$ independent repetitions unless stated otherwise.}
\end{figure}

\section{Parameter dependence in diffusion coefficient}\label{sec:parameter_diffusion}

We end this article by considering how one can extend the ideas in this article to accommodate 
the case where the diffusion coefficient also depends on the parameter $\theta \in \Theta$. 
In this case, we have the SDE
\begin{align}
dX_t = a_{\theta}(X_t)dt + \sigma_{\theta}(X_t)dW_t,\quad X_0=x_{\star}\in\mathbb{R}^d,
\end{align}
where $\sigma:\Theta\times\mathbb{R}^d\rightarrow\mathbb{R}^{d\times d}$ is assumed to be such that 
$\theta\mapsto\sigma_{\theta}$ is invertible and satisfies the conditions in Assumption~\ref{ass:D1} uniformly in $\theta\in\Theta$. 
Moreover, we shall suppose that all the conditions in \citet{schauer2017guided} hold. For $k\in\mathbb{N}_0$ and $t\in[k,k+1]$, 
consider the diffusion bridge
\begin{align}\label{eqn:diffusion_bridge}
dX_t = a_{\theta}^{\circ}(X_t)dt + \sigma_{\theta}(X_t)dW_t,\quad X_k=x_{k},~X_{k+1}=x_{k+1},
\end{align}
where the drift function $a_{\theta}^{\circ}$ is described in \citet{schauer2017guided}. 
Given a Brownian path $\bm{W}_k=(W_t)_{k\leq t\leq k+1}$, we denote the path-wise solution of the diffusion bridge 
as $F_{\theta,k}(\bm{W}_k,x_k,x_{k+1})$. 
Furthermore, for $G:\Theta\times\mathbb{R}^d\rightarrow\mathbb{R}$ given in \citet[Equation 2.3]{schauer2017guided} and a process 
$\bm{Z}_k=(Z_t)_{k\leq t \leq k+1}$, we define the functional
\begin{align}
H_{\theta,k}(\bm{Z}_k) = \int_{k}^{k+1} G_{\theta}(Z_t)dt.
\end{align}
Using the change of measure in \citet{schauer2017guided} along with the approach in \citet{beskos2021score} 
and \citet{yonekura2020online}, one can write the marginal likelihood of observations $y_{1:T}=(y_{t})_{t=1}^T$ as
\begin{align}\label{eqn:general_marginal_likelihood}
p_{\theta}(y_{1:T}) = \widetilde{\mathbb{E}}_{\theta}\left[\prod_{t=1}^T g_{\theta}(y_t|X_t)
\exp\Big\{\sum_{t=0}^{T-1} H_{\theta,t}(F_{\theta,t}(\bm{W}_t,X_t,X_{t+1}))\Big\}\right].
\end{align}
In the above, $\widetilde{\mathbb{E}}_{\theta}$ denotes expectation w.r.t. the probability measure $\widetilde{\mathbb{P}}_{\theta}$ 
defined as
\begin{align}
\widetilde{\mathbb{P}}_{\theta}(d(x_1,\dots,x_T,\bm{W}_0,\dots,\bm{W}_{T-1})) 
= \prod_{t=1}^T \bigg\{\widetilde{p}_{\theta}(x_{t-1},x_t)dx_t\bigg\}\widetilde{\mathbb{W}}(d(\bm{W}_0,\dots,\bm{W}_{T-1})),
\end{align}
where $\widetilde{p}_{\theta}(x_{t-1},x_t)$ is the transition density of an auxiliary process on a unit time interval 
as constructed in \cite{schauer2017guided} that is known and can be sampled, and 
$\widetilde{\mathbb{W}}(d(\bm{W}_0,\dots,\bm{W}_{T-1})) = \bigotimes_{k=0}^{T-1} \mathbb{W}(d\bm{W}_{k})$
is given by the Wiener measure $\mathbb{W}$. 

Under regularity conditions, one can differentiate \eqref{eqn:general_marginal_likelihood} and represent the score function as
\begin{align}
S(\theta) = 
\breve{\mathbb{E}}_{\theta}\left[
\sum_{t=1}^T \nabla_{\theta}\log g_{\theta}(y_t|X_t) + 
\sum_{t=1}^T \nabla_{\theta}\log\widetilde{p}_{\theta}(X_{t-1},X_t) +
\sum_{t=1}^{T-1} \nabla_{\theta} H_{\theta,t}(F_{\theta,t}(\mathbf{W}_t,X_t,X_{t+1}))
\right],
\end{align}
where $\breve{\mathbb{E}}_{\theta}$ denotes expectation w.r.t.\ the probability measure
\begin{align}
&\breve{\mathbb{P}}_{\theta}(d(x_1,\dots,x_T,\bm{W}_0,\dots,\bm{W}_{T-1})) \notag\\
=~&p_{\theta}(y_{1:T})^{-1}\prod_{t=1}^T g_{\theta}(y_t|x_t)
\exp\Big\{
\sum_{t=0}^{T-1} H_{\theta,t}(F_{\theta,t}(\mathbf{W}_t,x_t,x_{t+1}))\Big\}
\widetilde{\mathbb{P}}_{\theta}(d(x_1,\dots,x_T,\bm{W}_0,\dots,\bm{W}_{T-1})). 
\end{align}
Practical implementation will require a discretized approximation of $\nabla_{\theta} H_{\theta,p}(F_{\theta,p}(\cdot))$, 
which involves a gradient w.r.t.\ $\theta$ of a path-wise solution of the diffusion bridge \eqref{eqn:diffusion_bridge}. 
Although Euler-type approximations can be obtained, the resulting bias in the sense of Theorem \ref{prop:conv_grad_log_like} 
is significantly more complicated to analyze and is thus left as future work.  
We stress that only small modifications to our proposed methodology is necessary 
to obtain unbiased estimators of the score function in this case.
A similar approach is considered in \cite{beskos2021score} for a class of continuous-time models. 
Alternatively, one could also consider using Malliavan techniques (see e.g.\ \citet{fournie1999applications}), 
instead of the ideas described here.

% papers to cite
% franks et al. Unbiased inference for discretely observed hidden Markov model diffusions
% lee, singh, vihola. Coupled conditional backward sampling particle filter.
% vihola. unbiased estimators and multilevel Monte Carlo.
% mention that we can also compute Fisher/Hessian
% discuss that single randomization requires additional compute effort as level increases which is not desirable 
% since the increments are getting smaller
% link it to franks et al. where the compute effort does not increase but unbiasedness is lost

\subsubsection*{Acknowledgements}
Ajay Jasra was supported by KAUST baseline funding. 
Jeremy Heng was funded by CY Initiative of Excellence (grant ``Investissements
d'Avenir'' ANR-16-IDEX-0008).

\bibliographystyle{abbrvnat}
\bibliography{ref}

\appendix

\section{Theoretical analysis}
\subsection{Introduction and preliminaries}
Section~\ref{app:diff_proc} provides some results on time-discretization of diffusions, which are needed for the proofs associated to 
Theorem~\ref{prop:conv_grad_log_like} as well as the 4-CCPF (Algorithm~\ref{alg:4-CCPF}). 
Our main technical arguments associated to Theorem~\ref{theo:ub} are given in Section~\ref{app:CCPF}, 
followed by several remarks about the proofs and discussions of alternative strategies. 
This section of the appendix is intended to be read in the order in which it is presented.  
Some familiarity with the approach in \cite{jasra2017multilevel} is also useful.

Note that our results concerning $\mathbb{L}_r$-norms are stated for $r\in[1,\infty)$; and can be extended to the case $r\in(0,1)$ by H\"older's inequality. We will use this fact without further elaboration. Throughout our arguments, $C$ will represent a finite constant whose value may change from line to line, but does not depend upon the discretization level. Any other dependencies in the various parameters considered will be made explicit in the statement of our results.

\subsection{Results on time-discretized diffusion processes\label{app:diff_proc}}

In this section, we consider two diffusion process $X = (X_t)_{t\geq 0}$ and $X^{\star}=(X_t^{\star})_{t\geq 0}$ on the filtered probability space $(\Omega,\mathscr{F},\{\mathscr{F}_t\}_{t\geq 0},\mathbb{P}_{\theta})$ following \eqref{eq:intro_diff}, with the respective initial conditions $X_0=x\in\mathbb{R}^d$ and $X_0^{\star}=x_{\star}\in\mathbb{R}^d$, and driven by the same Brownian motion.
We will consider Euler discretizations \eqref{eq:disc_state} of $(X_t)_{t\geq 0}$ and $(X_t^{\star})_{t\geq 0}$ at some given level $l$, denoted as $\widetilde{X}_{0:T}$ and $\widetilde{X}_{0:T}^{\star}$, driven by the same Brownian motion and with the initial conditions $\widetilde{X}_0=x$ and $\widetilde{X}_0^{\star}=x_{\star}$. 
The expectation operator for the described processes is written as $\mathbb{E}_{\theta}$. 
Although many proof strategies in this section follow \citet{crisan2011discretizing}, we note that many of these arguments will be used in Section~\ref{app:CCPF} to study coupled CPFs.

In addition to the previously defined terms $b_\theta(x) = \Sigma(x)^{-1}\sigma(x)^*a_{\theta}(x)$ and $\Sigma(x)=\sigma(x)\sigma(x)^*$, 
we introduce the function $\rho_{\theta}(x) = b_{\theta}(x)^*\Sigma(x)^{-1}\sigma(x)^*$ which allows us 
to rewrite \eqref{eqn:smoothing_functional} and \eqref{eqn:discrete_test_function} as 
\begin{align}
	G_{\theta}(X) &= -\frac{1}{2}\int_{0}^{T}\nabla_{\theta} \| b_{\theta}(X_t) \|_2^2 dt 
	+ \int_{0}^{T}\nabla_{\theta}\rho_{\theta}(X_{t})dX_{t} + 
	\sum_{t=1}^T\nabla_{\theta} \log g_{\theta}(y_{t}|X_{t}),\\
	G_{\theta}^l(X_{0:T}) 
	= &-\frac{1}{2}\sum_{k=1}^{K_l} \nabla_{\theta} \| b_{\theta}(X_{s_{k-1}}) \|_2^2 \Delta_l 
	+ \sum_{k=1}^{K_l} \nabla_{\theta}\rho_{\theta}(X_{s_{k-1}})(X_{s_{k}} - X_{s_{k-1}}) + 
	\sum_{t=1}^T\nabla_{\theta}\log g_{\theta}(y_{t}|X_{t})\label{eqn:rewrite_discrete_test_function}.
\end{align}
For notational convenience, we define the $d\times 1$ vector of derivatives of $\rho_{\theta}$ as
$$
\kappa_{\theta,i}(x)^* = \Big(\frac{\partial}{\partial \theta_i}[\rho_{\theta}(x)]^1,\dots,
\frac{\partial}{\partial \theta_i}[\rho_{\theta}(x)]^d\Big),
$$
for any $(i,x)\in\{1,\dots,d_{\theta}\}\times\mathbb{R}^d$, 
%where $[\cdot]^i$ denotes the $i$th component of a vector. 
and the conditional likelihood given states $x_1,\dots,x_T \in \mathbb{R}^d$ as $\varphi_{\theta}(x_1,\dots,x_T) = \prod_{t = 1}^T g_{\theta}(y_t | x_t)$. 
We now give the proof of Theorem \ref{prop:conv_grad_log_like} followed by several technical lemmata 
that are required to establish the theorem.

\begin{proof}[Proof of Theorem \ref{prop:conv_grad_log_like}]
We consider the proof for any given component $i\in\{1,\dots,d_{\theta}\}$ and decompose 
the error of the score function \eqref{eqn:discretized_score} at level $l \in \mathbb{N}_0$ as
\begin{equation}\label{eq:theo_main_eq}
[S_l(\theta) - S(\theta)]^i = T_1 + T_2
\end{equation}
where
\begin{align*}
T_1 & = \frac{\mathbb{E}_{\theta}[\varphi_{\theta}(\widetilde{X}_1^{\star},\dots,\widetilde{X}_T^{\star})[G_{\theta}^l(\widetilde{X}_{0:T}^{\star})]^i]}
{\mathbb{E}_{\theta}[\varphi_{\theta}(\widetilde{X}_1^{\star},\dots,\widetilde{X}_T^{\star})]
\mathbb{E}_{\theta}[\varphi_{\theta}(X_1^{\star},\dots,X_T^{\star})]
}\Big(
\mathbb{E}_{\theta}[\varphi_{\theta}(X_1^{\star},\dots,X_T^{\star})] - 
\mathbb{E}_{\theta}[\varphi_{\theta}(\widetilde{X}_1^{\star},\dots,\widetilde{X}_T^{\star})] 
 \Big), \\
T_2 & = \frac{1}{\mathbb{E}_{\theta}[\varphi_{\theta}(X_1^{\star},\dots,X_T^{\star})]}
\Big(
\mathbb{E}_{\theta}[\varphi_{\theta}(\widetilde{X}_1^{\star},\dots,\widetilde{X}_T^{\star})[G_{\theta}^l(\widetilde{X}_{0:T}^{\star})]^i] - %\\ & &
\mathbb{E}_{\theta}[\varphi_{\theta}(X_1^{\star},\dots,X_T^{\star})[G_{\theta}(X^{\star})]^i]
\Big).
\end{align*}
Thus our objective is to provide bounds on the quantities $T_1$ and $T_2$ to conclude the proof.

For $T_1$, using Assumption~\ref{ass:D2}, one has the upper-bound
$$
T_1 \leq C \sum_{t=1}^T \mathbb{E}_{\theta}[\|\widetilde{X}_t^{\star}-X_t^{\star}\|_2],
$$
then by using results on the convergence of Euler approximations (e.g.\ \cite{kloeden2013numerical}), for $r>0$
\begin{equation}\label{eq:eul_conv_strong}
\mathbb{E}_{\theta}\big[\|\widetilde{X}^{\star}_t-X^{\star}_t\|_2^r\big]^{1/r} \leq C\Delta_l^{1/2}
\end{equation}
one has
\begin{equation}\label{eq:t1_bias_bd}
T_1 \leq C \Delta_l^{1/2}.
\end{equation}
Note that using standard results on weak errors for diffusions one can improve this upper-bound to $T_1 \leq C \Delta_l$.

For $T_2$, using Assumption~\ref{ass:D2}, we have $T_2\leq C(T_3+T_4)$ where
\begin{align*}
T_3 & = \mathbb{E}_{\theta}\big[\{\varphi_{\theta}(\widetilde{X}_1^{\star},\dots,\widetilde{X}_T^{\star})-\varphi_{\theta}(X_1^{\star},\dots,X_T^{\star})\}
[G_{\theta}^l(\widetilde{X}_{0:T}^{\star})]^i\big], \\
T_4 & = \mathbb{E}_{\theta}\big[\varphi_{\theta}(X_1^{\star},\dots,X_T^{\star})\{
[G_{\theta}^l(\widetilde{X}_{0:T}^{\star})]^i
-[G_{\theta}(X^{\star})]^i\}\big].
\end{align*}
For $T_3$, using Cauchy-Schwarz, we have the upper-bound
$$
T_3 \leq \mathbb{E}_{\theta}\big[\{\varphi_{\theta}(\widetilde{X}_1^{\star},\dots,\widetilde{X}_T^{\star})-\varphi_{\theta}(X_1^{\star},\dots,X_T^{\star})\}^2\big]^{1/2}
\mathbb{E}_{\theta}\big[\{[G_{\theta}^l(\widetilde{X}_{0:T}^{\star})]^i\}^2\big]^{1/2}.
$$
As the second term is bounded by $C$, we consider only the first. We have the upper-bound
\begin{equation}\label{eq:t3_bias_bd}
T_3 \leq C\sum_{t=1}^T \mathbb{E}_{\theta}\big[\|\widetilde{X}_t^{\star}-X_t^{\star}\|_2^2\big]^{1/2} \leq C\Delta_l^{1/2}.
\end{equation}
For $T_4$, noting that $\varphi_{\theta}$ is a bounded function under Assumption~\ref{ass:D2}, 
applying Lemma \ref{lem:diff1} allows one to conclude that $T_4 \leq C\Delta_l^{1/2}$. 
Therefore, using $T_2\leq C(T_3+T_4)$ along with \eqref{eq:t3_bias_bd}, we have
\begin{equation}\label{eq:t2_bias_bd}
T_2 \leq C \Delta_l^{1/2}.
\end{equation}
Combining \eqref{eq:t1_bias_bd} and \eqref{eq:t2_bias_bd} with \eqref{eq:theo_main_eq} allows one to conclude the proof.
\end{proof}

\begin{lemma}\label{lem:diff1}
Under Assumptions~\ref{ass:D1} and \ref{ass:D2}, for any $(T,r,\theta,i)\in \mathbb{N}\times[1,\infty)\times\Theta\times\{1,\dots,d_{\theta}\}$, there exists a constant $C<\infty$ such that for any $(l,x)\in\mathbb{N}_0\times\mathbb{R}^d$
$$
\mathbb{E}_{\theta}\Big[\big|[G_{\theta}^l(\widetilde{X}_{0:T})]^i-[G_{\theta}(X)]^i\big|^r\Big]^{1/r} \leq C\Delta_l^{1/2},
$$
with $\widetilde{X}_0 = X_0 = x$.
\end{lemma}

\begin{proof}
We have that
\begin{equation}\label{eq:main_lem_main_eq}
\mathbb{E}_{\theta}\Big[\big|[G_{\theta}^l(\widetilde{X}_{0:T})]^i-[G_{\theta}(X)]^i\big|^r\Big] \leq C(T_1+T_2)
\end{equation}
where
\begin{align*}
T_1 & = \mathbb{E}_{\theta}\Big[\big|[G_{\theta}^l(\widetilde{X}_{0:T})]^i-[G^l_{\theta}(X_{0:T})]^i\big|^r\Big], \\
T_2 & = \mathbb{E}_{\theta}\Big[\big|[G_{\theta}^l(X_{0:T})]^i-[G_{\theta}(X)]^i\big|^r\Big],
\end{align*}
where $X_{0:T} = (X_{s_k})_{k=0}^{K_l}$ are the states of the process $(X_t)_{t\geq 0}$ at the discretization times of the process $\widetilde{X}_{0:T}$.
From \eqref{eqn:rewrite_discrete_test_function}, we have that $T_1 \leq C\sum_{j=3}^6 T_j$, where
\begin{align*}
T_3 & = \mathbb{E}_{\theta}\Big[\Big|\sum_{t=1}^T \Big\{\big[\nabla_{\theta} \log g_{\theta}(y_t|\widetilde{X}_t)\big]^i - 
\big[\nabla_{\theta} \log g_{\theta}(y_t|X_t)\big]^i\Big\}\Big|^r\Big], \\
T_4 & = \Delta_l^r\mathbb{E}_{\theta}\Big[\Big|\sum_{k=1}^{K_l}\Big\{\big[\nabla_{\theta}\|b_{\theta}(\widetilde{X}_{s_{k-1}})\|_2^2\big]^i - \big[\nabla_{\theta}\|b_{\theta}(X_{s_{k-1}})\|_2^2\big]^i\Big\}\Big|^r\Big], \\
T_5 & = \mathbb{E}_{\theta}\Big[\Big| \sum_{k=1}^{K_l} \{\kappa_{\theta,i}(\widetilde{X}_{s_{k-1}})^* - \kappa_{\theta,i}(X_{s_{k-1}})^*\} [\widetilde{X}_{s_k}-\widetilde{X}_{s_{k-1}}]\Big|^r\Big], \\
T_6 & = \mathbb{E}_{\theta}\Big[\Big| \sum_{k=1}^{K_l} \kappa_{\theta,i}(X_{s_{k-1}})^* [(\widetilde{X}_{s_k}-\widetilde{X}_{s_{k-1}}) - (X_{s_k}-X_{s_{k-1}}) ]\Big|^r\Big].
\end{align*}
The term $T_3$ can be treated in almost the same manner as $T_1$ in the proof of Theorem~\ref{prop:conv_grad_log_like}, i.e.\ 
using a similar argument to the proof of the bound on $T_1$ in Theorem~\ref{prop:conv_grad_log_like}, one can deduce that 
\begin{equation}\label{eq:t3_bd_lem-b1}
T_3 \leq  C\Delta_l^{r/2}.
\end{equation}

For $T_4$, using the fact that $\partial/\partial\theta_i [b_{\theta}^2]^j\in \textrm{Lip}_{\|\cdot\|_2}(\mathbb{R}^d)$ for any $(i,j)\in\{1,\dots,d_{\theta}\}\times\{1,\dots,d\}$, we have by first applying Minkowski's inequality
$$
T_4 \leq  C\Delta_l^r\Big(\sum_{k=1}^{K_l}\mathbb{E}_{\theta}\big[\|\widetilde{X}_{s_{k-1}}-X_{s_{k-1}}\|_2^r\big]^{1/r}\Big)^r.
$$
Then using \eqref{eq:eul_conv_strong}, it follows that 
\begin{equation}\label{eq:t4_bd_lem-b1}
T_4 \leq  C\Delta_l^{r/2}.
\end{equation}
The terms $T_5$ and $T_6$ are bounded in Lemmata~\ref{lem:diff4}-\ref{lem:diff5}, so combining \eqref{eq:t3_bd_lem-b1}, \eqref{eq:t4_bd_lem-b1} and the afore-mentioned lemmata with $T_1 \leq C\sum_{j=3}^6 T_j$ yields
$$
T_1 \leq  C\Delta_l^{r/2}.
$$
By Lemma~\ref{lem:diff6}, $T_2 \leq  C\Delta_l^{r/2}$ and thus by \eqref{eq:main_lem_main_eq} the proof is concluded.
\end{proof}

\begin{corollary}\label{cor:diff1}
Under Assumptions~\ref{ass:D1} and \ref{ass:D2}, for any $(T,r,\theta)\in \mathbb{N}\times[1,\infty)\times\Theta$, there exists a constant $C<\infty$ such that for any $(l,x)\in\mathbb{N}_0\times\mathbb{R}^d$
$$
\mathbb{E}_{\theta}\Big[ \big\| G_{\theta}^l(\widetilde{X}_{0:T})-G_{\theta}(X) \big\|_2^r \Big]^{1/r} \leq C\Delta_l^{1/2}.
$$
\end{corollary}

\begin{proof}
By Minkowski's inequality
$$
\mathbb{E}_{\theta}\Big[ \big\| G_{\theta}^l(\widetilde{X}_{0:T})-G_{\theta}(X) \big\|_2^r \Big]^{1/r} \leq
\bigg(\sum_{i=1}^{d_{\theta}} \mathbb{E}_{\theta}\Big[ \big| [G_{\theta}^l(\widetilde{X}_{0:T})]^i - [G_{\theta}(X)]^i \big|^r \Big]^{2/r} \bigg)^{1/2}
$$
so the proof follows by Lemma~\ref{lem:diff1}.
\end{proof}

\begin{lemma}\label{lem:diff4}
Under Assumptions~\ref{ass:D1} and \ref{ass:D2}, for any $(T,r,\theta,i)\in \mathbb{N}\times[1,\infty)\times\Theta\times\{1,\dots,d_{\theta}\}$, there exists a constant $C<\infty$ such that for any $(l,x)\in\mathbb{N}_0\times\mathbb{R}^d$
$$
\mathbb{E}_{\theta}\bigg[\Big|
\sum_{k=1}^{K_l}
\{\kappa_{\theta,i}(\widetilde{X}_{s_{k-1}})^*-
 \kappa_{\theta,i}(X_{s_{k-1}})^*\}
 [\widetilde{X}_{s_k}-\widetilde{X}_{s_{k-1}}]\Big|^r\bigg] \leq C\Delta_l^{r/2},
$$
with $\widetilde{X}_0 = X_0 = x$.
\end{lemma}

\begin{proof}
We have the decomposition
$$
\sum_{k=1}^{K_l}
\{\kappa_{\theta,i}(\widetilde{X}_{s_{k-1}})^*-
 \kappa_{\theta,i}(X_{s_{k-1}})^*\}
 [\widetilde{X}_{s_k}-\widetilde{X}_{s_{k-1}}] = M_{K_l} + R_{K_l},
$$
where
\begin{align*}
M_{K_l} & = \sum_{k=1}^{K_l}\{\kappa_{\theta,i}(\widetilde{X}_{s_{k-1}})^*-
 \kappa_{\theta,i}(X_{s_{k-1}})^*\}\sigma(\widetilde{X}_{s_{k-1}})[W_{s_k}-W_{s_{k-1}}],\\
R_{K_l} & = \Delta_l\sum_{k=1}^{K_l}
\{\kappa_{\theta,i}(\widetilde{X}_{s_{k-1}})^*-
 \kappa_{\theta,i}(X_{s_{k-1}})^*\}a_{\theta}(\widetilde{X}_{s_{k-1}}).
\end{align*}
Thus by the $C_r-$inequality, 
 \begin{equation}\label{eq:diff_mart_prf_lem6}
 \mathbb{E}_{\theta}\bigg[\Big|
\sum_{k=1}^{K_l}
\{\kappa_{\theta,i}(\widetilde{X}_{s_{k-1}})^*-
 \kappa_{\theta,i}(X_{s_{k-1}})^*\}
 [\widetilde{X}_{s_k}-\widetilde{X}_{s_{k-1}}] \Big|^r\bigg] \leq C\big( \mathbb{E}_{\theta}[|M_{K_l}|^r]+\mathbb{E}_{\theta}[|R_{K_l}|^r]\big).
 \end{equation}
 We will bound the two terms on the R.H.S.~of \eqref{eq:diff_mart_prf_lem6} individually.
\paragraph{Bound for $\mathbb{E}_{\theta}[|M_{K_l}|^r]$.} If we define $M_{u} =  \sum_{k=0}^{u-1}\{\kappa_{\theta,i}(\widetilde{X}_{s_{k-1}})^*-
 \kappa_{\theta,i}(X_{s_{k-1}})^*\}\sigma(\widetilde{X}_{s_{k-1}})[W_{s_k}-W_{s_{k-1}}]$ for any $u \in \mathbb{N}_0$,
 then $(M_u,\mathscr{F}_{u\Delta_l})_{u\in \mathbb{N}_0}$ is a martingale. It follows from this fact and from an application of the Burkholder-Gundy-Davis (BGD) inequality that
 $$
 \mathbb{E}_{\theta}[|M_{K_l}|^r] \leq C
  \mathbb{E}_{\theta}\bigg[\Big|\sum_{k=1}^{K_l}
  (\{\kappa_{\theta,i}(\widetilde{X}_{s_{k-1}})^*-
 \kappa_{\theta,i}(X_{s_{k-1}})^*\}\sigma(\widetilde{X}_{s_{k-1}})[W_{s_k}-W_{s_{k-1}}])^2
  \Big|^{r/2}\bigg],
 $$
from which Minkowski's inequality yields
$$
 \mathbb{E}_{\theta}[|M_{K_l}|^r] \leq C
\bigg(
\sum_{k=1}^{K_l}  \mathbb{E}_{\theta}\Big[\big|
\{\kappa_{\theta,i}(\widetilde{X}_{s_{k-1}})^*-
 \kappa_{\theta,i}(X_{s_{k-1}})^*\}\sigma(\widetilde{X}_{s_{k-1}})[W_{s_k}-W_{s_{k-1}}]\big|^r\Big]^{2/r}
\bigg)^{r/2}.
$$
Using the $C_r-$inequality $d^2$ times, we obtain the bound
\begin{multline*}
 \mathbb{E}_{\theta}[|M_{K_l}|^r] \leq \\
C\Bigg( \sum_{k=1}^{K_l} \bigg( \sum_{(m,j)\in\{1,\dots,d\}^2}
 \mathbb{E}_{\theta}\Big[\big|\{\kappa_{\theta,i}(\widetilde{X}_{s_{k-1}}) - \kappa_{\theta,i}(X_{s_{k-1}})\}^m
\sigma(\widetilde{X}_{s_{k-1}})^{m,j}[W_{s_k}-W_{s_{k-1}}]^j\big|^r\Big]
\bigg)^{2/r}\Bigg)^{r/2}.
\end{multline*}
Using the fact that $\sigma^{m,j}\in\mathcal{B}_b(\mathbb{R}^d)$ along with the Cauchy-Schwarz inequality yields
\begin{multline}
\label{eq:diff_mart_prf_lem1}
\mathbb{E}_{\theta}[|M_{K_l}|^r] \leq \\
C\Bigg( \sum_{k=1}^{K_l}\bigg(\sum_{(m,j)\in\{1,\dots,d\}^2}
  \mathbb{E}_{\theta}\Big[ \big| \{\kappa_{\theta,i}(\widetilde{X}_{s_{k-1}})-
 \kappa_{\theta,i}(X_{s_{k-1}})\}^m \big|^{2r} \Big]^{1/2}
   \mathbb{E}_{\theta}\Big[ \big|[W_{s_k}-W_{s_{k-1}}]^j\big|^{2r} \Big]^{1/2}
\bigg)^{2/r} \Bigg)^{r/2}.
 \end{multline}
Since it holds that $[\kappa_{\theta,i}]^m\in\textrm{Lip}_{\|\cdot\|}(\mathbb{R}^d)$, it follows from the same type of inequality as \eqref{eq:eul_conv_strong} that
\begin{equation}\label{eq:diff_mart_prf_lem2}
\mathbb{E}_{\theta}\Big[ \big| \{\kappa_{\theta,i}(\widetilde{X}_{s_{k-1}}) - \kappa_{\theta,i}(X_{s_{k-1}})\}^m \big|^{2r} \Big]^{1/2} \leq C\Delta_l^{r/2},
\end{equation}
and by standard properties of Brownian motion, we obtain
\begin{equation}\label{eq:diff_mart_prf_lem3}
\mathbb{E}_{\theta}\Big[ \big| [W_{s_k}-W_{s_{k-1}}
]^j \big|^{2r} \Big]^{1/2}
\leq C\Delta_l^{r/2}.
\end{equation}
 Combining \eqref{eq:diff_mart_prf_lem1} with \eqref{eq:diff_mart_prf_lem2} and \eqref{eq:diff_mart_prf_lem3} yields the upper-bound
  \begin{equation}\label{eq:diff_mart_prf_lem4}
 \mathbb{E}_{\theta}[|M_{K_l}|^r] \leq C\Delta_l^{r/2}.
  \end{equation}
\paragraph{Bound for $\mathbb{E}_{\theta}[|R_{K_l}|^r]$.} We have the upper-bound by Minkowski's inequality
$$ 
\mathbb{E}_{\theta}[|R_{K_l}|^r] \leq \Delta_l^r\bigg(\sum_{k=1}^{K_l}
\mathbb{E}_{\theta}\Big[\big|\{\kappa_{\theta,i}(\widetilde{X}_{s_{k-1}})^*-
 \kappa_{\theta,i}(X_{s_{k-1}})^*\}a_{\theta}(\widetilde{X}_{s_{k-1}}) \big|^r \Big]^{1/r}
\bigg)^r.
$$
Then applying the $C_r-$inequality $d$ times and using the assumption that $a_{\theta}^j\in\mathcal{B}_b(\mathbb{R}^d)$ we have the upper-bound
$$ 
\mathbb{E}_{\theta}[|R_{K_l}|^r] \leq C\Delta_l^r\bigg(\sum_{k=1}^{K_l}
\sum_{j=1}^d 
   \mathbb{E}_{\theta}\Big[\big|\{\kappa_{\theta,i}(\widetilde{X}_{s_{k-1}})-
 \kappa_{\theta,i}(X_{s_{k-1}})\}^j \big|^{r} \Big]^{1/r} \bigg)^r.
$$
Using the same argument to obtain \eqref{eq:diff_mart_prf_lem2}, we have
 \begin{equation}\label{eq:diff_mart_prf_lem5}
 \mathbb{E}_{\theta}[|R_{K_l}|^r] \leq C\Delta_l^r\Delta_l^{-r/2} =C\Delta_l^{r/2}.
  \end{equation}

Combining \eqref{eq:diff_mart_prf_lem4} and \eqref{eq:diff_mart_prf_lem5} with \eqref{eq:diff_mart_prf_lem6} allows us to conclude.
\end{proof}

\begin{lemma}\label{lem:diff5}
Under Assumptions~\ref{ass:D1} and \ref{ass:D2}, for any $(T,r,\theta,i)\in \mathbb{N}\times[1,\infty)\times\Theta\times\{1,\dots,d_{\theta}\}$, there exists a constant $C<\infty$ such that for any $(l,x)\in\mathbb{N}_0\times\mathbb{R}^d$
$$
\mathbb{E}_{\theta}\bigg[\Big| \sum_{k=1}^{K_l} \kappa_{\theta,i}(X_{s_{k-1}})^* \big[(\widetilde{X}_{s_k}-\widetilde{X}_{s_{k-1}})- (X_{s_k}-X_{s_{k-1}}) \big]\Big|^r\bigg] \leq C\Delta_l^{r/2},
$$
with $\widetilde{X}_0 = X_0 = x$.
\end{lemma}

\begin{proof}
We have the decomposition
$$
\sum_{k=1}^{K_l} \kappa_{\theta,i}(X_{s_{k-1}})^* \big[(\widetilde{X}_{s_k}-\widetilde{X}_{s_{k-1}})- (X_{s_k}-X_{s_{k-1}}) \big] = M_{K_l} + R_{K_l},
$$
where
\begin{align*}
M_{K_l} & = \sum_{k=1}^{K_l}\sum_{(m,j)\in\{1,\dots,d\}^2}\kappa_{\theta,i}(X_{s_{k-1}})^m \int_{s_{k-1}}^{s_k}[\sigma(\widetilde{X}_{s_{k-1}})-\sigma(X_s)]^{m,j} dW_s^j,\\
R_{K_l} & = \sum_{k=1}^{K_l}\sum_{j=1}^d \kappa_{\theta,i}(X_{s_{k-1}})^j \int_{s_{k-1}}^{s_k} \big[a_{\theta}(\widetilde{X}_{s_{k-1}})-a_{\theta}(X_s) \big]^j ds.
\end{align*}
Thus by the $C_r-$inequality, 
\begin{equation}\label{eq:diff_mart_prf_lem7}
\mathbb{E}_{\theta}\bigg[ \Big|
\sum_{k=1}^{K_l} \kappa_{\theta,i}(X_{s_{k-1}})^*
\big[ (\widetilde{X}_{s_k}-\widetilde{X}_{s_{k-1}}) - (X_{s_k}-X_{s_{k-1}}) \big]
\Big|^r \bigg] \leq C\big( \mathbb{E}_{\theta}[|M_{K_l}|^r]+\mathbb{E}_{\theta}[|R_{K_l}|^r] \big).
\end{equation}
We will bound the two terms on the R.H.S.~of \eqref{eq:diff_mart_prf_lem7} individually.

\paragraph{Bound for $\mathbb{E}_{\theta}[|M_{K_l}|^r]$.} By applying the $C_r-$inequality $d^2$ times we have the upper-bound
$$
\mathbb{E}_{\theta}[|M_{K_l}|^r] \leq C \sum_{(m,j)\in\{1,\dots,d\}^2} \mathbb{E}_{\theta}\bigg[\Big|
\sum_{k=1}^{K_l} \kappa_{\theta,i}(X_{s_{k-1}})^m \int_{s_{k-1}}^{s_k}[\sigma(\widetilde{X}_{s_{k-1}})-\sigma(X_s)]^{m,j}dW_s^j \Big|^r\bigg].
$$
For any $(m,j) \in \{1,\dots,d\}^2$, we define $M_u^{m,j} = \sum_{k=0}^{u-1} \kappa_{\theta,i}(X_{s_{k-1}})^m \int_{s_{k-1}}^{s_k}[\sigma(\widetilde{X}_{s_{k-1}})-\sigma(X_s)]^{m,j}dW_s^j$ for $u \in \mathbb{N}_0$. 
As $(M_u^{m,j},\mathscr{F}_{u\Delta_l})_{u\in \mathbb{N}_0}$ is a martingale, applying the BGD inequality yields
$$
\mathbb{E}_{\theta}[|M_{K_l}|^r] \leq C \sum_{(m,j)\in\{1,\dots,d\}^2} \mathbb{E}_{\theta} \bigg[ \Big| \sum_{k=1}^{K_l} \Big(\kappa_{\theta,i}(X_{s_{k-1}})^m \int_{s_{k-1}}^{s_k}[\sigma(\widetilde{X}_{s_{k-1}})-\sigma(X_s)]^{m,j}dW_s^j\Big)^2 \Big|^{r/2}\bigg].
$$
Applying Minkowski's inequality and using the fact that $[\kappa_{\theta,i}]^m\in\mathcal{B}_b(\mathbb{R}^d)$, we obtain
 \begin{equation}\label{eq:diff_mart_prf_lem8}
\mathbb{E}_{\theta}[|M_{K_l}|^r] \leq C \sum_{(m,j)\in\{1,\dots,d\}^2} \Bigg( \sum_{k=1}^{K_l}
\mathbb{E}_{\theta}\bigg[\Big|\int_{s_{k-1}}^{s_k}[\sigma(\widetilde{X}_{s_{k-1}})-\sigma(X_s)]^{m,j}dW_s^j
\Big|^{r}\bigg]^{2/r} \Bigg)^{r/2}.
 \end{equation}
 Now we deal with the expectation on the R.H.S.~of \eqref{eq:diff_mart_prf_lem8}. Using the martingale property of the stochastic integral, it follows from applying the BGD inequality again that
\begin{align*}
\mathbb{E}_{\theta}\bigg[\Big|\int_{s_{k-1}}^{s_k}[\sigma(\widetilde{X}_{s_{k-1}})-\sigma(X_s)]^{m,j}dW_s^j
\Big|^{r}\bigg] & \leq
C\mathbb{E}_{\theta}\bigg[\Big|\int_{s_{k-1}}^{s_k} \big\{ [\sigma(\widetilde{X}_{s_{k-1}})-\sigma(X_s)]^{m,j} \big\}^2 ds
\Big|^{r/2}\bigg] \\
& \leq C\Delta_l^{r/2-1}\mathbb{E}_{\theta}\Big[\int_{s_{k-1}}^{s_k} \big|[\sigma(\widetilde{X}_{s_{k-1}})-\sigma(X_s)]^{m,j} \big|^r ds\Big],
\end{align*}
where we have used Jensen's inequality in the second line. Using the $C_r-$inequality, we have
\begin{multline*}
\mathbb{E}_{\theta}\bigg[\Big|\int_{s_{k-1}}^{s_k}[\sigma(\widetilde{X}_{s_{k-1}})-\sigma(X_s)]^{m,j}dW_s^j
\Big|^{r}\bigg] \leq \\
C\Delta_l^{r/2-1}\int_{s_{k-1}}^{s_k}
 \Big\{ \mathbb{E}_{\theta}\Big[ \big| [\sigma(\widetilde{X}_{s_{k-1}})-\sigma(X_{{s_{k-1}}})]^{m,j} \big|^r \Big]
+ \mathbb{E}_{\theta}\Big[ \big| [\sigma(X_{{s_{k-1}}})-\sigma(X_s)]^{m,j} \big|^r \Big]
\Big\}ds.
\end{multline*}
Using the fact that $\sigma^{m,j}\in\textrm{Lip}_{\|\cdot\|_2}(\mathbb{R}^d)$, we then obtain 
\begin{multline*}
\mathbb{E}_{\theta}\bigg[\Big|\int_{s_{k-1}}^{s_k}[\sigma(\widetilde{X}_{s_{k-1}})-\sigma(X_s)]^{m,j}dW_s^j
\Big|^{r}\bigg] \leq \\
C\Delta_l^{r/2-1}\int_{s_{k-1}}^{s_k} \Big\{\mathbb{E}_{\theta}
\big[ \|
\widetilde{X}_{s_{k-1}}-X_{{s_{k-1}}}\|_2^r \big] + \mathbb{E}_{\theta}\big[ \|X_{{s_{k-1}}}-X_s\|_2^r \big] \Big\}ds.
\end{multline*}
By the property \eqref{eq:eul_conv_strong} for $r>0$ (e.g.\ \cite{ikeda2014stochastic}), it follows that
\begin{equation}\label{eq:diff_cont_l_r}
  \sup_{(t,s)\in[0,T]^2}\mathbb{E}_{\theta} \big[ \|X_t-X_s\|_2^r \big] \leq C|t-s|^{r/2},
\end{equation}
hence we obtain the upper-bound
 \begin{equation}\label{eq:diff_mart_prf_lem9}
  \mathbb{E}_{\theta}\bigg[\Big|\int_{s_{k-1}}^{s_k} \big[ \sigma(\widetilde{X}_{s_{k-1}})-\sigma(X_s) \big]^{m,j} dW_s^j
\Big|^{r}\bigg] \leq C\Delta_l^{r/2-1}\Delta_l^{1+r/2}= C\Delta_l^r.
  \end{equation}
Combining \eqref{eq:diff_mart_prf_lem9} with \eqref{eq:diff_mart_prf_lem8} gives
 \begin{equation}\label{eq:diff_mart_prf_lem10}  
 \mathbb{E}_{\theta}[|M_{K_l}|^r] \leq C\Delta_l^{r/2}.
  \end{equation}

\paragraph{Bound for $\mathbb{E}_{\theta}[|R_{K_l}|^r]$.} Using Minkowski's inequality followed by Jensen's inequality, 
we obtain the following bounds
\begin{align*}
\mathbb{E}_{\theta}[|R_{K_l}|^r] & \leq 
\Bigg( \sum_{k=1}^{K_l} \sum_{j=1}^d 
\mathbb{E}_{\theta} \bigg[ \Big| \kappa_{\theta,i}(X_{s_{k-1}})^j
\int_{s_{k-1}}^{s_k}[a_{\theta}(\widetilde{X}_{s_{k-1}})-a_{\theta}(X_s)]^jds \Big|^r \bigg]^{1/r} \Bigg)^r\\
& \leq 
\Bigg(\sum_{k=1}^{K_l}\sum_{j=1}^d \Delta_l^{1-1/r} \mathbb{E}_{\theta}\Big[
\big| \kappa_{\theta,i}(X_{s_{k-1}})^j \big|^r
\int_{s_{k-1}}^{s_k}\big| [a_{\theta}(\widetilde{X}_{s_{k-1}})-a_{\theta}(X_s)]^j \big|^rds
\Big]^{1/r} \Bigg)^r.
\end{align*}
Using the assumption $[\kappa_{\theta,i}]^j\in\mathcal{B}_b(\mathbb{R}^d)$, the decomposition
$$
[a_{\theta}(\widetilde{X}_{s_{k-1}})-a_{\theta}(X_s)]^j = 
[a_{\theta}(\widetilde{X}_{s_{k-1}})-a_{\theta}(X_{s_{k-1}})]^j + [a_{\theta}(X_{s_{k-1}})-a_{\theta}(X_s)]^j,
$$
and the $C_r-$inequality, we have 
\begin{align*}
& \mathbb{E}_{\theta}[|R_{K_l}|^r] \leq \\
& \Bigg(\sum_{k=1}^{K_l} \sum_{j=1}^d C\Delta_l^{1-1/r} \bigg( \int_{s_{k-1}}^{s_k} \mathbb{E}_{\theta}\Big[ \big| [a_{\theta}(\widetilde{X}_{s_{k-1}})-a_{\theta}(X_{s_{k-1}})]^j \big|^r \Big] ds + \int_{s_{k-1}}^{s_k} \mathbb{E}_{\theta}\Big[ \big| [a_{\theta}(X_{s_{k-1}})-a_{\theta}(X_s)]^j \big|^r \Big] ds \bigg)^{1/r} \Bigg)^r.
\end{align*}
Since we have assumed that $a_{\theta}^j\in\textrm{Lip}_{\|\cdot\|_2}(\mathbb{R}^d)$, one has
\begin{equation*}
\mathbb{E}_{\theta}[|R_{K_l}|^r] \leq \Bigg(\sum_{k=1}^{K_l}\sum_{j=1}^d C\Delta_l^{1-1/r} \bigg(
\int_{s_{k-1}}^{s_k} \mathbb{E}_{\theta} \big[ \|\widetilde{X}_{s_{k-1}}-X_{s_{k-1}}\|_2^r \big] ds + \int_{s_{k-1}}^{s_k} \mathbb{E}_{\theta} \big[ \|X_{s_{k-1}}-X_s\|_2^r \big] ds \bigg)^{1/r} \Bigg)^r.
\end{equation*}
Using \eqref{eq:eul_conv_strong}  and \eqref{eq:diff_cont_l_r}, we therefore obtain
\begin{equation}\label{eq:diff_mart_prf_lem11}  
\mathbb{E}_{\theta}[|R_{K_l}|^r] \leq C\Delta_l^{r/2}.
\end{equation}
Combining \eqref{eq:diff_mart_prf_lem10} and \eqref{eq:diff_mart_prf_lem11} with \eqref{eq:diff_mart_prf_lem7} allows us to conclude.
\end{proof}

\begin{lemma}\label{lem:diff6}
Under Assumptions~\ref{ass:D1} and \ref{ass:D2}, for any $(T,r,\theta,i)\in \mathbb{N}\times[1,\infty)\times\Theta\times\{1,\dots,d_{\theta}\}$, there exists a constant $C<\infty$ such that for any $(l,x)\in\mathbb{N}_0\times\mathbb{R}^d$
$$
\mathbb{E}_{\theta}\Big[\big|[G_{\theta}^l(X_{0:T})]^i - [G_{\theta}(X)]^i\big|^r \Big] \leq C\Delta_l^{r/2},
$$
with $\widetilde{X}_0 = X_0 = x$.
\end{lemma}

\begin{proof}
We have the decomposition
$$
\big| [G_{\theta}^l(X_{0:T})]^i - [G_{\theta}(X)]^i \big| = R_{K_l}^{(1)} + M_{K_l} 
+ R_{K_l}^{(2)},
$$
where
\begin{align*}
R_{K_l}^{(1)} & = \sum_{k=1}^{K_l}\int_{s_{k-1}}^{s_k}
\Big([\nabla_{\theta}\|b_{\theta}(X_{s_{k-1}})\|_2^2]^i - [\nabla_{\theta}\|b_{\theta}(X_{s})\|_2^2]^i\Big)
ds, \\
M_{K_l} & = \sum_{k=1}^{K_l}\sum_{(m,j)\in\{1,\dots,d\}^2}
\int_{s_{k-1}}^{s_k}[\kappa_{\theta,i}(X_{s_{k-1}})^m\sigma(\widetilde{X}_{s_{k-1}})-
\kappa_{\theta,i}(X_{s})^m
\sigma(X_s)]^{m,j}dW_s^j, \\
R_{K_l}^{(2)} & = \sum_{k=1}^{K_l}\sum_{j=1}^d\int_{s_{k-1}}^{s_k}
\Big(\kappa_{\theta,i}(X_{s_{k-1}})^ja_{\theta}(X_{s_{k-1}})^j -
\kappa_{\theta,i}(X_{s})^ja_{\theta}(X_{s})^j\Big)ds.
\end{align*}
Thus, by the $C_r-$inequality, we have 
$$
\mathbb{E}_{\theta}\Big[ \big| [G_{\theta}^l(X_{0:T})]^i - [G_{\theta}(X)]^i \big|^r \Big] \leq
C\Big( \mathbb{E}_{\theta}\big[ |M_{K_l}|^r \big]+\mathbb{E}_{\theta}\big[ \big| R_{K_l}^{(1)} \big|^r \big]+\mathbb{E}_{\theta}\big[ \big| R_{K_l}^{(2)}\big|^r \big]
\Big).
$$
In order to prove that $\mathbb{E}_{\theta}[|M_{K_l}|^r]\leq C\Delta_l^{r/2}$, one can rely on very similar arguments to \eqref{eq:diff_mart_prf_lem10} in the proof of Lemma~\ref{lem:diff5}. Likewise, for $m\in\{1,2\}$ one can prove that $\mathbb{E}_{\theta}[|R_{K_l}^{(m)}|^r]\leq C\Delta_l^{r/2}$ using similar arguments to \eqref{eq:diff_mart_prf_lem11} in the proof of Lemma~\ref{lem:diff5}.
\end{proof}

\begin{lemma}\label{lem:diff2}
Under Assumptions~\ref{ass:D1} and \ref{ass:D2}, for any $(T,r,\theta)\in \mathbb{N}\times[1,\infty)\times\Theta$, there exists a constant $C<\infty$ such that for any $(x,x_{\star})\in\mathbb{R}^{d}\times\mathbb{R}^{d}$
$$
\mathbb{E}_{\theta}\big[\| G_{\theta}(X)-G_{\theta}(X^{\star}) \|_2^r \big]^{1/r} \leq C\|x-x_{\star}\|_2,
$$
with $X_0 = x$ and $X_0^{\star} = x_{\star}$.
\end{lemma}

\begin{proof}
This proof follows a similar type of arguments to those considered in the proofs of Lemmata \ref{lem:diff5}-\ref{lem:diff6}. The main difference is that one must use the following result (which can be deduced from \cite[Corollary v.11.7]{rogers2000diffusions} and Gr\"onwall's inequality)
\begin{equation*}%\label{eq:cont_exp_diff}
\sup_{t\in[0,T]}\mathbb{E}_{\theta}[\|X_t-X_t^{\star}\|_2^{2r}]^{\frac{1}{2r}} \leq C\|x-x_{\star}\|_2
\end{equation*}
instead of using Euler convergence. Given that the proofs of Lemmata~\ref{lem:diff5}-\ref{lem:diff6} are repetitive, these arguments are omitted.
\end{proof}

\begin{lemma}\label{lem:diff3}
Under Assumptions~\ref{ass:D1} and \ref{ass:D2}, for any $(T,r,\theta)\in \mathbb{N}\times[1,\infty)\times\Theta$, there exists a constant $C<\infty$ such that for any $(l,x,x_{\star})\in\mathbb{N}\times\mathbb{R}^{d}\times\mathbb{R}^{d}$
$$
\mathbb{E}_{\theta}\big[ \|G_{\theta}^l(\widetilde{X}_{0:T})- G_{\theta}^{l-1}(\widehat{X}^{\star}_{0:T}) \|_2^r \big]^{1/r}
\leq C\big(\Delta_l^{1/2} + \|x - x_{\star}\|_2\big),
$$
with $\widetilde{X}_0 = x$ and $\widetilde{X}_0^{\star} = x_{\star}$, 
where $\widehat{X}^{\star}_{0:T} = (\widetilde{X}^{\star}_{s_{2k}})_{k=0}^{K_{l-1}}$ are 
the sequence of states in $\widetilde{X}^{\star}_{0:T}=(\widetilde{X}^{\star}_{s_k})_{k=0}^{K_l}$ at the time steps corresponding to level $l-1$.
\end{lemma}

\begin{proof}
The expectation in the statement of the lemma is upper-bounded by $\sum_{j=1}^3 T_j$, where
\begin{align*}
T_1 & = \mathbb{E}_{\theta}\big[\|G_{\theta}^l(\widetilde{X}_{0:T})-G_{\theta}(X)\|_2^r\big]^{1/r}, \\
T_2 & = \mathbb{E}_{\theta}\big[\|G_{\theta}(X) - G_{\theta}(X^{\star})\|_2^r\big]^{1/r}, \\
T_3 & = \mathbb{E}_{\theta}\big[\|G_{\theta}^{l-1}
(\widehat{X}^{\star}_{0:T})-G_{\theta}(X^{\star})\|_2^r\big]^{1/r}.
\end{align*}
For $T_1$ and $T_3$, one can apply Corollary \ref{cor:diff1}; for $T_2$, we use Lemma \ref{lem:diff2}. 
This allows us to conclude the result.
\end{proof}

We now introduce two additional functions which will be useful in the following section. 
For $(t,l)\in\{1,\dots,T\}\times\mathbb{N}_0$, we define 
$G_{\cdot,t}^l:\Theta\times(\mathbb{R}^d)^{2^l t+1}\rightarrow\mathbb{R}^{d_{\theta}}$ as 
$$
G_{\theta,t}^l(X_{0:t}) =  -\frac{1}{2}
\sum_{k=1}^{2^lt} \nabla_{\theta}\|b_{\theta}(X_{s_{k-1}})\|_2^2\Delta_l
+ \sum_{k=1}^{2^lt} \nabla_{\theta}\rho_{\theta}(X_{s_{k-1}})
(X_{s_k}-X_{s_{k-1}})+\sum_{p=1}^t \nabla_{\theta} \log g_{\theta}(y_p|X_p),
$$
and $G_{\cdot,t-1:t}^l:\Theta \times \mathbb{R}^d \times (\mathbb{R}^d)^{2^l} \rightarrow \mathbb{R}^{d_{\theta}}$ as 
\begin{align*}
&G_{\theta,t-1:t}^l(X_{t-1},X_{t-1+\Delta_l:t}) \\
&=  -\frac{1}{2}
\sum_{k=0}^{2^l-1} \nabla_{\theta}\|b_{\theta}(X_{t-1+k\Delta_l})\|_2^2\Delta_l
+ \sum_{k=0}^{2^l-1} \nabla_{\theta}\rho_{\theta}(X_{t-1+k\Delta_l})
(X_{t-1+(k+1)\Delta_l}-X_{t-1+k\Delta_l}) + \nabla_{\theta} \log g_{\theta}(y_t|X_t).
\end{align*}
From \eqref{eqn:rewrite_discrete_test_function}, we note that $G_{\theta,T}^l(X_{0:T})=G_{\theta}^l(X_{0:T})$. 
The following remarks will facilitate our proofs in the next section. 

\begin{remark}\label{rem:euler_grad}
One can easily extend Lemma \ref{lem:diff3} as follows. Under Assumptions~\ref{ass:D1} and \ref{ass:D2}, 
for any $(t,r,\theta)\in \{1,\dots,T\}\times[1,\infty)\times\Theta$, there exists a constant $C<\infty$ such that for any 
$(l,x,x_{\star})\in\mathbb{N}\times\mathbb{R}^{d}\times\mathbb{R}^{d}$
$$
\mathbb{E}_{\theta}\Big[ \big\|G_{\theta,t}^l(\widetilde{X}_{0:t})-
G_{\theta,t}^{l-1}(\widehat{X}_{0:t}^{\star})
\big\|_2^r \Big]^{1/r} \leq C\big(\Delta_l^{1/2} + \|x-x_{\star}\|_2\big),
$$
with $\widetilde{X}_0 = x$ and $\widetilde{X}_0^{\star} = x_{\star}$, 
where $\widehat{X}^{\star}_{0:t} = (\widetilde{X}^{\star}_{s_{2k}})_{k=0}^{2^{l-1}t}$.
\end{remark}

\begin{remark}\label{rem:euler_grad1}
One can also deduce the following result. Under Assumptions~\ref{ass:D1} and \ref{ass:D2}, 
for any $(t,r,\theta)\in \{1,\dots,T\}\times[1,\infty)\times\Theta$, there exists a constant $C<\infty$ 
such that for any $(l,x,x_{\star})\in\mathbb{N}\times\mathbb{R}^{d}\times\mathbb{R}^{d}$
$$
\mathbb{E}_{\theta}\Big[ \big\|G_{\theta,t-1:t}^l(x,\widetilde{X}_{t-1+\Delta_l:t})-
G_{\theta,t-1:t}^{l-1}(x_{\star},\widehat{X}_{t-1+\Delta_{l-1}:t}^{\star})
\big\|_2^r \Big| \widetilde{X}_{t-1}=x,\widetilde{X}_{t-1}^{\star}=x_{\star} \Big]^{1/r}
 \leq  C\big(\Delta_l^{1/2} + \|x-x_{\star}\|_2 \big),
$$
where $\widehat{X}_{t-1+\Delta_{l-1}:t}^{\star} = (\widetilde{X}^{\star}_{t-1+\Delta_{l-1}})_{k=1}^{2^{l-1}}$.
\end{remark}

%%%%%%%%%%%%%%%%%%%%%%%%%%%%%%%%%%%%%%%%%%%%%%%%%%%%%%%%%%%%%%%%%%%%%%%%%%%%%%%%%%%%%%%%%%%%%%%

\subsection{Results on coupled conditional particle filters\label{app:CCPF}}

We begin with some definitions associated to Algorithm~\ref{alg:4-CCPF}. % \ref{alg:cond_pf_l}. 
%We extend the definition of the resampled indexes and let for $(t,s)\in\{0,\dots,T-1\}\times\{l,l-1\}$ $A_t^{s,N} = N$ and $\bar{A}_t^{s,N}=N$.
For any $(i,t,s)\in\{1,\dots,N\}\times\mathbb{N}_0\times\{l-1,l\}$, 
we will write $A_t^{s}(i) = A_t^{s,i}$ and $\bar{A}_t^{s}(i) = \bar{A}_t^{s,i}$. 
Using this notation, for any $(t,l)\in\{0,\dots,T-1\}\times\mathbb{N}$, we define
\begin{align*}
& \mathsf{S}_t^l=\{i\in\{1,\ldots, N\}: A_{t}^l(i)=A_{t}^{l-1}(i),A_{t-1}^l\circ A_{t}^l(i)=A_{t-1}^{l-1}\circ A_{t}^{l-1}(i),\dots,%\\
	A_{0}^l\circ\cdots\circ A_{t}^l(i)=A_{0}^{l-1}\circ\cdots\circ A_{t}^{l-1}(i)\},\\
& \bar{\mathsf{S}}_t^l=\{i\in\{1,\ldots, N\}: \bar{A}_{t}^l(i)=\bar{A}_{t}^{l-1}(i),\bar{A}_{t-1}^l\circ \bar{A}_{t}^l(i)=\bar{A}_{t-1}^{l-1}\circ \bar{A}_{t}^{l-1}(i),\dots,%\\
	\bar{A}_{0}^l\circ\cdots\circ \bar{A}_{t}^l(i)=\bar{A}_{0}^{l-1}\circ\cdots\circ \bar{A}_{t}^{l-1}(i)\}.	
\end{align*}
For $(l,\beta,C)\in\mathbb{N}\times\mathbb{R}^+\times\mathbb{R}^+$, we also introduce the following sets 
\begin{align*}
\mathsf{B}^l_{\beta,C} & = \{(x_{0:T},\bar{x}_{0:T})\in\mathsf{X}^l\times\mathsf{X}^{l-1}:%(\mathbb{R}^d)^n\times (\mathbb{R}^d)^n:
\|x_t-\bar{x}_t\|_2 \leq C\Delta_l^{\beta}, t\in\{1,\dots,T\}\},\\
\mathsf{G}^l_{\beta,C} & = \{(x_{0:T},\bar{x}_{0:T})\in\mathsf{X}^l\times\mathsf{X}^{l-1}:
\|G_{\theta,t}^l(x_{0:t}) - G_{\theta,t}^{l-1}(\bar{x}_{0:t})\|_2 \leq C\Delta_l^{\beta},t\in\{1,\dots,T\}\}.
\end{align*}

\subsubsection{Results associated to Steps (1) and (2) of Algorithm~\ref{alg:4-CCPF}}
\label{sec:analysis_first}

We consider Steps (1) and (2) of Algorithm~\ref{alg:4-CCPF} where the input pairs of trajectories 
are taken as $(X_{0:T}^{l-1,\star},\bar{X}_{0:T}^{l-1,\star})= (x_{0:T}^{l-1},\bar{x}_{0:T}^{l-1}) \in \mathsf{Z}^{l-1}$ and $(X_{0:T}^{l,\star},\bar{X}_{0:T}^{l,\star}) = (x_{0:T}^{l},\bar{x}_{0:T}^{l}) \in \mathsf{Z}^{l}$. In order to analyze the algorithm, it is useful to define the simulated trajectories recursively at time steps $t \in \{1,\dots,T\}$, for any $i \in \{1,\dots,N\}$, as 
\begin{align*}
(X_{0:t}^{l-1,i},\bar{X}_{0:t}^{l-1,i}) & = \Big((X_{0:t-1}^{l-1,A_{t-1}^{l-1,i}},X_{t-1+\Delta_{l-1}:t}^{l-1,i}),(\bar{X}_{0:t-1}^{l-1,\bar{A}_{t-1}^{l-1,i}},\bar{X}_{t-1+\Delta_{l-1}:t}^{l-1,i}) \Big),\\
(X_{0:t}^{l,i},\bar{X}_{0:t}^{l,i}) & = \Big((X_{0:t-1}^{l,A_{t-1}^{l,i}},X_{t-1+\Delta_l:t}^{l,i}),(\bar{X}_{0:t-1}^{l,\bar{A}_{t-1}^{l,i}},\bar{X}_{t-1+\Delta_l:t}^{l,i})\Big).
\end{align*}
After the completion of Step (2), we consider the output given by the two collections of pairs of trajectories 
$(X_{0:T}^{l-1,i},\bar{X}_{0:T}^{l-1,i})_{i=1}^N$ and $(X_{0:T}^{l,i},\bar{X}_{0:T}^{l,i})_{i=1}^N$. 
Under these conditions, the expectation associated to the law of Steps (1) and (2) of Algorithm~\ref{alg:4-CCPF}
is denoted as $\check{\mathbb{E}}_{\theta}^{l-1,l}$.

\begin{lemma}\label{lem:cccpf1}
Under Assumptions~\ref{ass:D1} and \ref{ass:D2}, for any $(t,r,\theta,C')\in \{1,\dots,T\}\times[1,\infty)\times\Theta\times\mathbb{R}^+$, there exists a constant $C<\infty$ such that for any $(l,\beta,N)\in\mathbb{N}\times\mathbb{R}^+\times\{2,3,\dots\}$ and any $(x_{0:T}^l,x_{0:T}^{l-1}) \in \mathsf{B}^l_{\beta,C'}$, it holds that
$$
\check{\mathbb{E}}_{\theta}^{l-1,l}\bigg[\frac{1}{N}\sum_{i\in \mathsf{S}_{t-1}^l}\|X_{t}^{l,i}-X_{t}^{l-1,i}\|_2^r\bigg]^{1/r}\leq C\Delta_l^{\frac{1}{2}\wedge\beta}.
$$
If $(\bar{x}_{0:T}^l,\bar{x}_{0:T}^{l-1})\in \mathsf{B}^l_{\beta,C'}$ also holds, then 
$$
\check{\mathbb{E}}_{\theta}^{l-1,l}\bigg[\frac{1}{N}\sum_{i\in \bar{\mathsf{S}}_{t-1}^l}\|\bar{X}_{t}^{l,i}-\bar{X}_{t}^{l-1,i}\|_2^r\bigg]^{1/r}\leq C\Delta_l^{\frac{1}{2}\wedge\beta}.
$$
\end{lemma}

\begin{proof}
We only consider the first inequality as it is the same proof for the second inequality.
The proof is almost identical to \citet[Lemma D.3.]{jasra2017multilevel}. The only difference is if $A_p^{l,i}=N$ for some $(p,i)\in\{1,\dots,t\}
\times\{1,\dots,N\}$ with $i\in\mathsf{S}_p^l$, but in such a case, one can use that $(x_{0:T}^l,x_{0:T}^{l-1})\in\mathsf{B}^l_{\beta,C'}$.
\end{proof}

\begin{remark}\label{rem:res_cccpf}
Implicit in the proof of Lemma \ref{lem:cccpf1} is the following result. Under Assumptions~\ref{ass:D1} and \ref{ass:D2}, for any $(t,r,\theta,C')\in \{1,\dots,T\}\times[1,\infty)\times\Theta\times\mathbb{R}^+$, there exists a constant $C<\infty$ such that for any $(l,\beta,N)\in\mathbb{N}\times\mathbb{R}^+\times\{2,3,\dots\}$ and any $(x_{0:T}^l,x_{0:T}^{l-1}) \in \mathsf{B}^l_{\beta,C'}$, it holds that
$$
\check{\mathbb{E}}_{\theta}^{l-1,l}\bigg[\frac{1}{N}\sum_{i\in \mathsf{S}_{t-1}^l} \Big\|X_{t-1}^{l,A_{t-1}^{l,i}}-X_{t-1}^{l-1, A_{t-1}^{l-1,i}}\Big\|_2^r\bigg]^{1/r}\leq C\Delta_l^{\frac{1}{2}\wedge\beta}.
$$
If $(\bar{x}_{0:T}^l,\bar{x}_{0:T}^{l-1})\in \mathsf{B}^l_{\beta,C'}$ also holds, then 
$$
\check{\mathbb{E}}_{\theta}^{l-1,l}\bigg[\frac{1}{N}\sum_{i\in \bar{\mathsf{S}}_{t-1}^l} \Big\|\bar{X}_{t-1}^{l,\bar{A}_{t-1}^{l,i}}-\bar{X}_{t-1}^{l-1,\bar{A}_{t-1}^{l-1,i}} \Big\|_2^r\bigg]^{1/r}\leq C\Delta_l^{\frac{1}{2}\wedge\beta}.
$$
\end{remark}

\begin{lemma}\label{lem:cccpf3}
Under Assumptions~\ref{ass:D1} and \ref{ass:D2}, for any $(t,\theta,C')\in \{1,\dots,T\}\times\Theta\times\mathbb{R}^+$, there exists a constant $C<\infty$ such that for any $(l,\beta,N)\in\mathbb{N}\times \mathbb{R}^+\times\{2,3,\dots\}$ and any $(x_{0:T}^l,x_{0:T}^{l-1}) \in \mathsf{B}^l_{\beta,C'}$, it holds that
$$
1-\check{\mathbb{E}}_{\theta}^{l-1,l}\Bigg[\frac{\textrm{\emph{Card}}(\mathsf{S}_{t-1}^l)}{N}\Bigg]\leq C\Delta_l^{\frac{1}{2}\wedge\beta},
$$
where $\emph{Card}(\cdot)$ denotes the cardinality of a set. 
If $(\bar{x}_{0:T}^l,\bar{x}_{0:T}^{l-1})\in \mathsf{B}^l_{\beta,C'}$ also holds, then 
$$
1-\check{\mathbb{E}}_{\theta}^{l-1,l}\Bigg[\frac{\textrm{\emph{Card}}(\bar{\mathsf{S}}_{t-1}^l)}{N}\Bigg]\leq C\Delta_l^{\frac{1}{2}\wedge\beta}.
$$
\end{lemma}

\begin{proof}
We only consider the first inequality as it is the same proof for the second inequality.
The proof is almost identical to \citet[Lemma D.4.]{jasra2017multilevel}. The only difference is if $A_p^{l,i}=N$ for some $(p,i)\in\{1,\dots,t\}
\times\{1,\dots,N\}$ with $i\in\mathsf{S}_p^l$, but in such a case, one can use that $(x_{0:T}^l,x_{0:T}^{l-1})\in\mathsf{B}^l_{\beta,C'}$ 
and Lemma \ref{lem:cccpf1}.
\end{proof}

\begin{lemma}\label{lem:cccpf2}
Under Assumptions~\ref{ass:D1} and \ref{ass:D2}, for any $(t,r,\theta,C')\in \{1,\dots,T\}\times[1,\infty)\times\Theta\times\mathbb{R}^+$, there exists a constant $C<\infty$ such that for any $(l,\beta,N)\in\mathbb{N}\times\mathbb{R}^+\times\{2,3,\dots\}$ and any $(x_{0:T}^l,x_{0:T}^{l-1}) \in \mathsf{B}^l_{\beta,C'}\cap \mathsf{G}^l_{\beta,C'}$, it holds that
$$
\check{\mathbb{E}}_{\theta}^{l-1,l}\Big[\frac{1}{N}\sum_{i\in\mathsf{S}_{t-1}^l} \big\|
G_{\theta,t}^l(X_{0:t}^{l,i})
-G_{\theta,t}^{l-1}(X_{0:t}^{l-1,i})\big\|_2^r\Big]^{1/r}\leq C\Delta_l^{\frac{1}{2}\wedge\beta}.
$$
If $(\bar{x}_{0:T}^l,\bar{x}_{0:T}^{l-1})\in \mathsf{B}^l_{\beta,C'}\cap \mathsf{G}^l_{\beta,C'}$ also holds, then 
$$
\check{\mathbb{E}}_{\theta}^{l-1,l}\Big[\frac{1}{N}\sum_{i\in\bar{\mathsf{S}}_{t-1}^l}\big\|
G_{\theta,t}^l(\bar{X}_{0:t}^{l,i})
-G_{\theta,t}^{l-1}(\bar{X}_{0:t}^{l-1,i})\big\|_2^r\Big]^{1/r}\leq C\Delta_l^{\frac{1}{2}\wedge\beta}.
$$
\end{lemma}

\begin{proof}
We only consider the first inequality as it is the same proof for the second inequality.
The proof is by induction on $t$. The initialization holds for $i\in\{1,\dots,N-1\}$ by the result stated in Remark \ref{rem:euler_grad}. The case $i=N$ is trivial as $(x_{0:T}^l,x_{0:T}^{l-1})\in \mathsf{B}^l_{\beta,C'}\cap \mathsf{G}^l_{\beta,C'}$.

We now consider
\begin{equation}\label{eq:cccpft1}
\check{\mathbb{E}}_{\theta}^{l-1,l}\Big[\frac{1}{N}\sum_{i\in\mathsf{S}_{t-1}^l}\big\|
G_{\theta,t}^l(X_{0:t}^{l,i})
-G_{\theta,t}^{l-1}(X_{0:t}^{l-1,i})\big\|_2^r\Big]^{1/r}\leq C(T_1+T_2),
\end{equation}
where
\begin{align*}
T_1 & = \check{\mathbb{E}}_{\theta}^{l-1,l}\Big[\frac{1}{N}\sum_{i\in\mathsf{S}_{t-1}^l}\big\|
G_{\theta,t-1:t}^l(X_{t-1}^{l,i},X_{t-1+\Delta_l:t}^{l,i})
-G_{\theta,t-1:t}^{l-1}(X_{t-1}^{l-1,i},X_{t-1+\Delta_{l-1}:t}^{l-1,i})\big\|_2^r\Big]^{1/r},\\
T_2 & = \check{\mathbb{E}}_{\theta}^{l-1,l}\bigg[\frac{1}{N}\sum_{i\in\mathsf{S}_{t-1}^l}\Big\|
G_{\theta,t-1}^l(X_{0:t-1}^{l,A_{t-1}^{l,i}})
-G_{\theta,t-1}^{l-1}(X_{0:t-1}^{l-1,A_{t-1}^{l-1,i}})\Big\|_2^r\bigg]^{1/r}.
\end{align*}
So we consider bounds on $T_1$ and $T_2$. 

For $T_1$, by applying the result in Remark \ref{rem:euler_grad1}, we have 
$$
T_1 \leq C\bigg(\Delta_l^{1/2} + \check{\mathbb{E}}_{\theta}^{l-1,l}\Big[\frac{1}{N}\sum_{i\in\mathsf{S}_{t-1}^l}\Big\|
X_{t-1}^{l,A_{t-1}^{l,i}} - X_{t-1}^{l-1,A_{t-1}^{l-1,i}}
\Big\|_2^r\Big]^{1/r}\bigg).
$$
Then applying the result in Remark \ref{rem:res_cccpf} gives
\begin{equation}\label{eq:cccpft2}
T_1 \leq C\Delta_l^{\frac{1}{2}\wedge\beta}.
\end{equation}
For $T_2$, one can use the same approach as in the proof of \citet[Lemma D.3.]{jasra2017multilevel} from the bottom of page 3092, along with the induction hypothesis, to obtain
\begin{equation}\label{eq:cccpft3}
T_2 \leq C\Delta_l^{\frac{1}{2}\wedge\beta}.
\end{equation}
Combining \eqref{eq:cccpft1} with \eqref{eq:cccpft2} and \eqref{eq:cccpft3} concludes the proof.
\end{proof}

\begin{lemma}\label{lem:cccpf5}
Under Assumptions~\ref{ass:D1} and \ref{ass:D2}, for any $(t,r,\theta)\in \{1,\dots,T\}\times[1,\infty)\times\Theta$, there exists a constant $C<\infty$ such that for any $(l,s,N,i)\in\mathbb{N}\times\{l-1,1\}\times\{2,3,\dots\}\times\{1,\dots,N\}$ and any $(x_{0:T}^l,x_{0:T}^{l-1}) \in \mathsf{X}^l\times\mathsf{X}^{l-1}$, it holds that
$$
\check{\mathbb{E}}_{\theta}^{l-1,l}\Big[\big\|
G_{\theta,t}^s(X_{0:t}^{s,i})\big\|_2^{r}
\Big] \leq 
C\Big(1+
\sum_{p=1}^{t}\big\|G_{\theta,p-1:p}^{s}(x_{p-1}^{s,N},x_{p-1+\Delta_s:p}^{s,N})\big\|_2^r
\Big).
$$
Also for any $(\bar{x}_{0:T}^l,\bar{x}_{0:T}^{l-1})\in \mathsf{X}^l\times\mathsf{X}^{l-1}$ 
$$
\check{\mathbb{E}}_{\theta}^{l-1,l}\Big[\big\|
G_{\theta,t}^s(\bar{X}_{0:t}^{s,i})\big\|_2^{r}
\Big] \leq 
C\Big(1+
\sum_{p=1}^{t}\big\|G_{\theta,p-1:p}^{s}(\bar{x}_{p-1}^{s,N},\bar{x}_{p-1+\Delta_s:p}^{s,N})\big\|_2^r
\Big).
$$
\end{lemma}

\begin{proof}
We only consider the first inequality as it is the same proof for the second inequality. 
We consider a proof by induction. In the case of $t=1$, one can use the boundedness properties of the appropriate terms along with the martingale-remainder methods in the proof of Lemma \ref{lem:diff4} to deduce the given bound, except for the case $i=N$, which is exhibited in the bound.

Now consider the case of $t>1$, we have the upper-bound
\begin{equation}\label{eq:cccpft6}
\check{\mathbb{E}}_{\theta}^{l-1,l}\Big[\big\|
G_{\theta,t}^s(X_{0:t}^{s,i}) \big\|_2^{r} \Big] \leq C(T_1+T_2),
\end{equation}
where
\begin{align*}
T_1 & = \check{\mathbb{E}}_{\theta}^{l-1,l}\Big[\big\|
G_{\theta,t-1:t}^s(X_{t-1}^{s,i},X_{t-1+\Delta_s:t}^{s,i})\big\|_2^r\Big],\\
T_2 & = \check{\mathbb{E}}_{\theta}^{l-1,l}\Big[\big\|
G_{\theta,t-1}^s(X_{0:t-1}^{s,A_{t-1}^{s,i}})\big\|_2^r\Big].
\end{align*}
So we consider bounds on $T_1$ and $T_2$. 

For $T_1$, by the same argument as for the initialization
\begin{equation}\label{eq:cccpft7}
T_1 \leq C\Big(1+\big\|G_{\theta,t-1:t}^{s}(x_{t-1}^{s,N},x_{t-1+\Delta_s:t}^{s,N})\big\|_2^r\Big).
\end{equation}
For $T_2$, we first note that for the resampling probabilities associated to $A_{t-1}^{s,i}$, we can deduce the following upper-bound using Assumption~\ref{ass:D2}
\begin{equation}\label{eq:cccpft5}
\frac{g_{\theta}(y_{t-1}|x_{t-1}^i)}{\sum_{j=1}^N g_{\theta}(y_{t-1}|x_{t-1}^j)}\leq \frac{C}{N}.
\end{equation}
Averaging over the resampling indexes, using \eqref{eq:cccpft5} and the induction hypothesis one has
\begin{equation}\label{eq:cccpft8}
T_2 \leq C\Big(1+
\sum_{p=1}^{t-1}\big\|G_{\theta,p-1:p}^{s}(x_{p-1}^{s,N},x_{p-1+\Delta_s:p}^{s,N})\big\|_2^r\Big).
\end{equation}
Combining \eqref{eq:cccpft6} with \eqref{eq:cccpft7} and \eqref{eq:cccpft8} concludes the proof.
\end{proof}

\begin{corollary}\label{cor:cccpf5}
Under Assumptions~\ref{ass:D1} and \ref{ass:D2}, for any $(t,r,\theta)\in \{1,\dots,T\}\times[1,\infty)\times\Theta$, there exists a constant $C<\infty$ such that for any $(l,s,N,i)\in\mathbb{N}\times\{l-1,1\}\times\{2,3,\dots\}\times\{1,\dots,N\}$ and any $(x_{0:T}^l,x_{0:T}^{l-1}) \in \mathsf{X}^l\times\mathsf{X}^{l-1}$, it holds that
$$
\check{\mathbb{E}}_{\theta}^{l-1,l}\Big[\big\|X_{t}^{s,i}\big\|_2^{r}
\Big] \leq 
C\Big(1+
\sum_{p=1}^{t}\big\|x_p^{s,N}\big\|_2^r
\Big).
$$
Also for any $(\bar{x}_{0:T}^l,\bar{x}_{0:T}^{l-1})\in \mathsf{X}^l\times\mathsf{X}^{l-1}$ 
$$
\check{\mathbb{E}}_{\theta}^{l-1,l}\Big[\big\|
\bar{X}_{t}^{s,i}\big\|_2^{r}
\Big] \leq 
C\Big(1+
\sum_{p=1}^{t}\big\|\bar{x}_p^{s,N}\big\|_2^r
\Big).
$$
\end{corollary}

\begin{proof}
We only consider the first inequality as it is the same proof for the second inequality. 
We consider a proof by induction. In the case of $t=1$, one can use the boundedness properties of the appropriate terms along with the martingale-remainder methods in the proof of Lemma \ref{lem:diff4} to deduce the given bound, except for the case $i=N$, which is exhibited in the bound.

For $t>1$, when $i\in\{1,\dots,N-1\}$, repeating the argument of the initialization, one has
$$
\check{\mathbb{E}}_{\theta}^{l-1,l}\Big[\big\|X_{t}^{s,i}\big\|_2^{r}\Big] \leq C\Big(1+
\check{\mathbb{E}}_{\theta}^{l-1,l}\Big[\Big\|X_{t-1}^{s,A_{t-1}^{s,i}}\Big\|_2^{r}\Big]
\Big).
$$
Then one can repeat the argument that leads to \eqref{eq:cccpft8}. The case $i=N$ is trivially true.
\end{proof}

\begin{remark}\label{rem:extend_res}
The results in Lemmata \ref{lem:cccpf2} and \ref{lem:cccpf5} can be extended to the case where one considers
$$
\check{\mathbb{E}}_{\theta}^{l-1,l}\bigg[\frac{1}{N}\sum_{i\in\mathsf{S}_{t-1}^l}\Big\| G_{\theta,t}^l(X_{0:t}^{l,A_t^{l,i}}) -G_{\theta,t}^{l-1}(X_{0:t}^{l-1,A_t^{l-1,i}})\Big\|_2^r\bigg], \quad 
\check{\mathbb{E}}_{\theta}^{l-1,l}\Big[\frac{1}{N}\sum_{i\in\bar{\mathsf{S}}_{t-1}^l}\big\|
G_{\theta,t}^l(\bar{X}_{0:t}^{l,\bar{A}_t^{l,i}})
-G_{\theta,t}^{l-1}(\bar{X}_{0:t}^{l-1,\bar{A}_t^{l-1,i}})\big\|_2^r\Big]^{1/r},
$$
and
$$
\check{\mathbb{E}}_{\theta}^{l-1,l}\bigg[\Big\|
G_{\theta,t}^s(X_{0:t}^{s,A_t^{s,i}})\Big\|_2^{r}
\bigg],\quad 
\check{\mathbb{E}}_{\theta}^{l-1,l}\bigg[\Big\|
G_{\theta,t}^s(\bar{X}_{0:t}^{s,\bar{A}_t^{s,i}})\Big\|_2^{r}
\bigg],
$$
by using very similar arguments to the proof of those lemmata. 
\end{remark}

\begin{lemma}\label{lem:cccpf4}
Under Assumptions~\ref{ass:D1} and \ref{ass:D2}, for any $(t,r,\theta,C')\in \{1,\dots,T\}\times[1,\infty)\times\Theta\times\mathbb{R}^+$, there exists a constant $C<\infty$ such that for any $(l,\beta,N,i,\delta)\in\mathbb{N}\times\mathbb{R}^+\times\{2,3,\dots\}\times\{1,\dots,N\}\times\mathbb{R}^+$ and any $(x_{0:T}^l,x_{0:T}^{l-1}) \in \mathsf{B}^l_{\beta,C'}\cap \mathsf{G}^l_{\beta,C'}$, it holds that
$$
\check{\mathbb{E}}_{\theta}^{l-1,l}\Big[\big\|
G_{\theta,t}^l(X_{0:t}^{l,i})
-G_{\theta,t}^{l-1}(X_{0:t}^{l-1,i})\big\|_2^r\Big]^{1/r}\leq C(\Delta_l^{\frac{1}{2}\wedge\beta})^{\frac{1}{r(1+\delta)}}
\Big(1+\sum_{p=1}^t\sum_{s=l-1}^l \big\|G_{\theta,p-1:p}^s(x_{p-1}^{s,N},x_{p-1+\Delta_s:p}^{s,N})\big\|_2^r\Big).
$$
If $(\bar{x}_{0:T}^l,\bar{x}_{0:T}^{l-1})\in \mathsf{B}^l_{\beta,C'}\cap \mathsf{G}^l_{\beta,C'}$ also holds, then 
$$
\check{\mathbb{E}}_{\theta}^{l-1,l}\Big[\big\|
G_{\theta,t}^l(\bar{X}_{0:t}^{l,i})
-G_{\theta,t}^{l-1}(\bar{X}_{0:t}^{l-1,i})\big\|_2^r\Big]^{1/r}\leq 
C(\Delta_l^{\frac{1}{2}\wedge\beta})^{\frac{1}{r(1+\delta)}}
\Big(1+\sum_{p=1}^t\sum_{s=l-1}^l \big\|G_{\theta,p-1:p}^s(\bar{x}_{p-1}^{s,N},\bar{x}_{p-1+\Delta_s:p}^{s,N})\big\|_2^r\Big).
$$
\end{lemma}

\begin{proof}
We only consider the first inequality as it is the same proof for the second inequality. We have
\begin{equation}\label{eq:cccpft11}
\check{\mathbb{E}}_{\theta}^{l-1,l}\Big[\big\|
G_{\theta,t}^l(X_{0:t}^{l,i})
-G_{\theta,t}^{l-1}(X_{0:t}^{l-1,i})\big\|_2^r\Big]^{1/r} = T_1 + T_2,
\end{equation}
where
\begin{align*}
T_1 & = \check{\mathbb{E}}_{\theta}^{l-1,l}\Big[\big\|
G_{\theta,t}^l(X_{0:t}^{l,i})
-G_{\theta,t}^{l-1}(X_{0:t}^{l-1,i})\big\|_2^r\mathbb{I}_{\mathsf{S}_{t-1}^l}(i)\Big]^{1/r}, \\
T_2 & = \check{\mathbb{E}}_{\theta}^{l-1,l}\Big[\big\|
G_{\theta,t}^l(X_{0:t}^{l,i})
-G_{\theta,t}^{l-1}(X_{0:t}^{l-1,i})\big\|_2^r\mathbb{I}_{(\mathsf{S}_{t-1}^l)^c}(i)\Big]^{1/r}.
\end{align*}
So we consider bounds on $T_1$ and $T_2$. 

For $T_1$, we have the upper-bound
$$
T_1 \leq C\check{\mathbb{E}}_{\theta}^{l-1,l}\Big[\frac{1}{N}\sum_{i\in \mathsf{S}_{t-1}^l}\big\|
G_{\theta,t}^l(X_{0:t}^{l,i})
-G_{\theta,t}^{l-1}(X_{0:t}^{l-1,i})\big\|_2^r\Big]^{1/r}
$$
then applying Lemma \ref{lem:cccpf2} gives
\begin{equation}\label{eq:cccpft12}
T_1 \leq C\Delta_l^{\frac{1}{2}\wedge\beta}.
\end{equation}

For $T_2$, applying H\"older's inequality gives
\begin{equation}\label{eq:cccpft10}
T_2 \leq \check{\mathbb{E}}_{\theta}^{l-1,l}\Big[\big\|
G_{\theta,t}^l(X_{0:t}^{l,i})
-G_{\theta,t}^{l-1}(X_{0:t}^{l-1,i})\big\|_2^{\frac{r(1+\delta)}{\delta}}\Big]^{\frac{\delta}{(1+\delta)}}
\check{\mathbb{E}}_{\theta}^{l-1,l}\big[\mathbb{I}_{(\mathsf{S}_{t-1}^l)^c}(i)\big]^{\frac{1}{r(1+\delta)}}.
\end{equation}
Note that 
\begin{equation}\label{eq:cccpft9}
\check{\mathbb{E}}_{\theta}^{l-1,l}\big[\mathbb{I}_{(\mathsf{S}_{t-1}^l)^c}(i)\big] = 
1-\check{\mathbb{E}}_{\theta}^{l-1,l}\Bigg[\frac{\textrm{Card}(\mathsf{S}_{t-1}^l)}{N}\Bigg] \leq C\Delta_l^{\frac{1}{2}\wedge\beta},
\end{equation}
where we have used Lemma \ref{lem:cccpf3}, and 
\begin{align*}
&\check{\mathbb{E}}_{\theta}^{l-1,l}\Big[\big\|
G_{\theta,t}^l(X_{0:t}^{l,i})
-G_{\theta,t}^{l-1}(X_{0:t}^{l-1,i})\big\|_2^{\frac{r(1+\delta)}{\delta}}]^{\frac{\delta}{(1+\delta)}}\\
&\leq C\Big(\check{\mathbb{E}}_{\theta}^{l-1,l}\Big[\big\|
G_{\theta,t}^l(X_{0:t}^{l,i})\big\|_2^{\frac{r(1+\delta)}{\delta}}\Big]^{\frac{\delta}{(1+\delta)}} +
\check{\mathbb{E}}_{\theta}^{l-1,l}\Big[\big\|
G_{\theta,t}^{l-1}(X_{0:t}^{l-1,i})\big\|_2^{\frac{r(1+\delta)}{\delta}}\Big]^{\frac{\delta}{(1+\delta)}}\Big).
\end{align*}
Then applying Lemma \ref{lem:cccpf5} and combining with \eqref{eq:cccpft10} and \eqref{eq:cccpft9}, one can deduce that
\begin{equation} \label{eq:cccpft13}
T_2\leq C(\Delta_l^{\frac{1}{2}\wedge\beta})^{\frac{1}{r(1+\delta)}}
\Big(1+\sum_{p=1}^t\sum_{s=l-1}^l \big\|G_{\theta,p-1:p}^s(x_{p-1}^{s,N},x_{p-1+\Delta_s:p}^{s,N})\big\|_2^r\Big).
\end{equation}
Combining \eqref{eq:cccpft11}, \eqref{eq:cccpft12} and \eqref{eq:cccpft13} completes the proof. 
\end{proof}

\begin{corollary}\label{cor:cccpf4}
Under Assumptions~\ref{ass:D1} and \ref{ass:D2}, for any $(t,r,\theta,C')\in \{1,\dots,T\}\times[1,\infty)\times\Theta\times\mathbb{R}^+$, there exists a constant $C<\infty$ such that for any $(l,\beta,N,i,\delta)\in\mathbb{N}\times\mathbb{R}^+\times\{2,3,\dots\}\times\{1,\dots,N\}\times\mathbb{R}^+$ and any $(x_{0:T}^l,x_{0:T}^{l-1}) \in \mathsf{B}^l_{\beta,C'}\cap \mathsf{G}^l_{\beta,C'}$, it holds that
$$
\check{\mathbb{E}}_{\theta}^{l-1,l}\Big[\big\|X_{t}^{l,i}
-X_{t}^{l-1,i}\big\|_2^r\Big]^{1/r}\leq C(\Delta_l^{\frac{1}{2}\wedge\beta})^{\frac{1}{r(1+\delta)}}
\Big(1+\sum_{p=1}^t\sum_{s=l-1}^l \big\|x_p^{s,N}\big\|_2^r\Big).
$$
If $(\bar{x}_{0:T}^l,\bar{x}_{0:T}^{l-1})\in \mathsf{B}^l_{\beta,C'}\cap \mathsf{G}^l_{\beta,C'}$ also holds, then 
$$
\check{\mathbb{E}}_{\theta}^{l-1,l}\Big[\big\|\bar{X}_{t}^{l,i}
-\bar{X}_{t}^{l-1,i}\big\|_2^r\Big]^{1/r}\leq 
C(\Delta_l^{\frac{1}{2}\wedge\beta})^{\frac{1}{r(1+\delta)}}
\Big(1+\sum_{p=1}^t\sum_{s=l-1}^l \big\|\bar{x}_p^{s,N}\big\|_2^r\Big).
$$
\end{corollary}

\begin{remark}\label{rem:extend_res1}
Using Remark \ref{rem:extend_res}, one can extend Lemma \ref{lem:cccpf4} using a similar argument to its proof
to obtain the same bound on
$$
\check{\mathbb{E}}_{\theta}^{l-1,l}\bigg[\Big\|
G_{\theta,t}^l(X_{0:t}^{l,A_t^{l,i}})
-G_{\theta,t}^{l-1}(X_{0:t}^{l-1,A_t^{l-1,i}})\Big\|_2^r\bigg]^{1/r}.
$$
A similar statement applies to Corollary \ref{cor:cccpf4}.
\end{remark}

We introduce the following sets, which will be of use later on in our proofs.  For $(t,l)\in\{0,\dots,T-1\}\times\mathbb{N}$, we define
\begin{align}
\label{eq:set_level_coup1}
\check{\mathsf{S}}_t^l = \{i\in\{1,\ldots, N-1\}: ~&A_{t}^l(i)=\bar{A}_{t}^{l}(i)\neq N,A_{t-1}^l\circ A_{t}^l(i)=
\bar{A}_{t-1}^{l}\circ \bar{A}_{t}^{l}(i)\neq N,\dots,\notag\\
&A_{0}^l\circ\cdots\circ A_{t}^l(i)=\bar{A}_{0}^{l}\circ\cdots\circ \bar{A}_{t}^{l}(i)\neq N\}
\end{align}
and
\begin{align}
\label{eq:set_level_coup2}
\check{\mathsf{S}}_t^{l-1} = \{i\in\{1,\ldots, N-1\}: ~&A_{t}^{l-1}(i)=\bar{A}_{t}^{l-1}(i)\neq N,A_{t-1}^{l-1}\circ A_{t}^{l-1}(i)=
\bar{A}_{t-1}^{l}\circ \bar{A}_{t}^{l}(i)\neq N,\dots,\notag\\
&A_{0}^{l-1}\circ\cdots\circ A_{t}^{l-1}(i)=\bar{A}_{0}^{l-1}\circ\cdots\circ \bar{A}_{t}^{l-1}(i)\neq N\}.
\end{align}
\begin{lemma}\label{lem:prob_lower}
Under Assumptions~\ref{ass:D1} and \ref{ass:D2}, for any $(t,\theta,N)\in \{0,\dots,T-1\}\times\Theta\times\{2,3,\dots\}$, there exists a constant $\varepsilon\in(0,1)$ such that for any $(l,i)\in\mathbb{N}\times\{1,\dots,N-1\}$ and any $((x_{0:T}^{l-1},\bar{x}_{0:T}^{l-1}),(x_{0:T}^l,\bar{x}_{0:T}^{l})) \in \mathsf{Z}^{l-1}\times \mathsf{Z}^{l}$, it holds that
$$
\check{\mathbb{E}}_{\theta}^{l-1,l}\big[\mathbb{I}_{\check{\mathsf{S}}_t^l}(i)\big]\wedge
\check{\mathbb{E}}_{\theta}^{l-1,l}\big[\mathbb{I}_{\check{\mathsf{S}}_t^{l-1}}(i)\big]
\geq \varepsilon.
$$
\end{lemma}

\begin{proof}
To proceed we first introduce some notation. For $(t,l)\in\{1,\dots,T\}\times\mathbb{N}$, we define
\begin{eqnarray*}
\check{\vartheta}_t^l(i,j) & = &  \Bigg\{\frac{g_{\theta}(y_t|x_t^{l,i})}{\sum_{j_1=1}^N g_{\theta}(y_t|x_t^{l,j_1})}\wedge
\frac{g_{\theta}(y_t|x_t^{l-1,j})}{\sum_{j_1=1}^N g_{\theta}(y_t|x_t^{l-1,j_1})}\Bigg\}\mathbb{I}_{\{i\}}(j) + \\ & &
\Bigg(\frac{g_{\theta}(y_t|x_t^{l,i})}{\sum_{j_1=1}^N g_{\theta}(y_t|x_t^{l,j_1})} - 
\Bigg\{\frac{g_{\theta}(y_t|x_t^{l,i})}{\sum_{j_1=1}^N g_{\theta}(y_t|x_t^{l,j_1})}\wedge
\frac{g_{\theta}(y_t|x_t^{l-1,i})}{\sum_{j_1=1}^N g_{\theta}(y_t|x_t^{l-1,j_1})}\Bigg\}\Bigg)\times\\ & &
\Bigg(\frac{g_{\theta}(y_t|x_t^{l-1,j})}{\sum_{j_1=1}^N g_{\theta}(y_t|x_t^{l-1,j_1})} - 
\Bigg\{\frac{g_{\theta}(y_t|x_t^{l,j})}{\sum_{j_1=1}^N g_{\theta}(y_t|x_t^{l,j_1})}\wedge
\frac{g_{\theta}(y_t|x_t^{l-1,j})}{\sum_{j_1=1}^N g_{\theta}(y_t|x_t^{l-1,j_1})}\Bigg\}\Bigg)\times\\ & &
\Bigg(1 - 
\sum_{j_2=1}^N\Bigg\{\frac{g_{\theta}(y_t|x_t^{l,j_2})}{\sum_{j_1=1}^N g_{\theta}(y_t|x_t^{l,j_1})}\wedge
\frac{g_{\theta}(y_t|x_t^{l-1,j_2})}{\sum_{j_1=1}^N g_{\theta}(y_t|x_t^{l-1,j_1})}\Bigg\}\Bigg)^{-1}
\end{eqnarray*}
which is the maximal coupling of the resampling distributions across levels. 
We write $\check{\bar{\vartheta}}_t^l(i,j)$ when
one replaces $(x_{t}^{l,1:N},x_{t}^{l-1,1:N})$ with $(\bar{x}_{t}^{l,1:N},\bar{x}_{t}^{l-1,1:N})$. We will write the maximal coupling (in the above sense with independent residuals) of $\check{\vartheta}_t^l(j_1,j_2)$ and $\check{\bar{\vartheta}}_t^l(j_3,j_4)$ for $(j_1,\dots,j_4)\in\{1,\dots,N\}^4$, 
as $\overline{\vartheta}_t^l(j_1,\dots,j_4)$. We also define 
$\mathsf{D}^l=\{(x_{0:T},\bar{x}_{0:T})\in\mathsf{Z}^l:x_{0:T}= \bar{x}_{0:T}\}$ and
\begin{equation}\label{eq:res_l_marginal}
\overline{\vartheta}_t^{(l)}(j_1,j_3)  =  \mathbb{I}_{(\mathsf{D}^l)^c}(x_{0:T}^l,\bar{x}_{0:T}^l)\sum_{(j_2,j_4)\in\{1,\dots,N\}^2} \overline{\vartheta}_t^{(l)}(j_1,\dots,j_4) + %\\ & &
\mathbb{I}_{\mathsf{D}^l}(x_{0:T}^l,\bar{x}_{0:T}^l)
\mathbb{I}_{\{j_1\}}(j_3)\frac{g_{\theta}(y_t|x_t^{l,j_1})}{\sum_{j=1}^N g_{\theta}(y_t|x_t^{l,j})}.
\end{equation}
$\overline{\vartheta}_t^{(l)}(j_1,j_3)$ is the distribution of the resampled indexes within a level under Algorithm~\ref{alg:maximal-maximal}. % \ref{alg:max_coup_4}.
One can make a similar definition for $\overline{\vartheta}_t^{(l-1)}(j_2,j_4)$.

We give the proof in the case of $l$ only as the proof for $l-1$ is similar.
The proof is by induction on $t$ and the initial case $t=0$ is trivial by definition. 
For $t\geq 1$, we have
\begin{align*}
\check{\mathbb{E}}_{\theta}^{l-1,l}\big[\mathbb{I}_{\check{\mathsf{S}}_t^l}(i)\big] & =
\check{\mathbb{E}}_{\theta}^{l-1,l}\Big[\sum_{j=1}^{N-1} \mathbb{I}_{\check{\mathsf{S}}_{t-1}^l}(i) 
\overline{\vartheta}_t^{(l)}(j,j)
\Big] \\
& \geq \check{\mathbb{E}}_{\theta}^{l-1,l}\Big[\mathbb{I}_{\check{\mathsf{S}}_{t-1}^l}(1) 
\overline{\vartheta}_t^{(l)}(1,1)
\Big]\\
& \geq \check{\mathbb{E}}_{\theta}^{l-1,l}\Big[\mathbb{I}_{\check{\mathsf{S}}_{t-1}^l}(1)\Big]\frac{C}{N} \\
& \geq \varepsilon
\end{align*}
where we have used Assumption~\ref{ass:D2} to establish that 
\begin{equation}\label{eq:prob_lower1}
\overline{\vartheta}_t^{(l)}(1,1)\geq \frac{C}{N} 
\end{equation}
on the third line, and the induction hypothesis in the final line. This completes our proof.
\end{proof}

\subsubsection{Results associated to the entirety of Algorithm~\ref{alg:4-CCPF}} %\ref{alg:cccpf}}
\label{sec:analysis_second}
We now consider Algorithm~\ref{alg:4-CCPF} in its entirety. 
We will denote expectation and probability w.r.t.\ a single step of the corresponding 4-CCPF kernel $\bar{M}_{\theta}^{l-1,l}$ 
by $\bar{\mathbb{P}}_{\theta}^{l-1,l}$ and $\bar{\mathbb{E}}_{\theta}^{l-1,l}$, respectively. 
%write expectations w.r.t.~a single step of $\check{K}_{\theta}^l$, which is described in Algorithm \ref{alg:cccpf}, as $\bar{\mathbb{E}}_{\theta}^l$, with associated probability $\bar{\mathbb{P}}_{\theta}^l$.

\begin{corollary}\label{cor:cccpf_first_res}
Under Assumptions~\ref{ass:D1} and \ref{ass:D2}, for any $(T,r,\theta,C')\in \mathbb{N}\times[1,\infty)\times\Theta\times\mathbb{R}^+$, there exists a constant $C < \infty$ such that for any $(l,\beta,N,\delta)\in\mathbb{N}\times\mathbb{R}^+\times\{2,3,\dots\}\times\mathbb{R}^+$ and any $(x_{0:T}^l,x_{0:T}^{l-1}) \in \mathsf{B}^l_{\beta,C'}\cap \mathsf{G}^l_{\beta,C'}$, it holds that
$$
\bar{\mathbb{E}}_{\theta}^{l-1,1}\Big[\big\|G_{\theta}^l\big(X_{0:T}^{l,B_T^l}\big) - G_{\theta}^{l-1}\big(X_{0:T}^{l-1,B_T^{l-1}}\big)\big\|_2^r\Big]^{1/r}
\leq C(\Delta_l^{\frac{1}{2}\wedge\beta})^{\frac{1}{r(1+\delta)}}
\Big(1+\sum_{p=1}^T\sum_{s=l-1}^l \big\|G_{\theta,p-1:p}^s(x_{p-1}^{s},x_{p-1+\Delta_s:p}^{s})\big\|_2^r\Big).
$$
If $(\bar{x}_{0:T}^l,\bar{x}_{0:T}^{l-1})\in \mathsf{B}^l_{\beta,C'}\cap \mathsf{G}^l_{\beta,C'}$ also holds, then 
$$
\bar{\mathbb{E}}_{\theta}^{l-1,l}\Big[\big\|G_{\theta}^l\big(\bar{X}_{0:T}^{l,\bar{B}_T^l}\big) - G_{\theta}^{l-1}\big(\bar{X}_{0:T}^{l-1,\bar{B}_T^{l-1}}\big)\big\|_2^r\Big]^{1/r}
\leq C(\Delta_l^{\frac{1}{2}\wedge\beta})^{\frac{1}{r(1+\delta)}}
\Big(1+\sum_{p=1}^T\sum_{s=l-1}^l \big\|G_{\theta,p-1:p}^s(\bar{x}_{p-1}^{s},\bar{x}^{s}_{p-1+\Delta_s:p})\big\|_2^r\Big).
$$
\end{corollary}

\begin{proof}
This follows from the discussion in Remark \ref{rem:extend_res1}.
\end{proof}

\begin{remark}\label{rem:extend_cccpf}
By following the discussion in Remark \ref{rem:extend_res1}, 
one can also extend Corollary \ref{cor:cccpf_first_res} to a bound of the type
$$
\bar{\mathbb{E}}_{\theta}^{l-1,l}\Big[\big\|G_{\theta,t}^l\big(X_{0:t}^{l,B_T^l}\big) - G_{\theta,t}^{l-1}\big(X_{0:t}^{l-1,B_T^{l-1}}\big)\big\|_2^r\Big]^{1/r}
\leq C(\Delta_l^{\frac{1}{2}\wedge\beta})^{\frac{1}{r(1+\delta)}}
\Big(1+\sum_{p=1}^t\sum_{s=l-1}^l \big\|G_{\theta,p-1:p}^s(x_{p-1}^{s},x_{p-1+\Delta_s:p}^{s})\big\|_2^r\Big),
$$
for $t\in\{1,\dots,T\}$.
\end{remark}

\begin{remark}\label{rem:g_bound_cccpf}
One also use the discussion of Remark \ref{rem:extend_res1} to extend Lemma \ref{lem:cccpf5} to
$$
\bar{\mathbb{E}}_{\theta}^{l-1,1}\Big[\big\|
G_{\theta,t}^l\big(X_{0:t}^{l,B_T^l}\big)\big\|_2^{r}
\Big] \leq 
C\Big(1+
\sum_{p=1}^{t}\big\|G_{\theta,p-1:p}^{l}(x_{p-1}^{l},x_{p-1+\Delta_{l}:p}^{l})\big\|_2^r
\Big),
$$
and similarly for $\bar{\mathbb{E}}_{\theta}^{l-1,l}\Big[\big\|
G_{\theta,t}^l\big(\bar{X}_{0:t}^{l,\bar{B}_T^l}\big)\big\|_2^{r}
\Big]$.
\end{remark}

\begin{corollary}\label{cor:cccpf_add}
Under Assumptions~\ref{ass:D1} and \ref{ass:D2}, for any $(t,r,\theta,C')\in \{1,\dots,T\}\times[1,\infty)\times\Theta\times\mathbb{R}^+$, there exists a constant $C<\infty$ such that for any $(l,\beta,N,i,\delta)\in\mathbb{N}\times\mathbb{R}^+\times\{2,3,\dots\}\times\{1,\dots,N\}\times\mathbb{R}^+$ and any $(x_{0:T}^l,x_{0:T}^{l-1}) \in \mathsf{B}^l_{\beta,C'}\cap \mathsf{G}^l_{\beta,C'}$, it holds that
$$
\bar{\mathbb{E}}_{\theta}^{l-1,1}\Big[\big\|X_{t}^{l,B_T^l}
-X_{t}^{l-1,B_T^{l-1}}\big\|_2^r\Big]^{1/r}\leq C(\Delta_l^{\frac{1}{2}\wedge\beta})^{\frac{1}{r(1+\delta)}}
\Big(1+\sum_{p=1}^t\sum_{s=l-1}^l \|x_p^{s}\|_2^r\Big).
$$
If $(\bar{x}_{0:T}^l,\bar{x}_{0:T}^{l-1})\in \mathsf{B}^l_{\beta,C'}\cap \mathsf{G}^l_{\beta,C'}$ also holds, then 
$$
\bar{\mathbb{E}}_{\theta}^{l-1,l}\Big[\big\|\bar{X}_{t}^{l,\bar{B}_T^l}
-\bar{X}_{t}^{l-1,\bar{B}_T^{l-1}}\big\|_2^r\Big]^{1/r}\leq 
C(\Delta_l^{\frac{1}{2}\wedge\beta})^{\frac{1}{r(1+\delta)}}
\Big(1+\sum_{p=1}^t\sum_{s=l-1}^l \|\bar{x}_p^{s}\|_2^r\Big).
$$
\end{corollary}

\begin{proof}
This follows from Corollary \ref{cor:cccpf4} and the discussion in Remark \ref{rem:g_bound_cccpf}.
\end{proof}

\begin{lemma}\label{lem:cccpf_prob_trans}
Under Assumptions~\ref{ass:D1} and \ref{ass:D2}, for any $(T,\theta,C')\in \mathbb{N}\times\Theta\times\mathbb{R}^+$, there exists a constant $C<\infty$ such that for any $(l,\beta,N,\delta,\gamma)\in\mathbb{N}\times\mathbb{R}^+\times\{2,3,\dots\}\times\mathbb{R}^+\times(0,\frac{\frac{1}{2}\wedge\beta}{\beta(1+\delta)})$ and any $(x_{0:T}^l,x_{0:T}^{l-1}) \in \mathsf{B}^l_{\beta,C'}\cap \mathsf{G}^l_{\beta,C'}$, it holds that
$$
\bar{\mathbb{E}}_{\theta}^{l-1,1}\Big[\mathbb{I}_{(\mathsf{B}^l_{\beta,C'}\cap \mathsf{G}^l_{\beta,C'})^c}\big(X_{0:T}^{l,B_T^l},X_{0:T}^{l-1,B_T^{l-1}}\big)\Big]\leq 
C(\Delta_l)^{\frac{{\frac{1}{2}\wedge\beta}}{(1+\delta)}-\gamma\beta}
\Big(1+\sum_{p=1}^T\sum_{s=l-1}^l \big\{\|x_p^{s}\|_2^{\gamma}+\|G_{\theta,p-1:p}^{s}(x_{p-1}^{s},x_{p-1+\Delta_s:p}^{s})\|_2^{\gamma} \big\}\Big).
$$
If $(\bar{x}_{0:T}^l,\bar{x}_{0:T}^{l-1})\in \mathsf{B}^l_{\beta,C'}\cap \mathsf{G}^l_{\beta,C'}$ also holds, then 
$$
\bar{\mathbb{E}}_{\theta}^{l-1,l}\Big[\mathbb{I}_{(\mathsf{B}^l_{\beta,C'}\cap \mathsf{G}^l_{\beta,C'})^c}\big(\bar{X}_{0:T}^{l,\bar{B}_T^l},\bar{X}_{0:T}^{l-1,\bar{B}_T^{l-1}}\big)\Big]\leq 
C(\Delta_l)^{\frac{{\frac{1}{2}\wedge\beta}}{(1+\delta)}-\gamma\beta}
\Big(1+\sum_{p=1}^T\sum_{s=l-1}^l \big\{ \|\bar{x}_p^{s}\|_2^{\gamma}+\|G_{\theta,p-1:p}^{s}(\bar{x}_{p-1}^{s},\bar{x}_{p-1+\Delta_s:p}^{s})\|_2^{\gamma} \big\}\Big).
$$
\end{lemma}

\begin{proof}
We only consider the first inequality as it is the same proof for the second inequality.
For any $t\in\{1,\dots,T\}$, by Markov's inequality and Corollary \ref{cor:cccpf_add}, we have 
$$
\bar{\mathbb{P}}_{\theta}^{l-1,l}\Big( \big\| X_t^{l,B_T^l} - X_t^{l-1,B_T^{l-1}} \big\|_2 
%(X_t^{l,B_T^l},X_t^{l-1,B_T^{l-1}}) \in\Big\{(x,x')\in\mathbb{R}^d\times\mathbb{R}^d:\|x-x'\|_2
>C'\Delta_l^{\beta} \Big) \leq 
C(\Delta_l)^{\frac{{\frac{1}{2}\wedge\beta}}{(1+\delta)}-\gamma\beta}
\Big(1+\sum_{p=1}^t\sum_{s=l-1}^l \|x_p^{s}\|_2^{\gamma}\Big).
$$
Similarly, for any $t\in\{1,\dots,T\}$, by Markov's inequality and the results discussed in Remark \ref{rem:extend_cccpf}
\begin{align*}
&\bar{\mathbb{P}}_{\theta}^{l-1,l}\Big( \big\| G_{\theta,t}^l\big(X_{0:t}^{l,B_T^l}\big) - G_{\theta,t}^{l-1}\big(X_{0:t}^{l-1,B_T^{l-1}}\big) \big\|_2
%(X_{0:T}^{l,B_T^l},X_{0:T}^{l-1,B_T^{l-1}})\in\Big\{(x_{0:T},x_{0:T}')\in\mathsf{X}^l\times\mathsf{X}^{l-1}:\| G_{\theta,t}^l(x_{0:t}) - G_{\theta}^{l-1}(x_{0:t}') \|_2
> C'\Delta_l^{\beta}\Big) \\ 
&\leq C(\Delta_l)^{\frac{{\frac{1}{2}\wedge\beta}}{(1+\delta)}-\gamma\beta}
\Big(1+\sum_{p=1}^t\sum_{s=l-1}^l \|G_{\theta,p-1:p}^{s}(x_{p-1}^{s},x_{p-1+\Delta_s:p}^{s})\|_2^{\gamma}\Big).
\end{align*}
Hence there exists a constant $C<\infty$ which depends on $T$ but not $l$ such that the result holds.
\end{proof}

We recall the definition of $\mathsf{D}^l=\{(x_{0:T},\bar{x}_{0:T})\in\mathsf{Z}^l:x_{0:T}= \bar{x}_{0:T}\}$.

\begin{lemma}\label{lem:diag_prob_kernel}
Under Assumptions~\ref{ass:D1} and \ref{ass:D2}, for any $(T,\theta,N)\in \mathbb{N}\times\Theta\times\{2,3,\dots\}$, there exists a constant $\varepsilon\in(0,1)$ such that for any $l \in \mathbb{N}$ and any $((x_{0:T}^l,\bar{x}_{0:T}^{l}),(x_{0:T}^{l-1},\bar{x}_{0:T}^{l-1}))\in \mathsf{Z}^l\times \mathsf{Z}^{l-1}$, it holds that
$$
\bar{\mathbb{E}}_{\theta}^{l-1,l}\big[\mathbb{I}_{\mathsf{D}^l}\big(X_{0:T}^{l,B_T^l},\bar{X}_{0:T}^{l,\bar{B}_T^l}\big)
\big]\wedge \bar{\mathbb{E}}_{\theta}^{l-1,l}\big[\mathbb{I}_{\mathsf{D}^{l-1}}
\big(X_{0:T}^{l-1,B_T^{l-1}},\bar{X}_{0:T}^{l-1,\bar{B}_T^{l-1}}\big)
\big] \geq \varepsilon.
$$
\end{lemma}

\begin{proof}
Recall the definition \eqref{eq:res_l_marginal} of $\overline{\vartheta}_t^{(l)}(j_1,j_3)$ in the proof of Lemma \ref{lem:prob_lower}. This can be extended to time $T$ using the same construction for both level $l$ and $l-1$
and will correspond to the marginal distributions of $(B_T^l,\bar{B}_T^l)$ and $(B_T^{l-1},\bar{B}_T^{l-1})$. We denote these two probability distributions
as $\overline{\vartheta}_t^{(l)}(B_T^l,\bar{B}_T^l)$ and $\overline{\vartheta}_t^{(l-1)}(B_T^{l-1},\bar{B}_T^{l-1})$. 
Also recall the definitions of $\check{\mathsf{S}}_t^l$ and $\check{\mathsf{S}}_t^{l-1}$ in \eqref{eq:set_level_coup1}-\eqref{eq:set_level_coup2}.

We give the proof for level $l$ only as the case of level $l-1$ is almost identical. We have the following inequalities
\begin{align*}
\bar{\mathbb{E}}_{\theta}^{l-1,l}\big[\mathbb{I}_{\mathsf{D}^l}\big(X_{0:T}^{l,B_T^l},\bar{X}_{0:T}^{l,\bar{B}_T^l}\big)\big] & \geq
\check{\mathbb{E}}_{\theta}^{l-1,l}\Big[\sum_{j=1}^{N-1}\mathbb{I}_{\check{\mathsf{S}}_{T-1}}(j)\overline{\vartheta}_T^{(l)}(j,j)\Big]\\
& \geq \check{\mathbb{E}}_{\theta}^{l-1,l}\Big[\mathbb{I}_{\check{\mathsf{S}}_{T-1}}(1)\overline{\vartheta}_T^{(l)}(1,1)\Big]\\
& \geq \varepsilon.
\end{align*}
In the first line, we have noted that for $(x_{0:T}^{l,B_T^l},\bar{x}_{0:T}^{l,\bar{B}_T^l})\in\mathsf{D}^l$ to occur, one must at least pick two equal indexes of pairs of particles at level $l$ which were equal at time step $T-1$ of Algorithm~\ref{alg:4-CCPF}. 
In the final line, we have used \eqref{eq:prob_lower1} and Lemma \ref{lem:prob_lower}. This concludes the proof. 
\end{proof}

\subsubsection{Results associated to the initialization}% \eqref{eq:cccpf_ini}}
\label{sec:analysis_third}
Recall from Section \ref{sec:unbiased_discretized_increment} that the two pairs of CPF chains on $\mathsf{Z}^{l-1}\times\mathsf{Z}^{l}$ 
are initialized by sampling pairs $(X_{0:T}^{l-1,\star},X_{0:T}^{l,\star})$ and $(\bar{X}_{0:T}^{l-1},\bar{X}_{0:T}^{l})$ 
independently from $\nu_{\theta}^{l-1,l}$, and sampling $(X_{0:T}^{l-1},X_{0:T}^l)\sim M_{\theta}^{l-1,l}(\cdot | X_{0:T}^{l-1,\star},X_{0:T}^{l,\star})$ 
using the ML-CPF in Algorithm \ref{alg:ML-CPF}. 
We will denote the law of the tuple $(X_{0:T}^{l-1}, \bar{X}_{0:T}^{l-1}, X_{0:T}^l, \bar{X}_{0:T}^l)$ under 
this initialization by $\check{\nu}_{\theta}^{l-1,l}$. 
Expectations w.r.t.\ $\nu_{\theta}^{l-1,l}$, $\check{\nu}_{\theta}^{l-1,l}$ and the ML-CPF kernel $M_{\theta}^{l-1,l}$ 
will be written as $\mathbb{E}_{\theta,\nu}^{l-1,l}$, $\check{\mathbb{E}}_{\theta,\nu}^{l-1,l}$ and $\mathbb{E}_{\theta}^{l-1,l}$, respectively.

%To facilitate our analysis, we define $\nu_{\theta}^{l-1,l}$
%we refer to the distribution resulting from the first two steps of the 4-CCPF algorithm as the initial distribution, that is, the joint distribution of the tuple $(X_{0:T}^l(1), \bar{X}_{0:T}^l(0), X_{0:T}^{l-1}(1), \bar{X}_{0:T}^{l-1}(0))$, where $(X_{0:T}^{l-1}(0), X_{0:T}^l(0)) \sim \nu_{\theta}^{l-1,l}$, $(\bar{X}_{0:T}^{l-1}(0), \bar{X}_{0:T}^l(0)) \sim \nu_{\theta}^{l-1,l}$ and where the ML-CPF is used once on the first pair of trajectories as $(X_{0:T}^{l-1}(1), X_{0:T}^l(1)) \sim M_{\theta}^{l-1,l}(\cdot \,|\, X_{0:T}^{l-1}(0), X_{0:T}^l(0))$. This initial distribution is denoted by $\check{\nu}_{\theta}^l$ and the corresponding expectation will be expressed as $\check{\mathbb{E}}_{\nu_{\theta}}^l$. In addition, the expectation w.r.t.\ the ML-CPF kernel $M_{\theta}^{l-1,l}$ will be denoted by $\mathbb{E}^l_{\theta}$. Finally, the expectation and probability w.r.t.\ the distribution $\nu_{\theta}^{l-1,l}$ of the coupled Euler simulations will be denoted by $\mathbb{E}_{\theta}$ and $\mathbb{P}_{\theta}$ respectively.
%
%\hfill \textcolor{Red}{Every $G^{l-1}_{\theta,t}$ in this section was missing the $t$ subscript.}

\begin{lemma}\label{lem:ini_g_l}
Under Assumptions~\ref{ass:D1} and \ref{ass:D2}, for any $(t,r,\theta)\in \{1,\dots,T\}\times[1,\infty)\times\Theta$, there exists a constant $C<\infty$ such that for any $(l,\beta,N,\delta)\in\mathbb{N}\times(0,\frac{1}{2})\times\{2,3,\dots\}\times\mathbb{R}^+$, we have 
$$
\check{\mathbb{E}}_{\theta,\nu}^{l-1,l} \Big[\big\|G_{\theta,t}^l(X_{0:t}^{l}) - G_{\theta,t}^{l-1}(X_{0:t}^{l-1})\big\|_2^r\Big]^{1/r}
\leq C\Delta_l^{\frac{\beta}{r(1+\delta)}},
$$
and
$$
\check{\mathbb{E}}_{\theta,\nu}^{l-1,l} \Big[\big\|G_{\theta,t}^l(\bar{X}_{0:t}^{l}) - G_{\theta,t}^{l-1}(\bar{X}_{0:t}^{l-1})\big\|_2^r\Big]^{1/r}
\leq C\Delta_l^{\frac{1}{2}}.
$$
\end{lemma}

\begin{proof}
The second inequality simply follows from Remark \ref{rem:euler_grad}, so we only consider the first.
%We will write expectations w.r.t.~$\check{C}_{\theta}^l$ (as described in Algorithm~\ref{alg:ML-CPF} %\ref{alg:coup_cond_pf_l_l-1}) as $\check{\mathbb{E}}_{C_\theta}^l$ and expectations w.r.t.~the law of the coupled Euler simulations as $\mathbb{E}_{\theta}$ (with associated probability $\mathbb{P}_{\theta}$).
We have
$$
\check{\mathbb{E}}_{\theta,\nu}^{l-1,l}\Big[\big\|G_{\theta,t}^l(X_{0:t}^{l}) - G_{\theta,t}^{l-1}(X_{0:t}^{l-1})\big\|_2^r\Big]^{1/r} = 
\mathbb{E}_{\theta,\nu}^{l-1,l}\Big[ \mathbb{E}_{\theta}^{l-1,l} \Big[\big\|G_{\theta,t}^l(X_{0:t}^{l}) - G_{\theta,t}^{l-1}(X_{0:t}^{l-1})\big\|_2^r\Big]
\Big]^{1/r}.
$$
We note that by construction
$$
\mathbb{E}_{\theta}^{l-1,l}\Big[\big\|G_{\theta,t}^l(X_{0:t}^{l}) - G_{\theta,t}^{l-1}(X_{0:t}^{l-1})\big\|_2^r\Big] = 
\bar{\mathbb{E}}_{\theta}^{l-1,l}\Big[\big\|G_{\theta,t}^l(X_{0:t}^{l}) - G_{\theta,t}^{l-1}(X_{0:t}^{l-1})\big\|_2^r\Big].
$$
We consider the decomposition
\begin{equation}\label{eq:cccpft14}
\check{\mathbb{E}}_{\theta,\nu}^{l-1,l}\Big[\big\|G_{\theta,t}^l(X_{0:t}^{l}) - G_{\theta,t}^{l-1}(X_{0:t}^{l-1})\big\|_2^r\Big] 
= T_1 + T_2,
\end{equation}
where
\begin{align*}
T_1 & = \mathbb{E}_{\theta,\nu}^{l-1,l}\Big[
\mathbb{I}_{\mathsf{B}^l_{\beta,C'}\cap \mathsf{G}^l_{\beta,C'}}(X_{0:T}^{l,\star},X_{0:T}^{l-1,\star})
\bar{\mathbb{E}}_{\theta}^{l-1,l}\Big[\big\|G_{\theta,t}^l(X_{0:t}^{l}) - G_{\theta,t}^{l-1}(X_{0:t}^{l-1})\big\|_2^r\Big]\Big],\\
T_2 & = \mathbb{E}_{\theta,\nu}^{l-1,l}\Big[
\mathbb{I}_{(\mathsf{B}^l_{\beta,C'}\cap \mathsf{G}^l_{\beta,C'})^c}(X_{0:T}^{l,\star},X_{0:T}^{l-1,\star})
\bar{\mathbb{E}}_{\theta}^{l-1,l}\Big[\big\|G_{\theta,t}^l(X_{0:t}^{l}) - G_{\theta,t}^{l-1}(X_{0:t}^{l-1})\big\|_2^r\Big]\Big],
\end{align*}
for any $0<C'<\infty$. We will deal with both terms separately.

For $T_1$, applying the result in Remark \ref{rem:extend_cccpf} gives
\begin{equation}\label{eq:cccpft18}
T_1 \leq C\Delta_l^{\frac{\beta}{(1+\delta)}}\Big(1 + 
\mathbb{E}_{\theta,\nu}^{l-1,l}\Big[
\sum_{p=1}^t\sum_{s=l-1}^l \big\|G_{\theta,p-1:p}^s(X_{p-1}^{s,\star},X_{p-1+\Delta_s:p}^{s,\star})\big\|_2^r
\Big]\Big).
\end{equation}
We can use boundedness properties of the appropriate terms along with the martingale-remainder methods in the proof of Lemma \ref{lem:diff4} to deduce that
\begin{equation}\label{eq:cccpft15}
T_1 \leq C\Delta_l^{\frac{\beta}{(1+\delta)}}.
\end{equation}

For $T_2$, applying H\"older's inequality for any $\varrho>0$ gives 
$$
T_2 \leq T_3 T_4,
$$
where
\begin{align*}
T_3 & = \mathbb{E}_{\theta,\nu}^{l-1,l}\Big[
\mathbb{I}_{(\mathsf{B}^l_{\beta,C'}\cap \mathsf{G}^l_{\beta,C'})^c}(X_{0:T}^{l,\star},X_{0:T}^{l-1,\star})\Big]^{\frac{1}{(1+\varrho)}},\\
T_4 & = \mathbb{E}_{\theta,\nu}^{l-1,l}\Big[
\bar{\mathbb{E}}_{\theta}^{l-1,l}\Big[\big\|G_{\theta,t}^l(X_{0:t}^{l}) - G_{\theta,t}^{l-1}(X_{0:t}^{l-1})\big\|_2^r\Big]^{\frac{(1+\varrho)}{\varrho}}\Big]^{\frac{\varrho}{(1+\varrho)}}.
\end{align*}
We now bound $T_3$ and $T_4$.
For any $t\in\{1,\dots,T\}$, using properties of the coupled Euler--Maruyama discretization and Markov's inequality, we have 
$$
\mathbb{P}_{\theta,\nu}^{l-1,l}\Big( \big\| X_t^{l,\star} - X_t^{l-1,\star} \big\|_2
>C'\Delta_l^{\beta}\Big) \leq C\Delta_l^{\alpha(\frac{1}{2}-\beta)}
$$
for any $\alpha>0$, where $\mathbb{P}_{\theta,\nu}^{l-1,l}$ denotes probability under $\nu_{\theta}^{l-1,l}$. 
Similarly, for any $t\in\{1,\dots,T\}$, it follows from Remark \ref{rem:euler_grad} that 
$$
\mathbb{P}_{\theta,\nu}^{l-1,l}\Big( \big\| G_{\theta,t}^l(X_{0:t}^{l,\star}) - G_{\theta,t}^{l-1}(X_{0:t}^{l-1,\star}) \big\|_2
>C'\Delta_l^{\beta}\Big) \leq C\Delta_l^{\alpha(\frac{1}{2}-\beta)}
$$
for any $\alpha>0$. 
Hence there exists a constant $C<\infty$, that depends on $T$ but not $l$, such that if $\alpha=(1+\varrho)$
\begin{equation}\label{eq:cccpft17}
T_3 \leq C\Delta_l^{(\frac{1}{2}-\beta)}.
\end{equation}
For $T_4$, one can use the results discussed in Remark \ref{rem:g_bound_cccpf}, along with the above argument
(below \eqref{eq:cccpft18}) to control terms such as $\mathbb{E}_{\theta,\nu}^{l-1,l}\big[
\sum_{p=1}^t\sum_{s=l-1}^l \big\|G_{\theta,p-1:p}^s(X_{p-1}^{s,\star},\dots,X_p^{s,\star})\big\|_2^r
\big]$ to deduce that $T_4\leq C$. 
Thus we have shown that
\begin{equation}\label{eq:cccpft16}
T_2 \leq C\Delta_l^{(\frac{1}{2}-\beta)}.
\end{equation}
Combining \eqref{eq:cccpft14}-\eqref{eq:cccpft16} completes the proof.
\end{proof}

\begin{lemma}\label{lem:ini_prob}
Under Assumptions~\ref{ass:D1} and \ref{ass:D2}, for any $(T,\theta,C',\delta,\gamma)\in \mathbb{N}\times\Theta\times\mathbb{R}^+\times\mathbb{R}^+\times(0,\frac{1}{(1+\delta)})$, there exists a constant $C<\infty$ such that for any $(l,\beta,N)\in\mathbb{N}\times(0,\frac{1}{2})\times\{2,3,\dots\}$
$$
\check{\mathbb{E}}_{\theta,\nu}^{l-1,l} \Big[
\mathbb{I}_{(\mathsf{B}^l_{\beta,C'}\cap \mathsf{G}^l_{\beta,C'})^c}(X_{0:T}^{l},X_{0:T}^{l-1})\Big] \leq
C(\Delta_l)^{\{\beta(\frac{1}{(1+\delta)}-\gamma)\}\wedge(\frac{1}{2}-\beta)},
$$
and
$$
\check{\mathbb{E}}_{\theta,\nu}^{l-1,l} \Big[
\mathbb{I}_{(\mathsf{B}^l_{\beta,C'}\cap \mathsf{G}^l_{\beta,C'})^c}(\bar{X}_{0:T}^{l},\bar{X}_{0:T}^{l-1})\Big] \leq
C\Delta_l^{\frac{1}{2}-\beta}.
$$
\end{lemma}

\begin{proof}
As the proof of the second inequality is contained within the calculations to obtain \eqref{eq:cccpft17}, we will only consider the first inequality. We have
\begin{equation}\label{eq:cccpft19}
\check{\mathbb{E}}_{\theta,\nu}^{l-1,l} \Big[
\mathbb{I}_{(\mathsf{B}^l_{\beta,C'}\cap \mathsf{G}^l_{\beta,C'})^c}(X_{0:T}^{l},X_{0:T}^{l-1})\Big]
\leq T_1 + T_2,
\end{equation}
where
\begin{align*}
T_1 & = \mathbb{E}_{\theta,\nu}^{l-1,l}\Big[
\bar{\mathbb{E}}_{\theta}^{l-1,l}\Big[\mathbb{I}_{(\mathsf{B}^l_{\beta,C'}\cap \mathsf{G}^l_{\beta,C'})^c}(X_{0:T}^{l},X_{0:T}^{l-1})\Big]
\mathbb{I}_{\mathsf{B}^l_{\beta,C'}\cap \mathsf{G}^l_{\beta,C'}}(X_{0:T}^{l,\star},X_{0:T}^{l-1,\star})
\Big], \\
T_2 & = \mathbb{E}_{\theta,\nu}^{l-1,l}\Big[\mathbb{I}_{(\mathsf{B}^l_{\beta,C'}\cap \mathsf{G}^l_{\beta,C'})^c}(X_{0:T}^{l,\star},X_{0:T}^{l-1,\star})\Big]. 
\end{align*}
By Lemma \ref{lem:cccpf_prob_trans}, we have
$$
T_1 \leq 
C(\Delta_l)^{\beta(\frac{1}{(1+\delta)}-\gamma)}
\mathbb{E}_{\theta,\nu}^{l-1,l}\Big[\Big(1+\sum_{p=1}^t\sum_{s=l-1}^l \{\|X_p^{s,\star}\|_2^{\gamma}+\big\|G_{\theta,p-1:p}^{s}(X_{p-1}^{s,\star},X_{p-1+\Delta_s:p}^{s,\star})\big\|_2^{\gamma}\}\Big)\Big].
$$
The expectation can be controlled using the argument below \eqref{eq:cccpft18}, so we have 
\begin{equation}\label{eq:cccpft20}
T_1 \leq 
C(\Delta_l)^{\beta(\frac{1}{(1+\delta)}-\gamma)}.
\end{equation}
For $T_2$, using the second inequality in the statement of the lemma
\begin{equation}\label{eq:cccpft21}
T_2\leq C\Delta_l^{\frac{1}{2}-\beta}.
\end{equation}
Combining \eqref{eq:cccpft19} with \eqref{eq:cccpft20} and \eqref{eq:cccpft21} concludes the proof.
\end{proof}

\subsubsection{Results associated to score estimation methodology} %Algorithm \ref{alg:xi_l_comp}}
\label{sec:analysis_fourth}
We now study the two pairs of CPF chains $(X_{0:T}^{l-1}(i),\bar{X}_{0:T}^{l-1}(i))_{i=0}^{\infty}$ and 
$(X_{0:T}^{l}(i),\bar{X}_{0:T}^{l}(i))_{i=0}^{\infty}$ that are assumed to be run indefinitely even 
if both pairs of chains have met. We will denote the corresponding expectations by $\bar{\mathbb{E}}_{\theta}^{l-1,l}$.  
For any level $l\in\mathbb{N}_0$ and any probability measure $\pi$ defined on $\mathsf{X}^l$, 
we denote by $\mathbb{L}_2(\pi)$ the set of all measurable functions $\psi:\mathsf{X}^l\rightarrow\mathbb{R}$ such that 
$\pi(\psi^2)=\int_{\mathsf{X}^l}\psi(x)^2\pi(dx)\in(0,\infty)$. 
The following results will involve the smoothing distribution $\pi_{\theta}^l$ defined in \eqref{eqn:smoothing_distribution}. 

\begin{lemma}\label{lem:control_l2_mc}
Under Assumptions~\ref{ass:D1} and \ref{ass:D2}, for any $(T,\theta)\in \mathbb{N}\times\Theta$, there exists $(\varepsilon,C)\in(0,1)\times\mathbb{R}^+$ such that for any
$(l,N,i,\psi)\in\mathbb{N} \times\{2,3,\dots\}\times \mathbb{N}\times\mathbb{L}_2(\pi_{\theta}^l)$
$$
\bar{\mathbb{E}}_{\theta}^{l-1,l}[\psi(X_{0:T}^{l}(i))]\vee 
\bar{\mathbb{E}}_{\theta}^{l-1,l}[\psi(\bar{X}_{0:T}^{l}(i-1))]
\leq C\varepsilon^i\pi_{\theta}^l(\psi^2)^{1/2} + \pi_{\theta}^l(|\psi |).
$$
Also if $\psi\in\mathbb{L}_2(\pi_{\theta}^{l-1})$ then
$$
\bar{\mathbb{E}}_{\theta}^{l-1,l}[\psi(X_{0:T}^{l-1}(i))]\vee 
\bar{\mathbb{E}}_{\theta}^{l-1,l}[\psi(\bar{X}_{0:T}^{l-1}(i-1))]
\leq C\varepsilon^i\pi_{\theta}^{l-1}(\psi^2)^{1/2} + \pi_{\theta}^{l-1}(|\psi |).
$$
\end{lemma}

\begin{proof}
We will prove the result for $\bar{\mathbb{E}}_{\theta}^{l-1,l}[\psi(X_{0:T}^{l}(i))]$ only. The other results can be obtained in a similar way.

Marginally, the sequence $(X_{0:T}^l(i))_{i=1}^{\infty}$ is a Markov chain that has the initial distribution
$$
\check{\nu}_{\theta}^l(dx_{0:T}) = \int_{\mathsf{X}^l} \nu_{\theta}^l(dx_{0:T}^{\star})
%\textcolor{Blue}{\delta_{x_{\star}}(u_0)} \Big\{\prod_{k=1}^{K_l}p_{\theta}^l(u_{(k-1)\Delta_l},u_{k\Delta_l})\Big\}
M_{\theta}^l(dx_{0:T}|x_{0:T}^{\star})dx_{0:T}^{\star}
$$ 
and Markov transition kernel $M_{\theta}^l$ as described in Algorithm~\ref{alg:CPF}. % \ref{alg:cond_pf_0}.
By \citet[Theorem 1b]{andrieu2018uniform}, one has 
$$
\big|\bar{\mathbb{E}}_{\theta}^{l-1,l}[\psi(X_{0:T}^{l}(i))] -\pi_{\theta}^l(\psi) \big| \leq \Big(\int_{\mathsf{X}^l}
%\textcolor{Blue}{\delta_{x_{\star}}(u_0)}\Big\{\prod_{k=1}^{K_l}p_{\theta}^l(u_{(k-1)\Delta_l},u_{k\Delta_l})\Big\}^2
\frac{\nu_{\theta}^l(x_{0:T}^{\star})}{\pi_{\theta}^l(x_{0:T}^{\star})}
\nu_{\theta}^l(dx_{0:T}^{\star})\Big)
\varepsilon^{i+1}\pi_{\theta}^l(\psi^2)^{1/2},
$$
where we note that the extra power in $\varepsilon$ follows as $X_{0:T}^l(0)\sim\nu_{\theta}^l$. 
Using Assumptions~\ref{ass:D2}, it follows that
$$
\big|\bar{\mathbb{E}}_{\theta}^{l-1,l}[\psi(X_{0:T}^{l}(i))] - \pi_{\theta}^l(\psi)\big| \leq C\varepsilon^i\pi_{\theta}^l(\psi^2)^{1/2},
$$
and from here the proof is easily completed.
\end{proof}

\begin{lemma}\label{lem:mc_first_lem}
Under Assumptions~\ref{ass:D1} and \ref{ass:D2}, for any $(T,r,\theta,C')\in \mathbb{N}\times[1,\infty)\times\Theta\times\mathbb{R}^+$, there exists a constant $C<\infty$ such that for any $(l,\beta,N,\delta,i)\in\mathbb{N} \times\mathbb{R}^+\times\{2,3,\dots\}\times\mathbb{R}^+\times \mathbb{N}$
$$
\bar{\mathbb{E}}_{\theta}^{l-1,l}\Big[\big\|G_{\theta}^l(X_{0:T}^{l}(i)) - G_{\theta}^{l-1}(X_{0:T}^{l-1}(i))\big\|_2^r
\mathbb{I}_{\mathsf{B}^l_{\beta,C'}\cap \mathsf{G}^l_{\beta,C'}}(X_{0:T}^{l}(i-1),X_{0:T}^{l-1}(i-1))
\Big]^{1/r}
\leq C(\Delta_l^{\frac{1}{2}\wedge\beta})^{\frac{1}{r(1+\delta)}},
$$
and
$$
\bar{\mathbb{E}}_{\theta}^{l-1,l}\Big[\big\|G_{\theta}^l(\bar{X}_{0:T}^{l}(i)) - G_{\theta}^{l-1}(\bar{X}_{0:T}^{l-1}(i))\big\|_2^r
\mathbb{I}_{\mathsf{B}^l_{\beta,C'}\cap \mathsf{G}^l_{\beta,C'}}(\bar{X}_{0:T}^{l}(i-1),\bar{X}_{0:T}^{l-1}(i-1))
\Big]^{1/r}
\leq C(\Delta_l^{\frac{1}{2}\wedge\beta})^{\frac{1}{r(1+\delta)}}.
$$
\end{lemma}

\begin{proof}
We only consider the first inequality as it is the same proof for the second inequality. By Corollary \ref{cor:cccpf_first_res},
we have the upper-bound
\begin{align*}
&\bar{\mathbb{E}}_{\theta}^{l-1,l}\Big[\big\|G_{\theta}^l(X_{0:T}^{l}(i)) - G_{\theta}^{l-1}(X_{0:T}^{l-1}(i))\big\|_2^r
\mathbb{I}_{\mathsf{B}^l_{\beta,C'}\cap \mathsf{G}^l_{\beta,C'}}(X_{0:T}^{l}(i-1),X_{0:T}^{l-1}(i-1))
\Big]^{1/r}\\
&\leq C(\Delta_l^{\frac{1}{2}\wedge\beta})^{\frac{1}{r(1+\delta)}}
\bar{\mathbb{E}}_{\theta}^{l-1,l}
\Big[\Big(1+\sum_{p=1}^T\sum_{s=l-1}^l \|G_{\theta,p-1:p}^s(X_{p-1}^{s}(i-1),X_{p-1+\Delta_s:p}^{s}(i-1))\|_2^r\Big)\Big].
\end{align*}
Note that $\|G_{\theta,p-1:p}^s(X_{p-1}^{s},X_{p-1+\Delta_s:p}^{s})\|_2^r\in\mathbb{L}_2(\pi_{\theta}^s)$ (see the argument below \eqref{eq:cccpft18}). In addition the expectation of the square of this function w.r.t.\ $\pi_{\theta}^s$ is bounded uniformly in $s$ (one can use Assumptions~\ref{ass:D2} to upper-bound expectations w.r.t.\ $\pi_{\theta}^s$ by expectations w.r.t.\ $\nu_{\theta}^s$). 
Hence using Lemma \ref{lem:control_l2_mc}, we obtain
\begin{equation}\label{eq:cccpft22}
\bar{\mathbb{E}}_{\theta}^{l-1,l}
\Big[\Big(1+\sum_{p=1}^T\sum_{s=l-1}^l \big\|G_{\theta,p-1:p}^s(X_{p-1}^{s}(i-1),X_{p-1:\Delta_s:p}^{s}(i-1))\big\|_2^r\Big)\Big] \leq C,
\end{equation}
which allows us to conclude the proof. 
\end{proof}

\begin{lemma}\label{lem:mc_second_lem}
Under Assumptions~\ref{ass:D1} and \ref{ass:D2}, for any $(T,\theta,C',\beta,\delta,\gamma)\in \mathbb{N}\times\Theta\times(\mathbb{R}^+)^3
\times(0,\frac{\frac{1}{2}\wedge\beta}{\beta(1+\delta)})$, there exists a constant $C<\infty$ such that for any $(l,N,i)\in\mathbb{N}\times\{2,3,\dots\}\times\mathbb{N}$
$$
\bar{\mathbb{E}}_{\theta}^{l-1,l}\Big[
\mathbb{I}_{(\mathsf{B}^l_{\beta,C'}\cap \mathsf{G}^l_{\beta,C'})^c}(X_{0:T}^{l}(i),X_{0:T}^{l-1}(i))
\mathbb{I}_{\mathsf{B}^l_{\beta,C'}\cap \mathsf{G}^l_{\beta,C'}}(X_{0:T}^{l}(i-1),X_{0:T}^{l-1}(i-1))
\Big] \leq C(\Delta_l)^{\frac{{\frac{1}{2}\wedge\beta}}{(1+\delta)}-\gamma\beta},
$$
and
$$
\bar{\mathbb{E}}_{\theta}^{l-1,l}\Big[
\mathbb{I}_{(\mathsf{B}^l_{\beta,C'}\cap \mathsf{G}^l_{\beta,C'})^c}(\bar{X}_{0:T}^{l}(i),\bar{X}_{0:T}^{l-1}(i))
\mathbb{I}_{\mathsf{B}^l_{\beta,C'}\cap \mathsf{G}^l_{\beta,C'}}(\bar{X}_{0:T}^{l}(i-1),\bar{X}_{0:T}^{l-1}(i-1))
\Big] \leq C(\Delta_l)^{\frac{{\frac{1}{2}\wedge\beta}}{(1+\delta)}-\gamma\beta}.
$$
\end{lemma}

\begin{proof}
The proof is essentially identical to that of Lemma \ref{lem:mc_first_lem}, except one must use Lemma \ref{lem:cccpf_prob_trans} instead of Corollary \ref{cor:cccpf_first_res}.
\end{proof}

\begin{lemma}\label{lem:mc_third_lem}
Under Assumptions~\ref{ass:D1} and \ref{ass:D2}, for any $(T,\theta,C',\beta,\delta,\gamma)\in \mathbb{N}\times\Theta\times\mathbb{R}^+\times(0,\frac{1}{2})\times\mathbb{R}^+
\times(0,\frac{1}{(1+\delta)})$, there exists a constant $C<\infty$ such that for any $(l,N,i)\in\mathbb{N}\times\{2,3,\dots\}\times\mathbb{N}$
$$
\bar{\mathbb{E}}_{\theta}^{l-1,l}\Big[
\mathbb{I}_{(\mathsf{B}^l_{\beta,C'}\cap \mathsf{G}^l_{\beta,C'})^c}(X_{0:T}^{l}(i),X_{0:T}^{l-1}(i))
\Big] \leq C(i+1)(\Delta_l)^{\{\beta(\frac{1}{(1+\delta)}-\gamma)\}\wedge(\frac{1}{2}-\beta)},
$$
and
$$
\bar{\mathbb{E}}_{\theta}^{l-1,l}\Big[
\mathbb{I}_{(\mathsf{B}^l_{\beta,C'}\cap \mathsf{G}^l_{\beta,C'})^c}(\bar{X}_{0:T}^{l}(i),\bar{X}_{0:T}^{l-1}(i))
\Big] \leq C(i+1)(\Delta_l)^{\{\beta(\frac{1}{(1+\delta)}-\gamma)\}\wedge(\frac{1}{2}-\beta)}.
$$
\end{lemma}

\begin{proof}
We only consider the first inequality as it is the same proof for the second inequality.
The proof is  by induction on $i$. The initialization follows by Lemma \ref{lem:ini_prob}. For the induction step, one
can easily conclude by using Lemma \ref{lem:mc_second_lem} and the induction hypothesis.
\end{proof}

\begin{lemma}\label{lem:mc_fourth_lem}
Under Assumptions~\ref{ass:D1} and \ref{ass:D2}, for any $(T,r,\theta,\beta,\delta,\gamma)\in \mathbb{N}\times[1,\infty)\times\Theta\times (0,\frac{1}{2})\times\mathbb{R}^+
\times(0,\frac{1}{(1+\delta)})$, there exists a constant $C<\infty$ such that for any $(l,N,i)\in\mathbb{N} \times\{2,3,\dots\}\times \mathbb{N}$
$$
\bar{\mathbb{E}}_{\theta}^{l-1,l}\Big[\big\|G_{\theta}^l(X_{0:T}^{l}(i)) - G_{\theta}^{l-1}(X_{0:T}^{l-1}(i))\big\|_2^r
\Big]^{1/r}
\leq C(i+1)\Delta_l^\phi,
$$
where $\phi=\frac{\beta}{r(1+\delta)}\wedge \frac{1}{(1+\delta)}(\{\beta(\frac{1}{(1+\delta)}-\gamma)\}\wedge(\frac{1}{2}-\beta))$, 
and
$$
\bar{\mathbb{E}}_{\theta}^{l-1,l}\Big[\big\|G_{\theta}^l(\bar{X}_{0:T}^{l}(i)) - G_{\theta}^{l-1}(\bar{X}_{0:T}^{l-1}(i))\big\|_2^r\Big]^{1/r}
\leq C(i+1)\Delta_l^\phi.
$$
\end{lemma}

\begin{proof}
We only consider the first inequality as it is the same proof for the second inequality.
The proof is by induction on $i$. The initialization follows from Lemma \ref{lem:ini_g_l}. For the induction step, one has
$$
\bar{\mathbb{E}}_{\theta}^{l-1,l}\Big[\big\|G_{\theta}^l(X_{0:T}^{l}(i)) - G_{\theta}^{l-1}(X_{0:T}^{l-1}(i))\big\|_2^r
\Big]^{1/r} \leq  C(T_1 + T_2),
$$
where 
\begin{align*}
T_1 & = \bar{\mathbb{E}}_{\theta}^{l-1,l}\Big[\big\|G_{\theta}^l(X_{0:T}^{l}(i)) - G_{\theta}^{l-1}(X_{0:T}^{l-1}(i))\big\|_2^r
\mathbb{I}_{\mathsf{B}^l_{\beta,C'}\cap \mathsf{G}^l_{\beta,C'}}(X_{0:T}^{l}(i-1),X_{0:T}^{l-1}(i-1))
\Big]^{1/r}, \\
T_2 & = \bar{\mathbb{E}}_{\theta}^{l-1,l}\Big[\big\|G_{\theta}^l(X_{0:T}^{l}(i)) - G_{\theta}^{l-1}(X_{0:T}^{l-1}(i))\big\|_2^r
\mathbb{I}_{(\mathsf{B}^l_{\beta,C'}\cap \mathsf{G}^l_{\beta,C'})^c}(X_{0:T}^{l}(i-1),X_{0:T}^{l-1}(i-1))
\Big]^{1/r},
\end{align*}
for any $0<C'<\infty$.
For $T_1$, one can apply Lemma \ref{lem:mc_first_lem} to obtain $T_1\leq C\Delta_l^\phi$. 
For $T_2$, one can use H\"older's inequality to get the bound
$$
T_2 \leq
\bar{\mathbb{E}}_{\theta}^{l-1,l}\Big[\big\|G_{\theta}^l(X_{0:T}^{l}(i)) - G_{\theta}^{l-1}(X_{0:T}^{l-1}(i))\big\|_2^{r\frac{1+\delta}{\delta}}\Big]^{\frac{\delta}{r(1+\delta)}}
\bar{\mathbb{E}}_{\theta}^{l-1,l}\Big[
\mathbb{I}_{(\mathsf{B}^l_{\beta,C'}\cap \mathsf{G}^l_{\beta,C'})^c}(X_{0:T}^{l}(i),X_{0:T}^{l-1}(i))
\Big]^{\frac{1}{(1+\delta)}}.
$$
To deal with the leftmost expectation on the R.H.S.~one can rely on the same argument that led to \eqref{eq:cccpft22} and for the other expectation one can use Lemma \ref{lem:mc_third_lem}. This allows us to obtain
$$
T_2\leq C(i+1)\Delta_l^\phi,
$$
and conclude the proof.
\end{proof}

In the following, we will employ the notation $\mathsf{A}_i=\{j\in\mathbb{N}:j> i\}$ for $i\in\mathbb{N}$.

\begin{lemma}\label{lem:coup_prob}
Under Assumptions~\ref{ass:D1} and \ref{ass:D2}, for any $(T,\theta,N)\in \mathbb{N}\times\Theta\times\{2,3,\dots\}$, there exists $(\varepsilon,C)\in(0,1)\times\mathbb{R}^+$ such that for any $(l,i)\in\mathbb{N}^2$
$$
\bar{\mathbb{E}}_{\theta}^{l-1,l} \big[ \mathbb{I}_{\mathsf{A}_i}\big(\bar{\tau}_{\theta}^l \big) \big] \leq C\varepsilon^{i}.
$$
\end{lemma}

\begin{proof}
We have
$$
\bar{\mathbb{E}}_{\theta}^{l-1,l}[\mathbb{I}_{\mathsf{A}_i}(\bar{\tau}_{\theta}^l)] \leq
\bar{\mathbb{E}}_{\theta}^{l-1,l}[\mathbb{I}_{\mathsf{A}_i\times \mathsf{A}_i}(\tau_{\theta}^{l-1},\tau_{\theta}^{l})] + 
\bar{\mathbb{E}}_{\theta}^{l-1,l}[\mathbb{I}_{\mathsf{A}_i^c\times \mathsf{A}_i}(\tau_{\theta}^{l-1},\tau_{\theta}^{l})] +
\bar{\mathbb{E}}_{\theta}^{l-1,l}[\mathbb{I}_{\mathsf{A}_i\times \mathsf{A}_i^c}(\tau_{\theta}^{l-1},\tau_{\theta}^{l})].
$$
By Lemma \ref{lem:diag_prob_kernel}, we have $\bar{\mathbb{E}}_{\theta}^{l-1,l}[\mathbb{I}_{\mathsf{A}_i}(\tau_{\theta}^s)]\leq C\varepsilon^{i}$ 
for $s\in\{l-1,l\}$,
and so the proof is now easily completed.
\end{proof}

\begin{remark}
It is worth noting that Lemmata \ref{lem:prob_lower}, \ref{lem:diag_prob_kernel} and \ref{lem:coup_prob} are the only cases where 
our bounds have constants that depend upon $N$.
\end{remark}

\begin{lemma}\label{lem:mc_fifth_lem}
Under Assumptions~\ref{ass:D1} and \ref{ass:D2}, for any $(T,\theta,\beta,\delta,\gamma,N,b)\in \mathbb{N}\times\Theta\times (0,\frac{1}{2})\times\mathbb{R}^+
\times(0,\frac{1}{(1+\delta)})\times\{2,3,\dots,\}\times\mathbb{N}_0$, there exists a constant $C<\infty$ such that for any $l\in\mathbb{N}$
$$
\bar{\mathbb{E}}_{\theta}^{l-1,l}\big[\|\widehat{I}_l(\theta)\|_2^2\big] \leq C\Delta_l^{2\phi},
$$
where $\phi=\frac{\beta}{2(1+\delta)}\wedge \frac{1}{(1+\delta)}(\{\beta(\frac{1}{(1+\delta)}-\gamma)\}\wedge(\frac{1}{2}-\beta))$.
\end{lemma}

\begin{proof}
We recall that $\widehat{I}_l(\theta)=\widehat{S}_l(\theta) - \widehat{S}_{l-1}(\theta)$, where 
$\widehat{S}_{l-1}(\theta)$ and $\widehat{S}_{l}(\theta)$ are time-averaged estimators of the form in \eqref{eqn:unbiased_discretized_score}. 
For level $s\in\{l-1,l\}$, note that we can rewrite the time-averaged estimator as 
$\widehat{S}_{s}(\theta) = (I-b+1)^{-1}\sum_{k=b}^I\widehat{S}_{s}^k(\theta)$ with 
\begin{align*}
	\widehat{S}_{s}^k(\theta) =  G_{\theta}^{s}(X_{0:T}^s(k)) + \sum_{i=b+1}^{\tau_{\theta}^s-1}\left(G_{\theta}^s(X_{0:T}^s(i)) 
	 - G_{\theta}^s(\bar{X}_{0:T}^s(i-1))\right).
\end{align*}
Hence we can rewrite 
\begin{align}\label{eqn:rewrite_increment}
	\widehat{I}_l(\theta) = \frac{1}{I-b+1} \sum_{k=b}^I \left(\widehat{S}_{l}^k(\theta) - \widehat{S}_{l-1}^k(\theta)\right).  
\end{align}
Since we have 
\begin{align}\label{eqn:rewrite_increment_start}
	\bar{\mathbb{E}}_{\theta}^{l-1,l}\big[ \|\widehat{I}_l(\theta)\|_2^2 \big] \leq C \sum_{k=b}^I 
	\bar{\mathbb{E}}_{\theta}^{l-1,l}\big[ \| \widehat{S}_{l}^k(\theta) - \widehat{S}_{l-1}^k(\theta) \|_2^2 \big] 		
\end{align}
using the representation in \eqref{eqn:rewrite_increment}, it suffices to establish 
$\bar{\mathbb{E}}_{\theta}^{l-1,l}\big[ \| \widehat{S}_{l}^k(\theta) - \widehat{S}_{l-1}^k(\theta) \|_2^2 \big] \leq C\Delta_l^{2\phi}$.		

We consider the decomposition
\begin{equation}\label{eq:main_res1}
	\bar{\mathbb{E}}_{\theta}^{l-1,l}\big[ \| \widehat{S}_{l}^k(\theta) - \widehat{S}_{l-1}^k(\theta) \|_2^2 \big] 		
	\leq C\sum_{j=1}^2 T_j, 
\end{equation}
where
\begin{align*}
T_1 & = \bar{\mathbb{E}}_{\theta}^{l-1,l}\Big[\|G_{\theta}^l(X_{0:T}^l(k))-G_{\theta}^{l-1}(X_{0:T}^{l-1}(k))\|_2^2\Big],\\
T_2 & = \bar{\mathbb{E}}_{\theta}^{l-1,l}\Big[\|\sum_{i=k+1}^{\bar{\tau}_{\theta}^l}\{G_{\theta}^l(X_{0:T}^l(i))-G_{\theta}^{l-1}(X_{0:T}^{l-1}(i))+
G_{\theta}^l(\bar{X}_{0:T}^l(i))-G_{\theta}^{l-1}(\bar{X}_{0:T}^{l-1}(i))\}\|_2^2\Big].
%,\\
%T_3 & = \bar{\mathbb{E}}_{\theta}^{l-1,l}\Big[
%\sum_{i=k+1}^{\bar{\tau}_{\theta}^l}
%\|G_{\theta}^l(\bar{X}_{0:T}^l(i))-G_{\theta}^{l-1}(\bar{X}_{0:T}^{l-1}(i))\|_2^2\Big].
\end{align*}
For $T_1$, one can apply Lemma \ref{lem:mc_fourth_lem} to obtain
\begin{equation}\label{eq:main_res2}
T_1 \leq  C\Delta_l^{2\phi}.
\end{equation}
For $T_2$, we have
$$
T_2  = \bar{\mathbb{E}}_{\theta}^{l-1,l}\Big[\|\sum_{i=k+1}^{\infty}\mathbb{I}_{\mathsf{A}_i}(\bar{\tau}_{\theta}^l)\{G_{\theta}^l(X_{0:T}^l(i))-G_{\theta}^{l-1}(X_{0:T}^{l-1}(i))+
G_{\theta}^l(\bar{X}_{0:T}^l(i))-G_{\theta}^{l-1}(\bar{X}_{0:T}^{l-1}(i))\}\|_2^2\Big].
$$
To shorten the notations, set 
$$
\upsilon_i = G_{\theta}^l(X_{0:T}^l(i))-G_{\theta}^{l-1}(X_{0:T}^{l-1}(i))+
G_{\theta}^l(\bar{X}_{0:T}^l(i))-G_{\theta}^{l-1}(\bar{X}_{0:T}^{l-1}(i))
$$
where we denote the $j^{th}$-component of $\upsilon_i$ as $[\upsilon_i]^j$, $j\in\{1,\dots,d\}$. Then by application of Minkowski's inequality, we have
$$
T_2 \leq \sum_{j=1}^{d_{\theta}}\Bigg(\sum_{i=k+1}^{\infty} \bar{\mathbb{E}}_{\theta}^{l-1,l}[\mathbb{I}_{\mathsf{A}_i}(\bar{\tau}_{\theta}^l)\{[\upsilon_i]^j\}^2]^{1/2}\Bigg)^2.
$$
Then applying the Cauchy-Schwarz inequality
$$
T_2 \leq \sum_{j=1}^{d_{\theta}}\Bigg(\sum_{i=k+1}^{\infty} \bar{\mathbb{E}}_{\theta}^{l-1,l}[\mathbb{I}_{\mathsf{A}_i}(\bar{\tau}_{\theta}^l)]^{1/4}\bar{\mathbb{E}}_{\theta}^{l-1,l}[\{[\upsilon_i]^j\}^4]^{1/4}\Bigg)^2.
$$
Now, using standard properties of the $\mathbb{L}_2$-norm along with Lemma \ref{lem:coup_prob}
$$
T_2 \leq C\Big(\sum_{i=k+1}^{\infty} (\varepsilon^{1/4})^i\bar{\mathbb{E}}_{\theta}^{l-1,l}[\|\upsilon_i\|_2^4]^{1/4}\Big)^2.
$$
It is simple to ascertain that:
$$
\bar{\mathbb{E}}_{\theta}^{l-1,l}[\|\upsilon_i\|_2^4]^{1/4} \leq C\Big(
\bar{\mathbb{E}}_{\theta}^{l-1,l}\Big[\big\|G_{\theta}^l(X_{0:T}^l(i))-G_{\theta}^{l-1}(X_{0:T}^{l-1}(i))\big\|_2^4\Big]+
\bar{\mathbb{E}}_{\theta}^{l-1,l}\Big[\big\|G_{\theta}^l(\bar{X}_{0:T}^l(i))-G_{\theta}^{l-1}(\bar{X}_{0:T}^{l-1}(i))\big\|_2^4\Big]
\Big)^{1/4}.
$$
Therefore, applying Lemma \ref{lem:mc_fourth_lem} gives
\begin{equation}\label{eq:main_res3}
T_2 \leq C\Delta_l^{2\phi}\Big(\sum_{i=k+1}^{\infty} (\varepsilon^{1/4})^i(i+1)\Big)^2 \leq C\Delta_l^{2\phi}.
\end{equation}
Combining \eqref{eqn:rewrite_increment_start}-\eqref{eq:main_res3} concludes the proof.

\end{proof}

\begin{remark}
The strategy in the proof of Lemma \ref{lem:mc_fifth_lem} can be improved by using martingale methods and Wald's equality for Markov chains as considered in \cite{jasra2020unbiased}. 
This strategy was not adopted as it would require more complicated arguments given the technical complexity of the problem and algorithms in this article.
\end{remark}

\begin{remark}\label{rem:rate_rem1}
A better rate of $\phi$ can be obtained in Lemma \ref{lem:mc_fifth_lem} if we consider the case of constant diffusion coefficient $\sigma$.
\end{remark}

\begin{remark}\label{rem:exp_xi_0}\label{rem:zero_case}
One can employ the approaches in Lemmata \ref{lem:control_l2_mc},  \ref{lem:coup_prob} and \ref{lem:mc_fifth_lem} to establish that the expected value of $\widehat{I}_0(\theta)$ is upper-bounded by a finite constant.
\end{remark}

\begin{proof}[Proof of Theorem \ref{theo:ub}]
We have to establish that \eqref{eq:ub1} and \eqref{eq:ub2} hold for some choice of PMF $(P_l)_{l=0}^{\infty}$. 
The unbiasedness property in \eqref{eq:ub1} can be established using the same approach as in \citet[Theorem 3.1]{jacob2020smoothing}. Assumption 1 of \citet{jacob2020smoothing} is implied by Assumption~\ref{ass:D2}$(iii)$; 
Assumption 2 of \citet{jacob2020smoothing} can be verified by inspecting the construction in the proof of Lemma \ref{lem:diag_prob_kernel}; and Assumption 3 of \citet{jacob2020smoothing} follows from the fact that 
$\|G_{\theta}^l\|_2^r\in\mathbb{L}_2(\pi_{\theta}^l)$ for any $(l,r)\in\mathbb{N}_0\times[1,\infty)$ 
and \cite[Theorem 1b]{andrieu2018uniform}. 
For the condition in \eqref{eq:ub2}, we apply Theorem \ref{prop:conv_grad_log_like}, Lemma \ref{lem:mc_fifth_lem} and Remark \ref{rem:exp_xi_0} to obtain 
\begin{align*}
	\sum_{l=0}^\infty \mathcal{P}_l^{-1}\left\lbrace \mathrm{Var}\left[\widehat{I}_l(\theta)^j\right] + 
	\left(S_{l}(\theta)^j-S(\theta)^j\right)^2\right\rbrace < C\sum_{l=0}^\infty\frac{\Delta_l^{2\phi}}{\mathcal{P}_l}
\end{align*}
for all $j\in\{1,\ldots,d_{\theta}\}$, where $\mathrm{Var}$ denotes variance under $\bar{\mathbb{E}}_{\theta}^{l-1,l}$ 
for all $l\in\mathbb{N}$. 
We can conclude the proof by selecting for instance $P_l\propto\Delta_l^{2\phi\alpha}$ for any $\alpha\in(0,1)$.
\end{proof}

\begin{remark}
The approach in Lemma \ref{lem:control_l2_mc} also suggests an alternative method of proof. If one could identify the invariant distribution of the ML-CPF kernel $M_{\theta}^{l-1,l}$ (Algorithm~\ref{alg:ML-CPF}) and establish 
an ergodic theorem as in \citet[Theorem 1b]{andrieu2018uniform}, 
one could then study the expectation of differences of the type $\|G_{\theta}^l-G_{\theta}^{l-1}\|_2^r$ 
under the invariant distribution. Characterizing the invariant distribution could follow the ideas in \cite{jasra2018central}. 
This potentially interesting strategy is left as a topic for future work. 
\end{remark}

\section{Model-specific expressions}

\subsection{Ornstein--Uhlenbeck process}\label{app:OU}
For this example, we have $\Sigma(x)=\sigma^2$, $b_{\theta}(x)=\sigma^{-1}a_{\theta}(x)$ for $x\in\mathbb{R}$ 
and $\theta\in\Theta=(0,\infty)\times\mathbb{R}\times(0,\infty)$. 
To evaluate \eqref{eqn:smoothing_functional} and \eqref{eqn:discrete_test_function}, the gradients required are given by 
\begin{align}
	\nabla_{\theta}a_{\theta}(x) =\left((\theta_{2}-x),\theta_{1},0\right),\quad 
	\nabla_{\theta}\log g_{\theta}(y|x)=\left(0,0,-\frac{1}{2\theta_{3}}+\frac{(y-x)^{2}}{2\theta_{3}^{2}}\right), 
\end{align}
for $x\in\mathbb{R}$ and $\theta=(\theta_1,\theta_2,\theta_3)\in\Theta$. 
In this example, the score function \eqref{eqn:score_function} can be computed using 
\begin{align}
	\nabla_{\theta}\log p_{\theta}(y_{1:T}) = \sum_{t=1}^T\int_{\mathbb{R}^2}
	\left\{ \nabla_{\theta}\log p_{\theta}(dx_{t}|x_{t-1})+\nabla_{\theta}\log g_{\theta}(y_{t}|x_{t})\right\} 
	p_{\theta}(dx_{t-1},dx_{t}|y_{1:T}),
\end{align}
where the transition kernel of the SDE \eqref{eqn:OU_SDE} on a unit interval is 
\begin{align}
	p_{\theta}(dx_{t}|x_{t-1}) = \mathcal{N}\left(x_{t};\theta_{2}+(x_{t-1}-\theta_{2})\exp(-\theta_{1}),
	\frac{\sigma^{2}(1-\exp(-2\theta_{1}))}{2\theta_{1}}\right)dx_{t},
\end{align}
and the marginal of the smoothing distribution $p_{\theta}(dx_{t-1},dx_{t}|y_{1:T})$ is a Gaussian 
distribution whose mean and covariance can be obtained using a Kalman smoother. 

\subsection{Logistic diffusion model for population dynamics of red kangaroos}\label{app:logistic_diffusion}
In this application, we have $\Sigma(x)=1$ and $b_\theta(x) = a_{\theta}(x)$ for $x\in\mathbb{R}$ 
and $\theta\in\Theta=\mathbb{R}\times(0,\infty)^3$. 
Evaluation of \eqref{eqn:smoothing_functional} and \eqref{eqn:discrete_test_function} require the following expressions. 
Firstly, we have 
\begin{align}
	\nabla_{\theta}\log\mu_{\theta}(x) = (0, 0, \partial_{\theta_3}\log\mu_{\theta}(x), 0),\quad 
	\partial_{\theta_3}\log\mu_{\theta}(x) = 
	\frac{1}{\theta_{3}}-\frac{\theta_{3}}{10^{2}}(x-5/\theta_{3})^{2}-\frac{5}{10^{2}}(x-5/\theta_{3}),
\end{align}
and 
\begin{align}
	\nabla_{\theta}a_{\theta}(x)
	=\left(\frac{1}{\theta_{3}},-\frac{1}{\theta_{3}}\exp(\theta_{3}x),-\frac{\theta_{1}}{\theta_{3}^{2}} 
	- \frac{\theta_{2}}{\theta_{3}^{2}}\exp(\theta_{3}x)(\theta_{3}x-1),0\right),	
\end{align}
for $x\in\mathbb{R}$ and $\theta=(\theta_1,\theta_2,\theta_3,\theta_4)\in\Theta$. 
The observation density can be written as 
\begin{align}
	g_{\theta}(y|x)&=\mathcal{NB}(y^{1};\theta_{4},\exp(\theta_3x))\mathcal{NB}(y^{2};\theta_{4},\exp(\theta_3x)) \\
	&=\frac{\Gamma(y^{1}+\theta_{4})\Gamma(y^{2}+\theta_{4})}{\Gamma(\theta_{4})^{2}(y^{1})!(y^{2})!}\left(\frac{\theta_{4}}{\theta_{4}+\exp(\theta_{3}x)}\right)^{2\theta_{4}}\left(\frac{\exp(\theta_{3}x)}{\theta_{4}+\exp(\theta_{3}x)}\right)^{y^{1}+y^{2}}\notag
\end{align}
for $y=(y^1,y^2)$ and $x\in\mathbb{R}$. Hence 
\begin{align}
	\nabla_{\theta}\log g_{\theta}(y|x) = (0,0,\partial_{\theta_3}\log g_{\theta}(y|x), \partial_{\theta_4}\log g_{\theta}(y|x)),
\end{align}
with 
\begin{align}
	\partial_{\theta_{3}}\log g_{\theta}(y|x) = -\frac{2\theta_{4}x\exp(\theta_{3}x)}{\theta_{4}
	+\exp(\theta_{3}x)}+(y^{1}+y^{2})x \left(1 - \frac{\exp(\theta_{3}x)}{\theta_{4}+\exp(\theta_{3}x)}\right) ,
\end{align}
and 
\begin{align}
	\partial_{\theta_{4}}\log g_{\theta}(y|x) &= \psi(y^{1}+\theta_{4})+\psi(y^{2}+\theta_{4})-2\psi(\theta_{4})
	+2\left\{ \log(\theta_{4})-\log(\theta_{4}+\exp(\theta_{3}x))\right\} \\
	&+2\left( 1-\frac{\theta_{4}}{\theta_{4}+\exp(\theta_{3}x)}\right) -\frac{(y^{1}+y^{2})}{(\theta_{4}+\exp(\theta_{3}x))},\notag
\end{align}
where $x\mapsto\psi(x)=(d/dx)\log\Gamma(x)$ denotes the digamma function.

\subsection{Neural network model for grid cells in the medial entorhinal cortex}\label{app:neural_network}
In this application, we have $\Sigma(x)=I_{2}$ and $b_{\theta}(x)=a_{\theta}(x)$ for $x=(x^1,x^2)\in\mathbb{R}^2$ and 
$\theta\in\Theta$. 
The following expressions are needed to evaluate \eqref{eqn:smoothing_functional} and \eqref{eqn:discrete_test_function}. 
The non-zero entries of the Jacobian matrix $\nabla_{\theta}a_{\theta}(x)\in\mathbb{R}^{d\times d_{\theta}}$ are given by

\begin{minipage}{.48\textwidth}
\begin{align*}
\partial_{\alpha_{1}}a_{\theta}^{1}(x) & =\tanh(\beta_{1}\sigma_{2}x^{2}+\gamma_{1})/\sigma_{1},\\
\partial_{\beta_{1}}a_{\theta}^{1}(x) & =\alpha_{1}\sigma_{2}x^{2}\left(1-\tanh^{2}(\beta_{1}\sigma_{2}x^{2}+\gamma_{1})\right)/\sigma_{1},\\
\partial_{\gamma_{1}}a_{\theta}^{1}(x) & =\alpha_{1}\left(1-\tanh^{2}(\beta_{1}\sigma_{2}x^{2}+\gamma_{1})\right)/\sigma_{1},\\
\partial_{\delta_{1}}a_{\theta}^{1}(x) & =-x^{1},\\
\partial_{\sigma_{1}}a_{\theta}^{1}(x) & =-\alpha_{1}\tanh(\beta_{1}\sigma_{2}x^{2}+\gamma_{1})/\sigma_{1}^{2}\\
%\partial_{\kappa_{1}}a_{\theta}^{1}(x) & = 0,\\
%\partial_{\alpha_{2}}a_{\theta}^{1}(x) & = 0,\\
%\partial_{\beta_{2}}a_{\theta}^{1}(x) & =0,\\
%\partial_{\gamma_{2}}a_{\theta}^{1}(x) & =0,\\
%\partial_{\delta_{2}}a_{\theta}^{1}(x) & =0,\\
\partial_{\sigma_{2}}a_{\theta}^{1}(x) & =\alpha_{1}\beta_{1}x^{2}\left(1-\tanh^{2}(\beta_{1}\sigma_{2}x^{2}+\gamma_{1})\right)/\sigma_{1},\\
%\partial_{\kappa_{2}}a_{\theta}^{1}(x) & =0,
\end{align*}
\end{minipage} \quad
\begin{minipage}{.48\textwidth}
%and the entries in second row are 
\begin{align*}
%\partial_{\alpha_{1}}a_{\theta}^{2}(x) & =0,\\
%\partial_{\beta_{1}}a_{\theta}^{2}(x) & =0,\\
%\partial_{\gamma_{1}}a_{\theta}^{2}(x) & =0,\\
%\partial_{\delta_{1}}a_{\theta}^{2}(x) & =0,\\
%\partial_{\kappa_{1}}a_{\theta}^{2}(x) & =0,\\
\partial_{\alpha_{2}}a_{\theta}^{2}(x) & =\tanh(\beta_{2}\sigma_{1}x^{1}+\gamma_{2})/\sigma_{2},\\
\partial_{\beta_{2}}a_{\theta}^{2}(x) & =\alpha_{2}\sigma_{1}x^{1}\left(1-\tanh^{2}(\beta_{2}\sigma_{1}x^{1}+\gamma_{2})\right)/\sigma_{2},\\
\partial_{\gamma_{2}}a_{\theta}^{2}(x) & =\alpha_{2}\left(1-\tanh^{2}(\beta_{2}\sigma_{1}x^{1}+\gamma_{2})\right)/\sigma_{2},\\
\partial_{\delta_{2}}a_{\theta}^{2}(x) & =-x^{2},\\
\partial_{\sigma_{1}}a_{\theta}^{2}(x) & =\alpha_{2}\beta_{2}x^{1}\left(1-\tanh^{2}(\beta_{2}\sigma_{1}x^{1}+\gamma_{2})\right)/\sigma_{2},\\
\partial_{\sigma_{2}}a_{\theta}^{2}(x) & =-\alpha_{2}\tanh(\beta_{2}\sigma_{1}x^{1}+\gamma_{2})/\sigma_{2}^{2}.\\
%\partial_{\kappa_{2}}a_{\theta}^{2}(x) & =0.
\end{align*}
\end{minipage}
%\textcolor{red}{
%The log-observation density is 
%\begin{align*}
%\log g_{\theta}^{l}(y_{t_{p}}|x_{t_{p}}) & =\log\left(\textrm{Poi}\left(y_{t_{p}}^{1};\lambda_{1}(x_{t_{p}}^{1})\Delta_{l}\right)\right)+\log\left(\textrm{Poi}\left(y_{t_{p}}^{2};\lambda_{2}(x_{t_{p}}^{2})\Delta_{l}\right)\right)\\
% & =y_{t_{p}}^{1}\left(\log(\Delta_{l})+\kappa_{1}+x_{t_{p}}^{1}\right)-\Delta_{l}\exp\left(\kappa_{1}+x_{t_{p}}^{1}\right)-\log\left(y_{t_{p}}^{1}!\right)\\
% & +y_{t_{p}}^{2}\left(\log(\Delta_{l})+\kappa_{2}+x_{t_{p}}^{2}\right)-\Delta_{l}\exp\left(\kappa_{2}+x_{t_{p}}^{2}\right)-\log\left(y_{t_{p}}^{2}!\right)
%\end{align*}
%and its partial derivatives are all zero except 
%\begin{align*}
% & \partial_{\kappa_{1}}\log g_{\theta}^{l}(y_{t_{p}}|x_{t_{p}})=y_{t_{p}}^{1}-\Delta_{l}\exp\left(\kappa_{1}+x_{t_{p}}^{1}\right),\\
% & \partial_{\kappa_{2}}\log g_{\theta}^{l}(y_{t_{p}}|x_{t_{p}})=y_{t_{p}}^{2}-\Delta_{l}\exp\left(\kappa_{2}+x_{t_{p}}^{2}\right).
%\end{align*}
%}
The partial derivatives of the log-observation density are all zero except the ones w.r.t.\ $\kappa_1$ and $\kappa_2$, which can be expressed as 
\begin{align}	
\partial_{\kappa_{i}}\log g^l_{\theta} (y_{t_p} | (X_t)_{t_{p-1}\leq t\leq t_{p}}) = 
y_{t_p}^i - \Delta_l\sum_{t: t_{p-1}\leq t\leq t_p}\lambda_i(X^i_t),
\end{align}
for $i = 1,2$.

\end{document}